\documentclass[11pt]{article}

\usepackage{fullpage,amsfonts,amsmath,amsthm,amssymb,mathtools, mathrsfs,graphicx,upgreek}
\usepackage[margin=1in,marginparwidth=0.75in]{geometry}

\usepackage{aliascnt}
\usepackage{hyperref}
\usepackage{cleveref}
\crefname{subsection}{subsection}{subsections}

\let\oldtheorem\newtheorem
\RenewDocumentCommand{\newtheorem}{s m o m O{}}{%
\IfBooleanTF{#1}%
{\oldtheorem{#2}{#4}}%
{\IfNoValueTF{#3}{\oldtheorem{#2}{#4}[#5]}%
{\newaliascnt{#2}{#3}%
\oldtheorem{#2}[#2]{#4}%
\aliascntresetthe{#2}}}}

\usepackage{epsfig}
\usepackage{bm}
\usepackage{comment}

\usepackage{tablefootnote}

\usepackage{enumitem}

\numberwithin{equation}{section}

\usepackage{algorithm}
\usepackage{mathtools}
\usepackage[noend]{algpseudocode}

\newtheorem{theorem}{Theorem}[section]
\newtheorem{proposition}[theorem]{Proposition}
\newtheorem{corollary}[theorem]{Corollary}

\newtheorem{lemma}[theorem]{Lemma}

\theoremstyle{definition}
\newtheorem{definition}{Definition}

\newcommand{\Pois}{\mathrm{Pois}}

\newtheorem{remark}[theorem]{Remark}

\usepackage{bbm}

\usepackage{algpseudocode}
\usepackage{datetime}

\usepackage[english]{babel}
\usepackage[autostyle, english = american]{csquotes}
\MakeOuterQuote{"}

\newcommand{\scr}{\mathcal}
\newcommand{\mb}{\mathbb}
\newcommand{\til}{\widetilde}

\newcommand{\eps}{\varepsilon}

\newcommand{\erew}{r}
\newcommand{\srew}{\text{val}}
\newcommand{\ac}{a}

\newcommand{\OPT}{\textsc{OPT}}

\newcommand{\LPOPT}{\textsc{LP}}

\newcommand{\MLPOPT}{\textsc{LPOPT}}

\newcommand{\poly}{\text{poly}}

\newcommand{\Ber}{\textup{Ber}}

\newcommand{\citet}{\cite}
\newcommand{\citep}{\cite}

\newcommand{\safe}{\mathsf{safe}}

\makeatletter
\def\keywords{\xdef\@thefnmark{}\@footnotetext}
\makeatother

\usepackage[usenames,dvipsnames]{xcolor}

\title{Approximation Algorithms for Action-Reward Query-Commit Matching }

\author{
Mahsa Derakhshan
\thanks{Khoury College of Computer Sciences, Northeastern University, \texttt{m.derakhshan@northeastern.edu}
}
\and
Andisheh Ghasemi
\thanks{Khoury College of Computer Sciences, Northeastern University, \texttt{ghasemi.e@northeastern.edu}}
\and
Calum MacRury
\thanks{School of Industrial and Systems Engineering (ISyE), Georgia Tech,
\texttt{calum.macrury@isye.gatech.edu}}
}

\date{}

\begin{document}

\maketitle

\begin{abstract}
Matching problems under uncertainty arise in applications such as kidney exchange, hiring, and online marketplaces. A decision-maker must sequentially explore potential matches under local exploration constraints, while committing irrevocably to successful matches as they are revealed. The \emph{query-commit matching problem} captures these challenges by modeling edges that succeed independently with known probabilities and must be accepted upon success, subject to vertex patience (time-out) constraints limiting the number of incident queries.

In this work, we introduce the \emph{action-reward query-commit matching problem}, a strict generalization of query-commit matching in which each query selects an action from a known action space, determining both the success probability and the reward of the queried edge. If an edge is queried using a chosen action and succeeds, it is irrevocably added to the matching, and the corresponding reward is obtained; otherwise, the edge is permanently discarded. We study the design of approximation algorithms for this problem on bipartite graphs.

This model captures a broad class of stochastic matching problems, including the \emph{sequential pricing problem}  introduced by Pollner, Roghani, Saberi, and Wajc (EC~2022). For this problem, computing the optimal policy is NP-hard even when one side of the bipartite graph consists of a single vertex. On the positive side, Pollner et al.\ designed a polynomial-time approximation algorithm achieving a ratio of $0.426$ in the one-sided patience setting, which degrades to $0.395$ when both sides have bounded patience.  These guarantees rely on a linear program that cannot be rounded beyond a factor of $0.544$ in the worst case.

In this work, we design computationally efficient algorithms for the action-reward query-commit in one-sided and two-sided patience settings, achieving
approximation ratios of $1-1/e \approx 0.63$ and $\frac{1}{27}\!\left(19-67/e^3\right) \approx 0.58$ respectively. These results improve the state of the art for the sequential pricing problem, surpassing the previous guarantees of $0.426$ and $0.395$. A key ingredient in our approach is a new configuration linear program that enables more effective rounding.

\end{abstract}

\newpage

\section{Introduction}
Matching problems under uncertainty appear in diverse applications, from kidney exchange to hiring and online marketplaces, where a set of feasible matches is known in advance, but the rewards for each match are uncertain. In such settings, a decision-maker may adaptively query information about potential matches while facing natural constraints on how many options can be explored per participant. 
Moreover, each query may require an immediate commitment to accept or reject a match based on its revealed outcome. These constraints give rise to a challenging algorithmic problem: selecting which matches to explore, subject to local query budgets, while making irrevocable decisions in order to achieve high overall reward.

Motivated by this challenge, \citet{Chen} and \citet{BansalGLMNR12} initiated the study of \textit{query-commit matching problem with patience constraints}, where the input is a graph $G=(V,E)$ with edge rewards
and probabilities $(r_e)_{e \in E}$ and $(q_e)_{e \in E}$, respectively. 
The policy must sequentially 
\textit{query} the edges in an order of its choosing, where a query to $e$ reveals whether it succeeds (i.e., can be included in the matching). In this problem, each edge $e$ is guaranteed to \textit{succeed} 
 independently with probability $q_e$.
Moreover, the policy must be \textit{committal}: if it queries $e$ and learns that it succeeds, then it must add $e$ to its current matching, in which case it gains the reward of $r_e$. Finally, each vertex $v \in V$ has an integer $\ell_v \ge 1$ \textit{patience} or \textit{time-out} constraint, indicating the maximum number of
queries which can be performed on its incidence edges.   The goal of the policy is maximize the sum of the rewards of the edges matched in expectation.  This problem can be formulated as an exponential-sized Markov Decision Process (MDP), and so the optimal policy can be described by a 
dynamic program (DP). Due to the impracticality of computing such a DP, much of the literature has focused on attaining approximation algorithms of the optimal policy (see \cite{Chen,Adamczyk11, BansalGLMNR12, Adamczyk15,BavejaBCNSX18,brubach2024offline}, and \Cref{ssec:extended_related} for an extended overview of the literature).\footnote{We note that an alternative benchmark is the expected reward of the optimal offline matching; however, due to the patience constraints, it is not possible to design a policy within a constant factor approximation of it.}

In this work, we introduce a generalization of the query-commit matching problem with patience constraints  in which the policy has a set of \textit{actions} $\scr{A}$ available when querying an edge.
The input is again a graph $G = (V,E)$; however, for each edge $e$ and each action $a \in \scr{A}$, we are given an action-dependent success probability $q_e(a)$ and reward $r_e(a)$.
The policy sequentially queries edges in an order of its choosing. When querying an edge $e$, it additionally selects an action $a \in \scr{A}$.
Querying $e$ with action $a$ \textit{succeeds} independently with probability $q_e(a)$.
If $e$ succeeds via $a$, the edge $e$ is irrevocably added to the current matching and the policy gains reward $r_e(a)$.
If it fails, the edge $e$ is not added to the matching, and it cannot be queried again in future steps.
As before, each vertex $v \in V$ has a patience constraint $\ell_v$, limiting the number of incident queries, and the objective is to maximize the expected reward of the resulting matching.
We refer to this problem as the \textit{action-reward query-commit matching problem}.

A special case of our model is the \emph{sequential pricing problem}, introduced by \citet{DBLP:journals/corr/abs-2205-08667} and motivated by challenges arising in labor markets. The input is a bipartite graph, with vertices on the two sides representing jobs and workers, respectively. Each job is associated with a known reward. The goal of the \emph{platform} is to have workers complete jobs by offering take-it-or-leave-it prices. This interaction is modeled stochastically: if the platform offers a price $\tau \ge 0$ to worker $v$ for completing job $u$, then $v$ accepts with a known probability $p_{u,v}(\tau)$, in which case $u$ and $v$ are matched and leave the system. As in the query-commit model, workers have known patience (time-out) constraints, specifying the maximum number of job offers they are willing to consider before exiting. The platform’s objective is to maximize either its expected \emph{revenue} or its expected \emph{social welfare}. By interpreting prices as actions in our framework (i.e., $\scr{A} = [0,\infty)$) and defining the edge rewards and success probabilities appropriately, our model captures both objectives.

 We are also motivated by a natural extension of the query-commit matching problem where the rewards of the edges are described by arbitrary
real-valued distributions, as  opposed to weighted Bernoulli random variables. This is closely related to the \textit{free-order prophet matching} problem \cite{fu_2021, brubach2024offline, DBLP:journals/corr/abs-2205-08667, macrury_contention_2023, macruryinduction2023}, with the distinctions being that there are patience constraints, and that we benchmark against
the optimal policy, as opposed to the prophet/hindsight optimum. Once again, our action-reward query-commit matching problem subsumes this problem, whether the objective is social welfare or revenue. We provide the details of both reductions in \Cref{sec:reductions}.

For the sequential pricing problem, \citet{DBLP:journals/corr/abs-2205-08667} design a polynomial-time algorithm achieving an approximation ratio of $0.426$ in the \emph{one-sided patience} setting, where only workers have patience constraints. They also consider the \emph{two-sided patience} setting, in which jobs have patience constraints as well; in this case, the approximation ratio drops to $0.395$. Their algorithm and analysis rely on rounding a linear program (LP) inspired by the work of \citet{BansalGLMNR12}. However, due to a classic result of \citet{karp1981maximum}, it is not possible
to design a policy with an expected reward greater than $0.544$-fraction of the LP's value, leaving open the question of whether improved approximation ratios are achievable via alternative LP formulations and rounding techniques.

In this work, we answer this question affirmatively by establishing improved approximation guarantees for the action-reward query-commit matching problem on bipartite graphs. Since our model subsumes the sequential pricing problem of \citet{DBLP:journals/corr/abs-2205-08667}, our results improve the best known approximation ratios to $1-1/e$ and $\frac{1}{27}\!\left(19-\frac{67}{e^3}\right)\approx 0.5801$ for the one-sided and two-sided patience settings, respectively.

Here, the approximation ratio of a policy is defined as the ratio between its expected reward and the expected reward of an optimal policy, i.e., the best policy without computational constraints, whose value on
$G$ we denote by $\OPT(G)$. Our main result is formally stated as follows.

\begin{theorem}[Main Theorem]\label{thm:main_theorem}
Given a bipartite graph $G=(U,V,E)$ with action set $\scr{A}$, let $\beta = 1-1/e$ if $\ell_u \in \{1, \infty\}$ 
for all $u \in U$, and $\beta = \frac{1}{27} (19 - \frac{67}{e^3})$, otherwise.
For any arbitrarily small $\epsilon \in (0,  1)$, there exists a policy with expected reward of
    $$\beta \cdot (1 -\epsilon) \OPT(G).$$
Moreover, the policy can be computed in time $ f(1/\epsilon)\cdot \poly(|G|, |\scr{A}|)$ where $f$ is some function and $\poly$ is a polynomial in $|G|$ and $|\scr{A}|$.
\end{theorem}

Most closely related to our work are \citep{borodin2025online} and \citep{brubach2024offline}, who gave polynomial-time $(1 - 1/e)$-approximations for the (standard) query-commit matching problem on bipartite graphs with one-sided patience, and \textit{unit patience values} (e.g., $\ell_u =1$ for all $u \in U$), respectively. 
Our $(1-1/e)$-approximation ratio can thus be seen as a generalization of both of these of results to our action-reward framework.

Moving to two-sided patience constraints,
even in the (standard) query-commit matching problem, the previous best known approximation ratio follows from the aforementioned $0.426$-approximation
due to \citep{DBLP:journals/corr/abs-2205-08667}. Thus, our $\frac{1}{27} (19 - \frac{67}{e^3})$-approximation where $\frac{1}{27} (19 - \frac{67}{e^3}) \approx 0.5801$, simultaneously improves on the approximation
ratios of past works, while also applying to a more general setting.

\subsection{Technical Overview} \label{sec:technical_overview}
Our approach to proving \Cref{thm:main_theorem} involves three high level steps. First, we design a new configuration LP
to relax/upper bound the optimal policy's expected reward. Second, we design a variant of an online contention resolution scheme for patience constraints,
and show how to use it to round our LP into a policy. Finally, we compute an approximately optimal solution to our LP
via a reduction to approximately solving the ``star graph'' case of our problem, where one partite set has a single vertex. Each of these steps is self-contained,
and taken together quickly imply \Cref{thm:main_theorem}, as we explain in \Cref{sec:final_proof}. Below, we provide an overview of the technical challenges faced
in each step, and how these relate to the past approaches in the literature.

\paragraph{Our Configuration LP (\Cref{sec:config_LP}).}
The majority of the past works in the literature use an LP which accounts solely for the marginal edge probabilities for which the optimal policy queries the various edges. While this LP was originally introduced by \citet{BansalGLMNR12} for the (standard) query-commit matching problem, \citet{DBLP:journals/corr/abs-2205-08667} used an analogous LP for their sequential pricing problem. It extends to our action-reward problem, and so we present it in this setting.
For each $e \in E$ and $\ac \in \scr{A}$, the LP has a decision variable $z_{e}(\ac)$ which is the probability the benchmark queries $e$ via $a$. 
This leads to the following LP, where we recall that $q_{e}(\ac)$ denotes the probability that $e$ succeeds via $a$,
$r_{e}(\ac)$ denotes the reward attained when this occurs, and $\partial(s)$ denotes the edges incident to vertex $s \in U \cup V$:
\begin{align}\label{LP:bansal}
	\tag{LP-M}
	&\text{maximize} &  \sum_{e \in E, \ac \in \scr{A}} r_{e}(\ac)  q_{e}(\ac) z_{e}(\ac)\\
	&\text{subject to} & \sum_{e \in \partial(s), \ac \in \scr{A}} q_{e}(\ac) z_{e}(\ac)  \leq 1 && \forall s \in U \cup V  \label{eqn:bansal_matching}\\
        && \sum_{e \in \partial(s), \ac \in \scr{A}} z_{e}(\ac) \leq \ell_s && \forall s \in U\cup V \label{eqn:bansal_patience}\\
        && \sum_{ \ac \in \scr{A}} z_{e}(\ac) \le 1 && \forall e \in E, \label{eqn:bansal_query} \\
	&&z_{e}(\ac) \ge 0 && \forall e \in E, \ac \in \scr{A}
\end{align}

After computing an optimal solution $(z_{e}(\ac))_{e \in E, \ac \in \scr{A}}$ to \ref{LP:bansal}, one still needs to efficiently round it into a policy. While earlier
works \citep{BansalGLMNR12,Adamczyk15,BrubachSSX20} used the  GKSP algorithm of \citep{GandhiGKSP06}, the last number of works \citep{BavejaBCNSX18,brubach2024offline,DBLP:journals/corr/abs-2205-08667}
 -- including the previous best state-of-the-art \citep{DBLP:journals/corr/abs-2205-08667} --
have all used independent sampling combined with a \textit{uniform at random} (u.a.r)
processing order of the edges. The simplest form
of this algorithmic template is due to \citet{BavejaBCNSX18},
where each time $e \in E$ is processed, if $u$ and $v$ are unmatched and have remaining patience, then $e$ is queried via action $\ac$ independently with probability (w.p.) $z_{e}(\ac)$
(this is well-defined, due to \eqref{eqn:bansal_query}). \citet{BavejaBCNSX18} showed
that any edge $e$ and action $\ac$ is queried w.p. $0.310 z_{e}(\ac)$, which implies their approximation ratio of $0.310$. The works of
\cite{brubach2024offline,DBLP:journals/corr/abs-2205-08667} further tuned this algorithm by introducing additional independent random bits $(B_e)_{e \in E}$, with the requirement that $e$ is queried only if $B_e =1$. These bits help to increase the \textit{minimum} probability that any edge $e$ is queried via any action $a$, and so they are able to get an improved guarantee of $\alpha z_{e}(\ac)$ for $\alpha = 0.382$ in \cite{brubach2024offline} and $\alpha = 0.395$ in \cite{DBLP:journals/corr/abs-2205-08667} \footnote{The policies of \cite{brubach2024offline, DBLP:journals/corr/abs-2205-08667} additionally require the input to have minimum patience $2$ for their approximation ratios to hold. For instance, the policy of
\cite{brubach2024offline} attains at most $1/3$ for an input with vertices with patience $1$.} . In addition
to the aforementioned $0.544$ hardness result against comparisons to \eqref{LP:bansal},
there is also a hardness result against algorithms which process the edges in a u.a.r order.
The $1/2$-impossibility result of \cite{macrury_contention_2023} implies that no such
algorithm can attain an approximation ratio better than $1/2$ against \ref{LP:bansal}.
Thus, there is limited room for attaining improved approximation ratios in this way.

While a tightening of \ref{LP:bansal} was introduced by \cite{costello2012matching, Gamlath2019} that is more easily rounded, its
definition is restricted to inputs with unlimited patience, and the (standard) query-commit matching problem (i.e., $|\scr{A}| =1$). In this restricted setting, \cite{Gamlath2019} added additional LP constraints
to \ref{LP:bansal}, and derived a $(1-1/e)$-approximation. Unfortunately, it's unclear how to extend this LP
to arbitrary patience constraints, let alone an action space with $|\scr{A}| \ge 2$. An alternative approach to handling the patience constraints was taken by \cite{borodin2025online,brubach2025star}, who presented a \textit{configuration} LP 
for the case when $|\scr{A}| =1$, where if $\bm{e}=(e_1, \ldots ,e_k)$ is a sequence of distinct edges incident to $v \in V$ with $k \le \ell_v$, then the LP has a decision variable $z_{v}(\bm{e})$ for the probability
the benchmark queries the edges of $\bm{e}$ in order. 

Our approach is to first generalize the configuration LP of \cite{borodin2025online,brubach2025star} (formally stated in \ref{LP:social_welfare} of \Cref{sec:config_LP})
to our action-reward problem. We introduce a decision variable for each policy associated with each vertex $v \in V$. That is,
if $\bm{e}=(e_1, \ldots ,e_k)$ is a sequence of distinct edges incident to $v \in V$ with $k \le \ell_v$,
and $\bm{\ac} = (\ac_1, \ldots , \ac_k) \in \scr{A}^k$ is a sequence of actions, then $z_{v}(\bm{e}, \bm{\ac})$ is the probability the benchmark
queries $e_i$ via $a_i$ for $i=1, \ldots ,k$ in order,
outputting the first $e_i$ which succeeds via $a_i$ (if any).
Setting $\srew(\bm{e}, \bm{\ac})$ as the expected reward
for $v$ in this case,
our LP's objective is 
\begin{equation} \label{eqn:overview_objective}
\sum_{v \in V} \sum_{(\bm{e},\bm{\ac}) \in \scr{C}_v} \srew(\bm{e}, \bm{\ac}) \cdot z_{v}(\bm{e}, \bm{\ac}),
\end{equation}
where the second sum is over $\scr{C}_v$, the set of all policies of length at most $\ell_v$ for $v$.
Next, for each $v \in V$, we add the constraint 

\begin{equation} \label{eqn:overview_distribution}
    \sum_{(\bm{e}, \bm{\ac})} z_v(\bm{e},\bm{\ac}) \le 1,
\end{equation}
which ensures each $v$ selects at most one policy $(\bm{e}, \bm{\ac})$. 
In order to motivate the remaining constraints, 
imagine drawing $(\bm{e}, \bm{a}) \in \scr{C}_v$ w.p. $z_{v}(\bm{e},\bm{a})$, and then
querying $e_i$ via $a_i$ in order, stopping if $e_i$ succeeds via $a_i$.
In this case,  for any $e \in E$ incident to $v$ and $a \in \scr{A}$, we can write out the probability
that $e$ is queried via $a$ as follows:
\begin{equation} \label{eqn:edge_variable_overview}
\til{z}_{e}(\ac):= \sum_{\substack{i \in [\ell_v], (\bm{e},\bm{\ac}) \in \scr{C}_v: \\ e_i = e, \ac_i =\ac}} \prod_{j < i} (1 - q_{e_j}(\ac_j))\cdot z_{v}( \bm{e},\bm{\ac}).
\end{equation}
By applying \eqref{eqn:edge_variable_overview} over all the vertices of $V$, we add two constraints for each vertex $u \in U$, which ensure that \textit{in expectation},
$u$ is matched at most once, and has at most $\ell_u$ incident edges queried: 
\begin{align}
\text{$\sum_{e \in \partial(u)} \sum_{a \in \scr{A}} q_{e}(a) \til{z}_{e}(a)  \leq 1,$ and  
    $\sum_{e \in \partial(u)} \sum_{a \in \scr{A}} \til{z}_{e}(a)  \leq \ell_u$. \label{eqn:overview_queried}}
\end{align}
We will expand upon the exact use of \eqref{eqn:overview_queried}
shortly, but we mention for now that they serve an analogous purpose as \eqref{eqn:bansal_matching} and \eqref{eqn:bansal_patience} do for \ref{LP:bansal}, and the algorithms which round it.

In \Cref{thm:LP_relaxation} of \Cref{sec:config_LP}, we prove that \ref{LP:social_welfare} relaxes our problem's benchmark. While \ref{LP:social_welfare} generalizes the configuration LP's of \cite{borodin2025online,brubach2025star}, it has two key differences.
First, it incorporates an arbitrary action space into its objective \eqref{eqn:overview_objective}. Second, it allows for arbitrary patience constraints on the vertices of $U$.
In contrast, the LP's of \cite{borodin2025online,brubach2025star} are defined only for the standard query-commit problem (i.e., $|\scr{A}| =1$),
and assume that the vertices of $U$ have unlimited patience (i.e., $\ell_u = \infty$ for all $u \in U$). 

\paragraph{Rounding our LP (\Cref{sec:LP_round}).}
We next explain how to round a solution to \ref{LP:social_welfare}
into a policy.
Assume for now that we have an optimal solution
to \ref{LP:social_welfare}, and first consider a natural \textit{greedy}
policy, which is well-defined due to \eqref{eqn:overview_distribution}:
\begin{enumerate}
    \item Draw a u.a.r. arrival order of $V$.
    \item For each arriving $v \in V$, draw a policy $(\bm{e}, \bm{\ac})$ for $v$ independently w.p. $z_{v}(\bm{e}, \bm{\ac})$
    \begin{itemize}
        \item Then, when processing edge $e_i = (u_i,v_i)$, suppose $u_i$ has remaining patience and is currently unmatched.
        In this case, query $e_i$ via action $\ac_i$ and match $e_i$ if $e_i$ via $a_i$ succeeds.
    \end{itemize} 
\end{enumerate}
This policy queries each edge $(u,v) \in E$ and $a \in \scr{A}$ w.p. $\beta_0 \cdot \til{z}_{u,v}(a)$,
where $\beta_0 = 1/2$ if $\ell_u \in \{1, \infty\}$, and $\beta_0 = (4 -e)/e \approx 0.4715$ otherwise \footnote{To get the guarantee of $\beta_0$, technically the greedy policy must use simulated bits drawn from $(q_{e}(a))_{a \in \scr{A}, e \in E}$ to make ``fake queries'' when necessary, however we defer this discussion to \Cref{sec:LP_round}.}. 
Thus, since we can equivalently write our LP's objective using \eqref{eqn:edge_variable_overview} as
$$
    \sum_{v \in V} \sum_{(\bm{e},\bm{\ac}) \in \scr{C}_v} \srew(\bm{e}, \bm{\ac}) \cdot z_{v}(\bm{e}, \bm{\ac}) = \sum_{e \in E}  \sum_{\ac \in \scr{A}} \erew_{e}(\ac) q_{e}(\ac) \til{z}_{e}(\ac),
$$
the greedy policy gets an approximation of $\beta_0$, and so it improves upon
\cite{DBLP:journals/corr/abs-2205-08667}. However, to get the guarantee of \Cref{thm:main_theorem}, we must further modify our policy. 
Our idea is to recognize that as the greedy policy executes, from the perspective
of a fixed  vertex $u$, the constraints \eqref{eqn:overview_queried} ensure that \textit{in
expectation}, at most $\ell_u$ edges of $u$ are queried, and at most one of the queried edges will succeed via some action. When $\ell_u \in \{1, \infty\}$, one of these constraints is subsumed, and we can execute a \textit{(single-item) random order
contention resolution (RCRS)} to decide which edge/action pair of $u$ to query (opposed to just being greedy).
Due to a known result of \cite{Lee2018,brubach2024offline}, this quickly implies the $1-1/e$ guarantee
of \Cref{thm:main_theorem}. On the other hand, if $2 \le \ell_u < \infty$ for some $u \in U$, then
we must design our own randomized rounding tool to determine which edge/action pair of $u$ to query. This is formalized in the
\textit{patience random-order contention resolution (P-RCRS) problem} of \Cref{sec:LP_round},
and in \Cref{thm:trs}, we state the rounding guarantee which we use to derive 
the $\frac{1}{27} (19 - \frac{67}{e^3}) \approx 0.5801$ result of \Cref{thm:main_theorem}.

We note that our P-RCRS problem is effectively the same rounding problem considered
in \cite{brubach2024offline, DBLP:journals/corr/abs-2205-08667}, except that we only have
to consider it for the special case of a fixed vertex $u \in U$ and its incident edges (i.e., a star graph). In contrast,  \cite{brubach2024offline, DBLP:journals/corr/abs-2205-08667} 
study this problem for an arbitrary graph whose edges are processed in random order. This explains
why we improve on these past works.
That being said, while our P-RCRS for this problem is based on the ideas in \cite{brubach2024offline, DBLP:journals/corr/abs-2205-08667},
there is a key technical step that was overlooked in \cite{brubach2024offline} involving the worst-case input of their rounding algorithm \footnote{Since \citet{DBLP:journals/corr/abs-2205-08667} builds upon the algorithm and analysis of \citep{brubach2024offline}, they were also not aware of this missing technical step.}. We discuss this is more detail in \Cref{sec:main_section_srcrs}, and mention that the updated arXiv version of \cite{brubach2024offline} now cites the ``exchange argument'' of \Cref{lem:exchange_argument} from \Cref{sec:patience_2_exchange} to correct their algorithm's analysis.

\paragraph{Solving our LP (\Cref{sec:solve_LP}).}

Finally, in order to solve our LP efficiently, we first take its dual.
While designing a separation oracle for an LP's dual is the standard way to solve an LP
with exponentially many variables, and was in fact the approach taken in \cite{borodin2025online, brubach2025star}, this is not possible for our LP's dual. This is because the separation problem reduces to  designing an optimal policy for the special case of our action-reward problem when $V$ contains a single vertex $v$. We call
this the \textit{star graph action-reward problem}, and note
that it is NP-hard if $|\scr{A}| \ge 2$, even when $\ell_v = \infty$. This is because we can set the edge rewards/probabilities so as to encode the (single-item) free-order prophet problem, of which computing an optimal policy is known to be NP-hard for random variables
with support at least $3$ \cite{agrawal2020}. Fortunately, there exists an efficient polynomial time approximation scheme (EPTAS) for this problem \cite{segev2021}. In \Cref{sec:santa_claus_approximation}, we leverage the techniques of \cite{segev2021} to design an EPTAS for our star graph problem. That is, for any $\eps > 0$,
we design a policy with runtime $f(\eps) \cdot \poly(|U|, |\scr{A}|)$  whose expected reward is at least $(1-\eps)$-fraction of the optimal policy's expected reward (here $f(\eps)$ is some arbitrary function of $\eps$,
and $\poly(|U|, |\scr{A}|)$ is a polynomial in $|U|$ and $|\scr{A}|$). As we explain
in \Cref{sec:solve_LP}, this allows us to get a $(1-\eps)$-approximate solution to \ref{LP:social_welfare},
which together with the aforementioned rounding techniques of \Cref{sec:LP_round}, implies \Cref{thm:main_theorem}. More, since our setting encodes the free-order prophet problem, as a corollary we get a generalization of the EPTAS of \cite{segev2021}. This generalization works
for either the social welfare or revenue objective, and also handles the extension of the problem
when the policy can only make a limited number of queries.

\subsection{Further Related Work} \label{ssec:extended_related}
For the query-commit matching problem with patience constraints and uniform edge rewards (i.e., unweighted graphs), \citet{Chen} showed that the greedy algorithm attains an approximation ratio of $1/4$
and its analysis was later improved to $1/2$ by \citet{Adamczyk11}. Afterwards, \citet{BansalGLMNR12} introduced the version
of this problem with edge rewards, and attained an approximation ratio of $1/4$ via a linear programming (LP) based policy. A sequence of
papers later improved upon this result, \cite{Adamczyk15,BavejaBCNSX18,brubach2024offline}, with the state of the art of $0.395$ being due to \citet{DBLP:journals/corr/abs-2205-08667}.

We also mention some related results for the (standard) query-commit matching problem \textit{without} patience values,
in which case the benchmark taken is the optimal offline matching.
For bipartite graphs, \citet{Gamlath2019} gave an algorithm with a guarantee of $1 - 1/e$. This was slightly improved by ~\citet{DBLP:conf/soda/Derakhshan023}, further refined by \citet{chen_2024} to $0.641$, and most recently improved to $0.659$ by \citet{huang2025edgeweightedmatchingdark}. For general graphs, \citet{fu_2021} derived a guarantee of $8/15 \approx 0.533$, and this was slightly improved
by \citet{macruryinduction2023}. We note that \citep{chen_2024, fu_2021, macruryinduction2023} allow the edges to have rewards drawn independently from
arbitrary distributions (i.e., the prophet matching problem), and \citep{huang2025edgeweightedmatchingdark} applies when the edge probabilities are correlated (i.e., the oblivious bipartite matching problem).

\subsection{Preliminaries}
We formalize notation for our problem, recapping some concepts from the introduction.

Let $G=(U,V,E)$ be a bipartite graph.  An edge $e=(u,v)$ is said to be \textit{incident} to vertices $u$ and $v$, and $v$ is said to be a \textit{neighbour} of $u$ (and vice versa).
For any vertex $s\in U \cup V$, let $N(s)\subseteq V$ denote its neighbours, and let $\partial(s)\subseteq E$ denote the set of edges incident to $s$.
A \textit{matching} $M$ of $G$ is a subset of edges no two of which are incident to the same vertex, i.e.\ satisfying $|M\cap\partial(s)|\le 1$ for all $s\in U \cup V$. We say that $s \in U \cup V$ is \textit{matched} by $M$, provided $|M \cap \partial(s)| =1$.

Suppose that associated with $G$ we have an action \textit{action space} $\scr{A}$, \textit{edge rewards} $\bm{r} =(r_{e}(a))_{e \in E, a \in \scr{A}}$ and \textit{edge probabilities} $\bm{q} = (q_{e}(a))_{e \in E, a \in \scr{A}}$. More, assume that each $s \in U \cup V$ has \textit{patience} $\ell_s \in \mb{N}_0$. A \textit{policy} for the \textit{action/reward matching problem} is given $G$, $\bm{r}, \bm{q}$ and $(\ell_s)_{s \in U \cup V}$ as its input, and its output is a matching  $\scr{M}$ of $G$. The policy must be \textit{committal}, which means its matching is constructed in the following way: In each step $t \ge 1$, it \textit{adaptively} chooses
 an edge $e_t \in E$, as well as action $a_t \in \scr{A}$. Then, it \textit{queries} $e_t$ via action $a_t \in \scr{A}$, at which point an independent Bernoulli $Q_{e_t}( a_t) \sim \Ber(q_{e_t}(a_t))$ is drawn. If $Q_{e_t}( a_t) =1$, then $(e_t,a_t)$ is said to \textit{succeed},
 and the policy must add $e_t$ to its current matching, in which
 case $e_t$ contributes a reward of $r_{e_t}(a_t)$. Otherwise, if $Q_{e_t}( a_t) = 0$, then $(e_t,a_t)$ is said to \textit{fail}, and $e_t$ cannot be added to the matching. In addition, the policy must follow the following additional rules:

\begin{enumerate}
    \item For each $e \in E$, $e$ is queried by at most one action.
    \item For each $s \in U \cup V$, there are at most $\ell_s$ edges in $\partial(s)$ which were queried via actions.
\end{enumerate}
If at step $t$, there have been less than $\ell_s$ queries to $\partial(s)$, then we say that $s$ has \textit{remaining patience} (at step $t$).

After at most $|E|$ steps, the matching $\scr{M}$ output by the policy is finalized,
and the \textit{reward} of the policy is
$
   \sum_{e \in E} \sum_{a \in \scr{A}} r_{e}(a) \cdot Q_{e}(a) \cdot \bm{1}_{\{\text{$e$ queried via $a$}\}}.
$

The goal of the problem is to design a policy whose \textit{expected} reward is as large
as possible, where the expectation is over $(Q_e(a))_{e \in E, a \in \scr{A}}$, as well as any randomized decisions made by the policy. We 
write as this
\begin{equation} \label{eqn:expected_reward_definition}
    \sum_{e \in E} \sum_{a \in \scr{A}} r_{e}(a) \mb{P}[\{Q_{e}(a) =1\} \cap \{\text{$e$ queried via $a$}\}] = \sum_{e \in E} \sum_{a \in \scr{A}} r_{e}(a) q_{e}(a) \mb{P}[\text{$e$ queried via $a$}],
\end{equation}
where the equality follows since the policy must decide whether to query $e$ via $a$, prior to observing $Q_{e}(a)$.
We denote the expected reward of the (computationally unconstrained) \textit{optimal policy} on $G$ by $\text{OPT}(G)$.
By extending $r_{e}(a) :=0$ and $q_{e}(a) :=0$ for each $e \in U \times V \setminus E$ and $a \in \scr{A}$, we hereby assume without loss that $E = U \times V$.  

Finally, we review some standard notation. For an arbitrary set $S$ and integer $k \ge 1$, let $S^{k}$ be the set of all vectors of length $k$ with entries in $S$, and $|\bm{s}| =k$ denote the length of $\bm{s} \in S^k$.  We also set $[k] := \{1, \ldots ,k\}$ for convenience.

\section{Our Configuration LP} \label{sec:config_LP}

In this section, we introduce our configuration LP for $G=(U,V,E)$
with actions $\scr{A}$, edge rewards $\bm{r} =(r_{e}(\ac))_{e \in E, \ac \in \scr{A}}$ and edge probabilities $\bm{q} =(q_{e}(\ac))_{e \in E, \ac \in \scr{A}}$. 
Our LP introduces a decision variable for each policy associated with each vertex $v \in V$.
Specifically, suppose that $\bm{e}=(e_1, \ldots ,e_k)$ is permutation of a subset of $\partial(v)$ of size $k \le \ell_v$,
and $\bm{\ac}= (\ac_1, \ldots , \ac_k) \in \scr{A}^{k}$.  
Let $\scr{C}_v$ be denote the set of all such pairs $(\bm{e}, \bm{\ac})$. Observe that $(\bm{e}, \bm{a}) \in \scr{C}_v$ corresponds to the policy which queries the edges $e_1, \ldots ,e_k$ incident to $v$ via actions $\ac_1, \ldots , \ac_k$, outputting the first edge/action pair $(e_i,a_i)$ which succeeds (if any).
We define $\srew_{\bm{r}, \bm{q}}(\bm{e}, \bm{\ac})$ to be the expected reward of this policy,
and observe that
\begin{equation} \label{eqn:reward_of_vertex_policy}
    \srew_{\bm{r}, \bm{q}}(\bm{e}, \bm{\ac}) = \sum_{i=1}^{k} \erew_{e_i}(\ac_i) \cdot q_{e_i}( \ac_i) \prod_{j <i} (1 - q_{e_j}(\ac_j)).
\end{equation}
We will sometimes drop the dependence on $\bm{r}$ and $\bm{q}$ in \eqref{eqn:reward_of_vertex_policy}
as the edge rewards and probabilities will be implicitly clear. With this notation, we now state our LP, where we introduce a decision variable $z_{v}(\bm{e}, \bm{\ac})$
for each $v \in V$ and $(\bm{e}, \bm{\ac}) \in \scr{C}_v$:

\begin{align}\label{LP:social_welfare}
	\tag{LP-C}
	&\text{maximize} &  \sum_{v \in V} \sum_{(\bm{e},\bm{\ac}) \in \scr{C}_v } \srew(\bm{e}, \bm{\ac}) \cdot z_{v}(\bm{e}, \bm{\ac}) \\
	&\text{subject to} & \sum_{v \in V} \sum_{\substack{i \in [\ell_v], (\bm{e}, \bm{\ac}) \in \scr{C}_v: \\ e_i = (u,v)}} 
	q_{(u,v)}(\ac_i) \cdot \prod_{j < i} (1 - q_{e_j}(\ac_j)) \cdot z_v( \bm{e}, \bm{\ac})  \leq 1 && \forall u \in U  \label{eqn:sw_matching}\\
        && \sum_{v \in V} \sum_{\substack{i \in [\ell_v], (\bm{e}, \bm{\ac}) \in \scr{C}_v: \\ e_i = (u,v)}} 
	 \prod_{j < i} (1 - q_{e_j}(\ac_j)) \cdot z_v( \bm{e}, \bm{\ac})  \leq \ell_u && \forall u \in U \label{eqn:sw_offline_patience}\\
	&& \sum_{ (\bm{e}, \bm{\ac}) \in \scr{C}_v} z_v(\bm{e},\bm{\ac}) \le 1 && \forall v \in V,  \label{eqn:sw_distribution} \\
	&&z_v( \bm{e}, \bm{\ac}) \ge 0 && \forall v \in V, (\bm{e}, \bm{\ac}) \in \scr{C}_v
\end{align}
Consider an optimal policy for $G$  (i.e, its expected reward is $\text{OPT}(G)$). If
we interpret $z_{v}(\bm{e},\bm{\ac})$ as the probability this policy queries $\bm{e} \in \scr{C}_v$ via actions $\bm{\ac}$, the objective provides an upper bound on the policy's expected reward, and \eqref{eqn:sw_distribution} ensures that each $v \in V$ chooses at most one policy for $\partial(v)$. Next, \eqref{eqn:sw_matching} corresponds to each $u \in U$ being matched at most once in expectation, and \eqref{eqn:sw_offline_patience} corresponds to each $u \in U$ receiving at most $\ell_u$
queries in expectation. 

\begin{theorem}[Proof in \Cref{pf:thm:LP_relaxation}]\label{thm:LP_relaxation}
Let $\LPOPT(G)$ denote the optimal value of \eqref{LP:social_welfare}. Then, $\OPT(G) \leq \LPOPT(G)$.
\end{theorem}

While the above decision variable interpretation is roughly correct, to prove \Cref{thm:LP_relaxation}, we in fact consider a \textit{relaxed} action-reward problem where the constraints on each $u \in U$ need only hold on average. That is,  \textit{in expectation}, (1) $u$ is matched at most once, and (2) $u$ has at most $\ell_u$ queries to its incident edges. We prove that for this relaxed problem, there exists an optimal \textit{vertex-iterative} policy,
which deals completely with all edges incident to a vertex before moving to the next. Policies of this
form are easily seen to be upper bounded by \ref{LP:social_welfare}, and so since the optimal reward of the relaxed problem is no smaller than our original problem (by definition), we establish \Cref{thm:LP_relaxation}.

\section{Rounding a Solution to \ref{LP:social_welfare}} \label{sec:LP_round}
Suppose that we are presented an arbitrary solution $(z_{v}(\bm{e},\bm{\ac}))_{v \in V, (\bm{e},\bm{\ac}) \in \scr{C}_v}$ to \ref{LP:social_welfare} for $G$. Our goal is to use  $(z_{v}(\bm{e},\bm{\ac}))_{v \in V, (\bm{e},\bm{\ac}) \in \scr{C}_v}$ to design a policy
whose expected reward is at least
\begin{equation} \label{eqn:main_approximation_goal}
\beta \cdot \sum_{v \in V} \sum_{(\bm{e},\bm{\ac}) \in \scr{C}_v } \srew(\bm{e}, \bm{\ac})  z_{v}(\bm{e}, \bm{\ac}),
\end{equation}
where $\beta := 1-1/e$ if $\ell_u \in \{1, \infty\}$ for all $u \in U$,
and $\beta := {27} (19 - \frac{67}{e^3})$ otherwise. 
Afterwards, in \Cref{sec:solve_LP},
we shall explain how to efficiently compute a $(1-\eps)$-approximation of an optimal solution
to \ref{LP:social_welfare} for any $\eps > 0$, which combined with \Cref{thm:LP_relaxation} will imply \Cref{thm:main_theorem}. 

In order to understand our approach to proving \eqref{eqn:main_approximation_goal}, it will be useful to again consider a \textit{relaxed policy}, as was done in the proof of \Cref{thm:LP_relaxation}.
A relaxed policy respects the patience constraints of the vertices of $V$,
and outputs a \textit{one-sided matching} $\scr{N}$ of $V$, i.e., $\scr{N} \subseteq E$ such that each $v \in V$ satisfies $|\partial(v) \cap \scr{N}| \le 1$, yet it can query $\partial(u)$ more than $\ell_u$ times or have
$|\scr{N} \cap \partial(u)| \ge 2$, for some $u \in U$. It still must follow the remaining rules of a (proper) policy (i.e.,
it is committal, and each edge is queried via at most one action). Consider now the following relaxed policy which is defined using the solution
$(z_{v}(\bm{e},\bm{\ac}))_{v \in V, (\bm{e},\bm{\ac}) \in \scr{C}_v}$:
\begin{algorithm}[H]
	\caption{Relaxed Rounding of \ref{LP:social_welfare}} 
	\label{alg:social_welfare_relaxed}
	\begin{algorithmic}[1]
		\Require a solution $(z_{v}(\bm{e},\bm{\ac}))_{v \in V, (\bm{e},\bm{\ac}) \in \scr{C}_v}$ to \ref{LP:social_welfare} for $G$.
		\Ensure a one-sided matching $\scr{N}$ of $G$.
                   \State For each $e \in E$ and $a \in \scr{A}$, set $Z_e(a) =0$.
                     \State Independently draw a u.a.r. order $\pi$ on $V$. 
		\For{$v \in V$ in increasing order of $\pi$} 
            \State Draw $(\bm{e}, \bm{\ac}) \in \scr{C}_v$ independently w.p. $z_{v}(\bm{e}, \bm{\ac})$,
            where $\bm{e} = ( e_1, \ldots , e_k)$ and $\bm{\ac} = (\ac_1, \ldots, \ac_k)$.
            \For{$j=1, \ldots, k$}
                    \State Set $Z_{e_j}(a_j) =1$.       \Comment{$e_j$ with $a_j$ is suggested.} \label{line:propose}
                        \State Query $e_j$ via $a_j$.
            \If{$Q_{e_j}(a_j) =1$}
            
            \State $\scr{N} \leftarrow \scr{N} \cup \{e_j\}$.
            \State Exit the internal for loop.
            \EndIf
            \EndFor
		\EndFor
		\State \Return $\scr{N}$.
	\end{algorithmic}
\end{algorithm}

Clearly, \Cref{alg:social_welfare_relaxed} respects the patience constraints of $V$, and each $v \in V$ is assigned at most one edge, so $\scr{N}$ is a one-sided matching.
Now, when line \ref{line:propose}. is reached, we say that edge $e_j$ with action $a_j$ is \textit{suggested}, and denote
the indicator random variable for this event by $Z_{e_j}(a_j)$. Of course $Z_{e_j}(a_j) =1$ if and only if $e_j$ is queried via $a_j$,
however this definition will be useful when we design a (proper) policy. Observe for now that for each $e=(u,v) \in E$
and $a \in \scr{A}$, 
\begin{equation}\label{eqn:edge_variable}
	\til{z}_{e}(\ac):= \mb{P}[Z_{e}(a) = 1] = \mb{P}[\text{$e$ is queried via $a$}]=  \sum_{\substack{i \in [\ell_v], (\bm{e},\bm{\ac}) \in \scr{C}_v: \\ e_i = e, \ac_i =\ac}} \prod_{j < i} (1 - q_{e_j}(\ac_j))\cdot z_{v}( \bm{e},\bm{\ac}). 
\end{equation}
Using the definition of $\til{z}_e(a)$ and \eqref{eqn:edge_variable},
notice that we can write the expected reward of \Cref{alg:social_welfare_relaxed} in two equivalent ways:
\begin{equation} \label{eqn:reward_equivalent}
    \sum_{e \in E}  \sum_{\ac \in \scr{A}} \erew_{e}(\ac) q_{e}(\ac) \til{z}_{e}(\ac) = \sum_{v \in V} \sum_{(\bm{e},\bm{\ac}) \in \scr{C}_v } \srew(\bm{e}, \bm{\ac})  z_{v}(\bm{e}, \bm{\ac}).
\end{equation}
On the other hand, due to the constraints of \eqref{LP:social_welfare}, $(\til{z}_e(a))_{e \in E, a \in \scr{A}}$ must also satisfy some linear constraints. The first
says that the relaxed policy queries any edge via at most one action. The second and third say 
that for any $u \in U$, \textit{in expectation}, the relaxed policy queries at most $\ell_u$ edges of $\partial(u)$,
and assigns at most one edge to $u$.
\begin{lemma}[Proof in \Cref{pf:lem:edge_variable}] \label{lem:edge_variable}
 For any solution $(z_{v}(\bm{e},\bm{\ac}))_{v \in V, (\bm{e},\bm{\ac}) \in \scr{C}_v}$ to \ref{LP:social_welfare} for $G$, it holds that
\begin{itemize}
        \item For each $(u,v) \in E$, $\sum_{a \in \scr{A}} \til{z}_{u,v}(a) \le 1$.
    \item For each $u \in U$, $\sum_{v \in V} \sum_{\ac \in \scr{A}}  \til{z}_{u,v}(\ac) \le \ell_u$, and $\sum_{v \in V} \sum_{\ac \in \scr{A}} q_{u,v}(\ac) \til{z}_{u,v}(\ac) \le 1$.
\end{itemize}
\end{lemma}

We shall now modify the decisions of the relaxed policy so that it respects patience constraints of $U$,
and outputs a (proper) matching of $G$, while following the decisions of the relaxed policy as best as possible.
Roughly speaking, we will query an edge $e$ via $a$ \textit{only if} $Z_e(a) =1$, but not all the time, due to the patience and matching constraints on $U$.
To formalize this, let us first consider an intermediate problem, where for
a fixed vertex $u_0 \in U$, we respect the patience constraints on $\partial(u_0)$, and ensure
$u_0$ is matched to at most one edge (for the vertices of $U \setminus \{u_0\}$, there are no hard restrictions).

Consider the execution of our relaxed policy from the perspective of $u_0$.  In each step
$t=1, \ldots, |V|$, the vertex $v \in V$ with $\pi(v) = t$ is processed. At this point, suppose that $(u_0,v)$ with $a$ is suggested (i.e., $Z_{u_0,v}(\ac) =1$) -- otherwise, we will never
match $u_0$ to $v$. In our intermediate problem, we can only query $(u_0,v)$ via $a$ if $u_0$ is currently unmatched and has remaining patience.
If we make this decision \textit{greedily}, i.e., query $(u_0,v)$ via $a$ whenever possible, one can prove 
that
\begin{equation} \label{eqn:threshold_selection_guarantee_informal}
    \Pr[\text{$(u_0,v)$ is queried via $a$} \mid Z_{u_0,v} =1] \ge \beta_0,
\end{equation}
for each $v \in V$ and $a \in \scr{A}$, where $\beta_0 := 1/2$ if $\ell_{u_0} \in \{1,\infty\}$, and $\beta_0 := (4 -e)/e \approx 0.4715$ otherwise. Thus, since $\mb{P}[Z_{u_0,v}(a) =1] = \til{z}_{u_0,v}(a)$ by \eqref{eqn:edge_variable}, the expected reward of the match for $u_0$ is
\begin{equation} \label{eqn:single_vertex_reward}
\beta_0 \sum_{v \in V} \sum_{\ac \in \scr{A}} \erew_{e}(\ac) q_{u_0,v}(\ac) \til{z}_{u_0,v}(\ac).
\end{equation}
If we could achieve the same bound for the remaining vertices $u \in U \setminus \{u_0\}$, 
the equivalence in \eqref{eqn:reward_equivalent} implies
that we'd get \eqref{eqn:main_approximation_goal} for the weaker bound of $\beta_0$. 
To achieve our goal of $\beta$, we must beat the greedy strategy,
and also extend this approach to work for all $u \in U$.
Since the former problem is the challenging step,
we focus on this. The latter problem will be handled via simulation in context of our formally defined policy (see \Cref{alg:social_welfare}).

In order to get the bound of $\beta$ claimed in \eqref{eqn:main_approximation_goal}, we must more carefully decide
when we should query $(u_0,v)$ when $Z_{u_0,v}(a) =1$, assuming $u_0$ has remaining patience and is currently unmatched. We make
this decision for $u_0$ using a \textit{patience random-order contention resolution scheme} (P-RCRS). The objective of the P-RCRS is to ensure that $\Pr[\text{$(u_0,v)$ is queried via $a$} \mid Z_{u_0,v}(a) =1] \ge \alpha$ for all $v \in V$ and $a \in \scr{A}$,
where $\alpha \in [0,1]$ is as large as possible. It operates from the aforementioned perspective of the fixed vertex $u_0$, where each time a vertex $v$
is processed with respect to $\pi$, it learns if $Z_{u_0,v}(a) =1$ for some $a \in \scr{A}$, 
and if so, decides whether to query $(u_0,v)$ via $a$. We now formalize an abstract version of this problem for an arbitrary collection of elements and random variables. The elements $[n] = \{1, \ldots ,n\}$ correspond to $\partial(u_0)$, and
$(P_{i}(a))_{i \in [n], a \in \scr{A}}$ (respectively, $(X_i(a))_{i \in [n], a \in \scr{A}}$) correspond
to $(Q_e(a))_{e \in \partial(u_0), a \in \scr{A}}$ (respectively, $(Z_e(a))_{e \in \partial(u_0), a \in \scr{A}}$).
The technical assumptions on the random variables are motivated by \Cref{lem:edge_variable}, and the independent
draws used in the execution of \Cref{alg:social_welfare_relaxed}.

\begin{definition} \label{def:patience_rcrs}
An input to the \textit{patience random-order contention resolution (P-RCRS) problem},
denoted $(\scr{A}, n, \ell, (p_i(a))_{i \in [n], a \in \scr{A}}, (x_i(a))_{i \in [n], a \in \scr{A}})$
is specified by an \textit{action set} $\scr{A}$, an integer $n \ge 1$, \textit{patience}  $1 \le \ell \le \infty$, and probability vectors
$(p_i(a))_{i \in [n], a \in \scr{A}}$ and $(x_i(a))_{i \in [n], a \in \scr{A}}$, where
\begin{equation} \label{eqn:threshold_constraints_formal}
    \text{$\sum_{i \in [n]} \sum_{a \in \scr{A}} x_{i}(a) \le \ell$, $\sum_{i \in [n]} \sum_{a \in \scr{A}} p_{i}(a) x_{i}(a) \le 1$ and $\sum_{a \in \scr{A}} x_{i}(a) \le 1$ for each $i \in [n]$.}
\end{equation}
We assume that each element $i \in [n]$ independently draws $P_{i}(a) \sim \Ber(p_i(a))$ for each $a \in \scr{A}$,
as well as $(X_i(a))_{a \in \scr{A}}$ where $X_i(a) \sim \Ber(x_i(a))$ and $\sum_{a \in \scr{A}} X_i(a) \le 1$
(which is well-defined due to the last assumption in \eqref{eqn:threshold_constraints_formal}).
We refer to $(P_i(a))_{i \in [n], a \in \scr{A}}$ as the \textit{states}, and $(X_i(a))_{i \in [n], a \in \scr{A}}$
as the \textit{suggestions}.

A \textit{patience random-order contention resolution} outputs at most one pair $(i,\ac)$ with $P_i(\ac) X_{i}(a) =1$. 
Suppose that $\pi$ is an independently drawn u.a.r. permutation of $[n]:= \{1, \ldots , n\}$.
Initially, the instantiations of $(P_i(a), X_i(a))_{i \in [n], a \in \scr{A}}$ are unknown,
and the elements of $[n]$ are presented in the increasing order specified by  $\pi$.
That is, in step $t\in [n]$, if $\pi(i) = t$ for $i \in [n]$, then
it is revealed $(X_i(a))_{a \in \scr{A}}$. If $\sum_{a \in \scr{A}} X_{i}(a) =0$, then it \textit{passes} on $i$ and moves to the next element. Otherwise, if $X_{i}(a) = 1$ for some $a \in \scr{A}$, then it either passes on $i$, or \textit{queries} $i$ via $a$, thereby revealing $P_{i}(a)$. If $P_i(a) =1$, then it must output $(i,a)$. Otherwise, it passes on $i$. It is also constrained in
that it may query at most $\ell$ elements of $[n]$ (if $\ell = \infty$, then it may query all the elements of $[n]$). Given $\alpha \in [0,1]$, we say that a P-RCRS is $\alpha$-\textit{selectable} on the input, provided for all $i \in [n]$ and $a \in \scr{A}$,
\begin{equation} \label{eqn:prcrs_selection_definition}
    \Pr[\text{$i$ is queried via $a$} \mid X_{i}(\ac) =1] \ge \alpha.
\end{equation}
Here \eqref{eqn:prcrs_selection_definition} is taken over the randomness in $(P_{i}(a), X_{i}(a))_{i \in [n], a \in \scr{A}}$, $\pi$, as well as any randomized decisions made by the P-RCRS. When \eqref{eqn:prcrs_selection_definition} holds,
we also say that the P-RCRS has a \textit{selection guarantee} of $\alpha$ on the input. If
for a fixed $\alpha \in [0,1]$, a P-RCRS satisfies \eqref{eqn:prcrs_selection_definition} for \textit{all} inputs with patience $\ell$, then we say it has a selection guarantee of $\alpha$ on inputs with patience $\ell$.
\end{definition}

Given this definition, we are now ready to describe our policy. 
We assume that our policy is given access to a P-RCRS $\psi_u$
for each $u \in U$. In order to use the selection guarantee of each $\psi_u$ (opposed to just a single vertex $u_0$ as in the previously discussed intermediate problem),  our policy defines its own suggestions $(\til{Z}_{e}(a))_{e \in E, a \in \scr{A}}$.
We ensure that these are distributed identically as $(Z_{e}(a))_{e \in E, a \in \scr{A}}$ from \Cref{alg:social_welfare_relaxed}, by using simulated bits $(\til{Q}_{e}(a))_{e \in E, a \in \scr{A}}$ to make ``fake'' queries when necessary.

\begin{algorithm}[H]
	\caption{Rounding \ref{LP:social_welfare} via P-RCRS} 
	\label{alg:social_welfare}
	\begin{algorithmic}[1]
		\Require a solution $(z_{v}(\bm{e},\bm{\ac}))_{v \in V, (\bm{e},\bm{\ac}) \in \scr{C}_v}$ to \ref{LP:social_welfare} for $G$.
		\Ensure a matching $\scr{M}$ of $G$.
		\State $\scr{M} \leftarrow \emptyset$.
            \State For each $e \in E$ and $a \in \scr{A}$, set $\til{Z}_e(a) =0$, and draw $\til{Q}_{e}(a) \sim \Ber(q_{e}(a))$ independently.
            \State Independently draw a u.a.r. order $\pi$ on $V$. 
		\For{$v \in V$ in increasing order of $\pi$} 
            \State Draw $(\bm{e}, \bm{\ac}) \in \scr{C}_v$ independently w.p. $z_{v}(\bm{e}, \bm{\ac})$,
            where $\bm{e} = ( e_1, \ldots , e_k)$ and $\bm{\ac} = (\ac_1, \ldots, \ac_k)$.

            \For{$j=1, \ldots, k$}
            \State Set $e_j = (u_j,v)$ and $\til{Z}_{u_j,v}(a_j) =1$.
            \State Execute $\psi_{u_j}$ on the P-RCRS input $(\scr{A}, \partial(u_j), (q_{u_j,v'}(a'))_{v' \in V, a' \in \scr{A}}, (\til{z}_{u_j,v'}(a'))_{v'\in V, a' \in \scr{A}})$ using the arrival order on $\partial(u_j)$
            induced by $\pi$, with random variables $(Q_{u_j,v'}(a))_{v' \in V, a' \in \scr{A}}$ as the states
            and $(\til{Z}_{u_j,v'}(a'))_{v' \in V, a' \in \scr{A}}$ as the suggestions. \label{eqn:coupling}
            \If{$\psi_{u_j}$ queries $(u_j,v)$ via $\ac_j$}
            
            \State $\scr{M} \leftarrow \scr{M} \cup \{(u_j,v)\}$ provided $Q_{u_j,v}(a_j) =1$.
            \State Exit the internal for loop.
            \ElsIf{$\til{Q}_{u_j,v}(a_j) =1$}                \Comment{simulate to maintain correct marginals}
            \State Exit the internal for loop.
            \EndIf
            \EndFor
		\EndFor
		\State \Return $\scr{M}$.
	\end{algorithmic}
\end{algorithm}
Now, for any $u \in U$, $\psi_u$ is passed the input $(\scr{A}, \partial(u), (q_{u,v}(a))_{v \in V, a \in \scr{A}}, (\til{z}_{u,v}(a))_{v\in V, a \in \scr{A}})$, which satisfies the required inequalities in \Cref{def:patience_rcrs}, due to \Cref{lem:edge_variable}.
However, we still must ensure that the random variables used in line \ref{eqn:coupling}. are drawn via
$(q_{u,v}(a))_{v \in V, a \in \scr{A}}$ and $(\til{z}_{u,v}(a))_{v\in V, a \in \scr{A}}$ in the required
way as specified in \Cref{def:patience_rcrs}. We verify this in \Cref{lem:fixed_query_vertex}.

\begin{lemma}[Proof in \Cref{pf:lem:fixed_query_vertex}] \label{lem:fixed_query_vertex}
For each $u \in U$, $(\scr{A}, \partial(u), (q_{u,v}(a))_{v \in V, a \in \scr{A}}, (\til{z}_{u,v}(a))_{v\in V, a \in \scr{A}})$ satisfies the required inequalities of a P-RCRS input. Moreover, the random variables $(Q_{u,v}(a))_{a \in \scr{A}}$
and $(\til{Z}_{u,v}(a))_{a \in \scr{A}, v \in V}$ used in line \ref{eqn:coupling}. are drawn via $(q_{u,v}(a), \til{z}_{u,v}(a))_{v\in V, a \in \scr{A}}$
in the way specified in \Cref{def:patience_rcrs}.
\end{lemma}

Assuming \Cref{lem:fixed_query_vertex}, it is now easy to analyze the performance of \Cref{alg:social_welfare}:

\begin{theorem} \label{thm:reduction_to_trs}
Fix $\alpha \ge 0$. Suppose that each P-RCRS $\psi_u$ used by \Cref{alg:social_welfare} is $\alpha$-selectable. Then, expected
reward of \Cref{alg:social_welfare} is at least
$
\alpha \cdot \sum_{v \in V} \sum_{(\bm{e},\bm{\ac}) \in \scr{C}_v } \srew(\bm{e}, \bm{\ac})  z_{v}(\bm{e}, \bm{\ac}).
$
\end{theorem}
\begin{proof}[Proof of \Cref{thm:reduction_to_trs}]

Using \Cref{lem:fixed_query_vertex}, the expected reward of the policy is
\begin{align}
    &\sum_{u \in U} \sum_{v \in V, \ac \in \scr{A}} \erew_{(u,v)}(\ac) \cdot \mb{P}[ \til{Z}_{u,v}(a) =1] \cdot \mb{P}[Q_{u,v}(a) =1 \cap \text{\{$(u,v)$ is queried via $a$\}} \mid  \til{Z}_{u,v}(a) =1] \cdot  \notag\\
    &=\sum_{u \in U} \sum_{v \in V, \ac \in \scr{A}}  \erew_{u,v}(\ac) \cdot q_{u,v}(\ac) \cdot \til{z}_{u,v}(\ac) \cdot \mb{P}[\text{$(u,v)$ is queried via $a$} \mid  \til{Z}_{u,v}(a) =1], \notag
\end{align}
where the equality follows since the policy decides to query $(u,v)$ via $a$ prior to revealing $Q_{u,v}(a)$, and that $\mb{P}[\til{Z}_{u,v}(a)] = \til{z}_{u,v}(a)$ by \Cref{lem:fixed_query_vertex}.
On the other hand, since each $u \in U$ executes an $\alpha$-selectable P-RCRS $\psi_u$ on a well-defined input due to \Cref{lem:fixed_query_vertex}, the definition of $\alpha$-selectable
implies that $\mb{P}[\text{$(u,v)$ is queried via $a$} \mid  \til{Z}_{u,v}(a) =1] \ge \alpha$
for each $\ac \in \scr{A}$ and $v \in V$. The theorem thus follows by \eqref{eqn:reward_equivalent}.
\end{proof}

It remains to argue that there exists a $\beta$-selectable P-RCRS for the dependence
on $\beta$ as claimed in \Cref{thm:main_theorem}.

% and $\alpha_{\ell} := \int_{0}^{1} e^{-y}\mb{P}[ \Pois( (\ell - 1) y) < \ell] dy$.
\begin{theorem} \label{thm:trs}
For inputs with patience $2 \le \ell < \infty$, there exists a P-RCRS with a selection guarantee of $\frac{1}{27} (19 - \frac{67}{e^3}) \approx 0.5801$. For inputs with patience $\ell \in \{1, \infty\}$,
this selection guarantee improves to $1-1/e \approx 0.6321$.
\end{theorem}

\subsection{Proving \Cref{thm:trs}} \label{sec:main_section_srcrs}
In order to design a P-RCRS with the guarantee claimed in \Cref{thm:trs},
we first prove a reduction to the case when the input has a \textit{trivial action set} (i.e., $|\scr{A}| =1$). 
When $|\scr{A}| =1$, we note that this is a special case of the definition in \cite{Adamczyk15} for single-item selection and uniform matroid (i.e., patience) query constraints. More, as we expand upon below, for $\ell \in \{1, \infty\}$ and $|\scr{A}| =1$, the P-RCRS problem is equivalent
to the \textit{(single-item) random order
contention resolution (RCRS) problem}, and so we will be able to use existing RCRS results to establish these cases of \Cref{thm:trs}.

We now state our reduction from the case of an arbitrary action space to
a trivial one. The main idea used to prove \Cref{lem:trs_to_prcrs} is as follows: Suppose we are given an arbitrary input to the P-RCRS problem with $|\scr{A}| \ge 2$, whose states and suggestions we denote by $(P_i(a))_{a \in \scr{A}, i \in [n]}$ and  $(X_i(a))_{a \in \scr{A}, i \in [n]}$, respectively.
We construct an input to the P-RCRS problem with a single action, and couple it so that its states $(X_i)_{i \in [n]}$ and suggestions $(P_i)_{i \in [n]}$ satisfy $X_i = \sum_{a \in \scr{A}} X_i(a)$ and $P_i X_i = \sum_{a \in \scr{A}} P_{i}(a) X_i(a)$ for all $i \in [n]$.
By executing an $\alpha$-selectable P-RCRS for a single action on $(P_i,X_i)_{i=1}^n$, we can then transfer its
selection guarantee to apply to the input with $|\scr{A}| \ge 2$.

\begin{lemma}[Proof in \Cref{pf:lem:trs_to_prcrs}]\label{lem:trs_to_prcrs}
Given an $\alpha$-selectable P-RCRS for inputs with patience $\ell$ and the trivial action space, there exists an $\alpha$-selectable P-RCRS
for inputs with patience $\ell$ and an arbitrary action space.
\end{lemma}

Due to \Cref{lem:trs_to_prcrs}, we now focus exclusively on the P-RCRS problem with the trivial action set for the remainder of the section. As previously suggested, we drop the dependence on $a \in \scr{A}$,
and write our input as  $\scr{I} =(n,  \ell, (p_i)_{i=1}^n, (x_i)_{i=1}^n)$, 
where now $\sum_{i=1}^n x_i \le \ell$ and $\sum_{i=1}^n p_i x_i \le 1$. 
Then, each $i \in [n]$ draws state $P_i \sim \Ber(p_i)$ and suggestion $X_i \sim \Ber(x_i)$ independently.
More, since there is now a single action, we simply say that the P-RCRS queries the elements of $[n]$ (opposed to specifying which action this is done by). Our goal is then to design a P-RCRS such that for all $i \in [n]$,
\begin{equation} \label{eqn:prcrs_guarantee_trivial}
    \Pr[\text{$i$ is queried} \mid X_{i} =1] \ge \beta,
\end{equation}
where $\beta = 1-1/e$ if $\ell \in \{2, \infty\}$, and $\beta = \frac{1}{27} (19 - \frac{67}{e^3})$ otherwise.
Now, if $\ell = \infty$, then the P-RCRS can query all the elements of $[n]$,
and so one can use an $\alpha$-selectable RCRS on the input $(n, (p_i x_i)_{i=1}^n)$ to get an $\alpha$-selectable P-RCRS. Similarly, if $\ell =1$, then $\sum_{i=1}^{n} x_i \le 1$, and so one can use an $\alpha$-selectable
RCRS on $(n,(x_i)_{i=1}^n)$ to get an $\alpha$-selectable P-RCRS. Thus, the cases of $\ell \in \{1, \infty\}$ in \Cref{thm:trs} are immediate due  to the $(1-1/e)$-selectable RCRS of \cite{Lee2018, brubach2024offline}.

We thus design an P-RCRS for $2 \le \ell < \infty$, in which case $\beta := \frac{1}{27} (19 - \frac{67}{e^3})$.
It will be convenient to define $x(S):= \sum_{i \in S} x_i$ for each $S \subseteq [n]$, and $\Pois(\mu)$ to be a Poisson random variable of mean $\mu \ge 0$.
We define the \textit{attenuation function} $b: [0,1] \rightarrow [0,1]$, where for $s \in [0,1]$, 
\begin{equation}
    b(s) = \frac{\beta}{\int_{0}^{1} e^{-y(1-s)}\mb{P}[ \Pois(2y) < 3] dy}
\end{equation}
It is easy to check that $\beta = \int_{0}^{1} e^{-y}\mb{P}[ \Pois( 2y) < 3] dy$,
and so $b(0) =1$. More, $b$ is decreasing function of $s$ on $[0,1]$.

Our P-RCRS draws an independent random bit $B_i \sim \Ber(b(p_i x_i))$ for each $i \in [n]$. Then, when $i$ arrives, it queries $i$ if $B_i X_i =1$, less than $\ell$ queries were previously made, and no element was previously output. We provide a formal
description below:
\begin{algorithm}[H]
\caption{P-RCRS} 
\label{alg:P-RCRS}
\begin{algorithmic}[1]
\Require $\scr{I} = (n, \ell, (x_i)_{i=1}^n, (p_i)_{i=1}^n)$
\Ensure at most one queried element $i \in [n]$ with $P_i X_i =1$.
\State Set $L \leftarrow 0$.
\For{arriving $i \in [n]$}

\State Draw $B_i \sim \Ber(b(p_i x_i))$ independently.
\If{$B_i X_{i} =1$, $L< \ell$, and no element was previously output}
\State Query $i$ and set $L \leftarrow L + 1$.
\If{$P_i =1$}
\State \Return $i$.
\EndIf
\EndIf
\EndFor

\end{algorithmic}
\end{algorithm}
In order to establish \eqref{eqn:prcrs_guarantee_trivial},  we reformulate the
arrival model, and define each $i \in [n]$ to have an independent and u.a.r. \textit{arrival time} $Y_i \in [0,1]$.
We assume that the elements of $[n]$ have distinct arrival times, and that they are presented to \Cref{alg:P-RCRS}
in increasing order of these values. 

For notational convenience, let us hereby focus on proving \eqref{eqn:prcrs_guarantee_trivial} for $i=1$. We say that $1$ is \textit{safe}
when executing \Cref{alg:P-RCRS} on $\scr{I}$, provided that before $1$ arrives:
\begin{enumerate}
    \item At most $\ell -1$ queries were made.
    \item No element of $[n] \setminus \{1\}$ was output.
\end{enumerate}
Observe then that
\begin{equation} \label{eqn:original_output_probability}
    \mb{P}[\text{$1$ is queried} \mid Y_1 = y, X_1 =1]   = b(p_1 x_1) \mb{P}[\text{$1$ is safe} \mid Y_1 =y]. 
\end{equation}
We first argue with respect to lower bounding \eqref{eqn:original_output_probability},
we can assume our input $\scr{I}$ is of a specific form. That is, $p_j \in \{0,1\}$
for all $j \in [n] \setminus \{1\}$, $p_1 =1$, and $\sum_{i \in [n]} x_i = \ell$. Note that the first part of this simplification
is proven in Lemma $2$ of \cite{brubach2024offline}, and our argument proceeds almost identically.
We construct a new input $\scr{\til{I}}$ which satisfies the required conditions,
and argue that for all $y \in [0,1]$, \eqref{eqn:original_output_probability} can only decrease when \Cref{alg:P-RCRS} is executed
on $\scr{\til{I}}$ instead of $\scr{I}$.

\begin{lemma}[Proof in \Cref{pf:lem:p_values}]\label{lem:p_values}
When lower bounding \eqref{eqn:original_output_probability} over all $y \in [0,1]$, we may assume that
the input $\scr{I}$ satisfies $p_1 =1$, $p_j \in \{0,1\}$ for all $j \in [n] \setminus \{1\}$,
and $\sum_{i \in [n]} x_i = \ell$.

\end{lemma}

Let us now define $N_{q}:= \{j \in [n] \setminus \{1\}: p_j = q\}$ for $q \in \{0,1\}$. Due
to \Cref{lem:p_values}, we hereby assume that $N_{0} \cup N_{1}$ partitions $[n] \setminus \{1\}$,
and that $p_1 =1$. 
Let us now define $L_0$ as the number of queries \Cref{alg:P-RCRS} makes to $N_0$ before $1$ arrives.
Observe then that
\begin{align*}
    \mb{P}[\text{$1$ is safe} \mid Y_1 =y] &= \mb{P}[L_0 \le \ell -1 \mid Y_1 = y] \cdot \mb{P}[\cap_{i \in N_1} \{ B_i X_i \bm{1}_{Y_i < y} =0 \}  \mid Y_1 =y]  \\
    &= \mb{P}[L_0 \le \ell -1 \mid Y_1 = y]  \prod_{j \in N_1} (1 - y b(x_j) x_j),
\end{align*}
where the second line uses the independence of $(B_i, X_i, \bm{1}_{Y_i <y})_{i \in N_1}$, conditional on $Y_1 = y$.
Thus, combined with \eqref{eqn:original_output_probability},
\begin{equation} \label{eqn:simplified_output_probability}
    \mb{P}[\text{$1$ is queried} \mid X_1 =1]  = b(x_1) \int_{0}^{1}\mb{P}[L_0 \le \ell -1 \mid Y_1 =y] \prod_{j \in N_1} (1 - y b(x_j) x_j) dy.
\end{equation}

We next argue that with respect to lower bounding \eqref{eqn:simplified_output_probability}, we may assume that $\max_{j \in N_1} x_j \rightarrow 0$. I.e., the worst-case configuration of $(x_j)_{j \in N_1}$ occurs in the \textit{Poisson regime}. Since $b(0) =1$, this allows us to replace $\prod_{j \in N_1} (1 - y b(x_j) x_j)$ by $e^{-y x(N_1)}$ in \eqref{eqn:simplified_output_probability}. Unlike
\Cref{lem:p_values}, this worst-case configuration only applies 
if we integrate over $y \in [0,1]$.
\begin{lemma}[Proof in \Cref{pf:lem:splitting_argument}]\label{lem:splitting_argument} 
% For each $0 \le x_1 \le 1$, $0 \le x(N_1) \le 1 - x_1$ and 
The probability $\mb{P}[\text{$1$ is queried} \mid X_1 =1]$ is lower
bounded by
\begin{equation} \label{eqn:splitting_argument}
b(x_1) \int_{0}^1{} e^{-y x(N_1)} \mb{P}[L_0 \le \ell -1 \mid Y_1 =y]  dy.
\end{equation}
\end{lemma}
We prove \Cref{lem:splitting_argument} using a \textit{splitting
argument}. Specifically, given $i \in N_1$, we construct a new input $\scr{\til{I}}$
with an additional element $i'$ added to $N_1$. Then, we set the new fractional value
of $i$ and $i'$ to each be $x_i/2$, and argue that \eqref{eqn:simplified_output_probability}
is lower bounded by the analogous term for $\scr{\til{I}}$. By iteratively applying this argument,
and using that $b(0)=1$, we conclude that the worst-case configuration of $(x_i)_{i \in N_1}$ occurs in the Poisson regime.
For this argument to work, we need the following analytic properties of the attenuation function $b$. 
\begin{proposition}[Proof in \Cref{pf:property_first_order}]\label{property:first_order}
 The attenuation function $b:[0,1] \rightarrow [0,1]$ satisfies the following properties:
 \begin{enumerate}
     \item $b(0) =1$ and $b$ is non-increasing on $[0,1]$.
     \item$\int_{0}^{1} (1 - y z b(z)) dy \ge \int_{0}^{1} (1 - y z b(z/2)/2)^2 dy$
     for each $z \in [0,1]$.
     \item For each $z \in [0,1]$, there exists $y_z \in [0,1]$, such that $ (1 - y z b(z)) - (1 - y z b(z/2)/2)^2 \ge 0$ for all $y \in [0, y_z]$, and $ (1 - y z b(z)) - (1 - y z b(z/2)/2)^2 \le 0 $
     for all $y \in (y_z, 1]$.
 \end{enumerate}
\end{proposition}

We next consider the term $\mb{P}[L_0 \le \ell -1 \mid Y_1 =y]$ of \eqref{eqn:splitting_argument}.
Observe that conditional on $Y_1 = y$, $L_0$ is a sum of $|N_0|$ independent Bernoulli random variables.
Moreover, $\mb{E}[L_0 \mid Y_i = y] = \sum_{j \in N_0} b(p_j x_j) x_j y = \sum_{j \in N_0} x_j y = y x(N_1)$,
as $b(0) = 1$, and $p_j =0$ for all $j \in N_0$. 
A partial characterization of the worst-case value of $\mb{P}[L_0 \le \ell -1 \mid Y_1 =y]$ is implied by Lemma $7$ of \cite{BavejaBCNSX18},
which we state below for convenience:
\begin{lemma}[Lemma 7 in \cite{BavejaBCNSX18}] \label{lem:poisson_lemma}
Fix an integer $t \ge 1$. Suppose $Z_1, \ldots , Z_k$ are independent Bernouli random variables, and $\mu :=\sum_{i=1}^k \mb{E}[Z_i] < t -1$. Then,
$$\mb{P}\left[\sum_{i=1}^k Z_i < t \right] \ge \mb{P}[\Pois(\mu) < t] = \sum_{i=0}^{t -1} \frac{ e^{-y \mu} \mu^{i}}{i!}.$$
\end{lemma}
Applied to our setting, if $\mb{E}[L_0 \mid Y_i = y] = y x(N_1) < \ell -1$,
then 
\begin{equation} \label{eqn:poisson_regime}
    \mb{P}[L_0 \le \ell -1 \mid Y_1 =y] \ge \mb{P}[ \Pois( x(N_0) y) < \ell] = \sum_{k=0}^{\ell -1} \frac{ e^{-y x(N_0)} (y x(N_0))^{k}}{k!}.
\end{equation}
On the other hand, if we recall the feasibility constraints on $\scr{I}$ after the simplification of \Cref{lem:p_values},
\begin{align}
  \text{$x(N_1) + x_1 \le 1$, and $x(N_1) = \ell - x(N_0) - x_1$.} \label{eqn:feasibility_constraints_simplified}
\end{align}
Thus, if $x(N_1) + x_1 < 1$, then $y x(N_0) \ge \ell -1$ for $y \ge (\ell -1)/x(N_0)$,
and so \eqref{eqn:poisson_regime} does not apply on the interval $[ (\ell -1)/x(N_0), 1]$ \footnote{This detail was overlooked in \cite{brubach2024offline, DBLP:journals/corr/abs-2205-08667}, as they apply Lemma $7$ of \cite{BavejaBCNSX18} for all $y \in [0,1]$ (even when 
their analogously defined quantity has $x(N_1) + x_1 < 1$). This is not an issue in \cite{BavejaBCNSX18}, as they do not use attenuation in their policy, and so they're able to easily argue that they can assume $x(N_1) + x_1 = 1$.}. 
Moreover, this is not simply a limitation
of the characterization of \cite{BavejaBCNSX18}. For $y$ close to $1$, the worst-case
of behaviour of $\mb{P}[L_0 \le \ell -1 \mid Y_1 =y]$ changes to an input when $|N_0| = \ell$,
and $x_i = x(N_0)/\ell$ for all $i \in N_0$. For instance, in the extreme case of $y=1$
and $x(N_0) = \ell$, one can make $\mb{P}[L_0 \le \ell -1 \mid Y_1 =y] = 0$.
Determining the worst-case configuration of $(x_i)_{i \in N_0}$ as a function of $y$ seems difficult, as it is closely related proving a longstanding conjecture of Samuels from probability theory
(see \cite{Samuels1966,Samuels1969} and \cite{Feige2006,paulin2017conjecturessamuelsfeige}). More, it is also challenging to follow the approach of \Cref{lem:splitting_argument}, and determine the worst-case behaviour of
$\int_{0}^{1} \mb{P}[L_0 \le \ell -1 \mid Y_1 =y] dy$. For instance, if $x(N_0) = \ell$, then the worst-case configuration $(x_i)_{i \in N_0}$ has $|N_0| = \ell$
and $x_i = 1$ for all $i \in N_0$, which is the opposite of the Poisson regime.
Instead, we side-step this problem, where our approach differs depending on if $\ell =2$, $3 \le \ell < 120$
or $\ell \ge 120$. 

For $\ell =2$, we apply an \textit{exchange argument} (formally stated in \Cref{lem:exchange_argument} of
\Cref{sec:patience_2_exchange}).
Specifically, we construct an input $\scr{\til{I}}$ which adds a new element $i_1$
to $N_1$ of fractional value at most $1 - x(N_1) - x_1$, and reduces the fractional value of an element of $i_0$ of $N_0$ by the same amount, thus ensuring that $\scr{\til{I}}$ is feasible. We then argue that \eqref{eqn:splitting_argument} for
$\scr{I}$ is lower bounded by the analogously defined term for $\scr{\til{I}}$.
This implies that for $\ell =2$, when lower bounding \eqref{eqn:splitting_argument},
we may assume that $x(N_1) = 1- x_1$.

For $3 \le \ell < 120$, while it is conceivable that an exchange argument could also apply, the comparisons become
more complicated, and so we abandon this approach. Instead, we apply \eqref{eqn:poisson_regime} to $\mb{P}[L_0 \le \ell -1 \mid Y_1 =y]$ when $y \in [0, (\ell -1)/x(N_0))$, and the trivial lower
bound of $0$ when $y \in [(\ell -1)/x(N_0)), 1]$. We then integrate over $y \in [0,1]$
to get a lower bound on \eqref{eqn:splitting_argument} of $b(x_1) \cdot F_{\ell}(x_1, x(N_1))$. The function $F_{\ell}$ has a closed-form expression,
and we prove that $F_{\ell}(x_1, x(N_1)) \ge F_{\ell}(x_1, 1- x_1)$ (see \Cref{prop:worst_case_patience_mid} of \Cref{sec:patience_mid}).
Thus, we are once again able to assume that $x(N_1) = 1- x_1$ when lower bounding
\eqref{eqn:splitting_argument}.

In either case (i.e., $\ell =2$ or $3 \le \ell < 120$), we may assume
that $x(N_1) = 1 -x_1$, and so $x(N_0) = \ell -1$. We can thus apply \eqref{eqn:poisson_regime}  to \eqref{eqn:splitting_argument} for all $y \in [0,1)$. Combining \Cref{lem:splitting_argument} with the definition of
$b$, this implies that
\begin{align}
    \mb{P}[\text{$1$ is queried}  \mid X_1 =1] &\ge b(x_1) \int_{0}^{1} \mb{P}[\Pois(y (\ell -1)) < \ell] e^{-y(1-x_1)} dy  \notag \\
                                               &= \frac{\beta}{\int_{0}^{1} \mb{P}[\Pois(2y) < 3] e^{-y(1-x_1) } dy}  \int_{0}^{1} \mb{P}[\Pois(y(\ell-1)) \le \ell] e^{-y(1-x_1)} dy. \label{eqn:final_opt_lower_bound}
\end{align}
For each $3 \le \ell < 120$, \eqref{eqn:final_opt_lower_bound} is a function of $0 \le x_1 \le 1$ with a closed form expression, and we verify that it is attains its minimum when $x_1 = 0$ and $\ell =3$, in which case
it is exactly $\beta$. This proves the guarantee of \Cref{thm:trs} assuming $2 \le \ell < 120$.

Finally, for the case when $\ell \ge 120$, since $\ell$ is sufficiently large, it suffices to apply Bennett's inequality to $\mb{P}[L_0 \le \ell -1 \mid Y_1 =y]$ for each $y \in [0,1]$. In \Cref{sec:patience_big},
we show that this leads to a bound of
\begin{align}
    \mb{P}[\text{$1$ is queried}  \mid X_1 =1] &\ge b(x_1) \int_{0}^{1} \left( 1- e^{-120 (y + \log\left(\frac{1}{y}\right) - 1)} \right) e^{-y(1-x_1)} dy. \label{eqn:final_opt_lower_bound_large} 
    % \\
    % &= b(x_1) 
\end{align}
Now, \eqref{eqn:final_opt_lower_bound_large} is a function of $0 \le x_1 \le 1$ with a closed-form expression. We verify that it attains its minimum occurs when $x_1 = 1$, in which case we get 
$$
 \frac{\beta}{\int_{0}^{1} \mb{P}[\Pois(2y) < 3] dy} \cdot \int_{0}^{1} \left( 1- e^{-120 (y + \log\left(\frac{1}{y}\right) - 1)} \right) dy \ge 0.5803,
$$ which is strictly large than $\beta$. This proves the
guarantee of \Cref{thm:trs} for $\ell \ge 120$.

We now provide the details of the $\ell =2$, $3 \le \ell < 120$ and $\ell \ge 120$ regimes, in Subsections \ref{sec:patience_2_exchange}, \ref{sec:patience_mid}, and
\ref{sec:patience_big}, respectively.

\subsection{Details of Patience $\ell=2$} \label{sec:patience_2_exchange}
Given the input $\scr{I} = (n,2, (x_i)_{i=}^n, (p_i)_{i=1}^n)$, we construct a new input 
$\scr{\til{I}} =(n +1, 2, (\til{x}_i)_{i=1}^{n+1}, (\til{p}_i)_{i=1}^{n+1})$, with an additional
element $i_1 := n +1$, and which modifies the fractional value of an arbitrary $i_0 \in N_0$. Specifically,
\begin{enumerate}
    \item $\til{p}_{i_1} = 1$ and $\til{x}_{i_1} := \min\{1 - x(N_1) - x_1, x_{i_0}\}$,
    \item $\til{p}_{i_0} = 0$ and $\til{x}_{i_0} = x_{i_0} - \min\{1 - x(N_1) - x_1, x_{i_0}\}$,
    \item $(\til{x}_i, \til{p}_i) = (x_i,p_i)$ for all $i \in [n] \setminus \{i_0\}$.
\end{enumerate}
Clearly, $\scr{\til{I}}$ is a feasible input,
which satisfies the conditions of \Cref{lem:p_values}. I.e., $\til{p}_j \in \{0,1\}$ for all $j \in [n+1]$, $\til{p_1} = 1$ and $\sum_{j \in [n+1]} \til{x}_j = \ell$. Define $\til{N}_{0} = N_0 \setminus \{i_0\}$ and $\til{N}_1 = N_1 \cup \{i_1\}$, and $\til{L}_{0}$ as the number of queries \Cref{alg:P-RCRS} makes to
$\til{N}_0 \cup \{i_0\} = N_0$ when executing on $\scr{\til{I}}$. Then, using $X_1$ and $Y_1$ for $1$ during the execution of \Cref{alg:P-RCRS} on $\scr{\til{I}}$, \Cref{lem:splitting_argument} implies that
$$
\mb{P}[\text{$1$ from $\scr{\til{I}}$ is queried} \mid X_1 =1] \ge b(x_1) \int_{0}^{1} \mb{P}[ \til{L}_0 \le 1 \mid Y_1 =y] e^{-y (x(N_1) + \til{x}_{i_1})} dy,
$$
where we've used that $\sum_{i \in \til{N}_1} \til{x}_i = x(N_1) + \til{x}_{i_1}$,
and $\til{x}_1 = x_1$. We are now ready to formalize the statement of our exchange
argument:
\begin{lemma}\label{lem:exchange_argument}
The following inequality holds:
$$\int_{0}^{1} \mb{P}[ L_0 \le 1 \mid Y_1 =y] e^{-y x(N_1)}  dy \ge \int_{0}^{1} \mb{P}[ \til{L}_0 \le 1 \mid Y_1 =y] e^{-y (x(N_1) + \til{x}_{i_1})}  dy.$$  
 \end{lemma}

\begin{proof}[Proof of \Cref{lem:exchange_argument}]
Our goal is to prove that
\begin{equation} \label{eqn:exchange_integral}
    \int_{0}^{1} e^{-y x(N_1)} \left( \mb{P}[ L_0 \le 1 \mid Y_1 =y] - e^{-y\til{x}_{i_1}}  \mb{P}[ \til{L}_0 \le 1 \mid Y_1 =y] dy  \right) \ge 0.
\end{equation}
First observe that $\mb{P}[ L_0 \le 1 \mid Y_1 =y]$ is equal to 
\begin{align} 
&& (1- y x_{i_0}) \prod_{j \in \til{N}_0}(1- y x_{j}) + y x_{i_0} \prod_{j \in \til{N}_0}(1 - y x_j) + \sum_{i \in \til{N}_0}y x_i (1 - y x_{i_0}) \prod_{j \in \til{N}_0 \setminus \{i\}}(1 - y x_j). \notag \\
&&= \prod_{j \in \til{N}_0}(1- y x_{j}) +  \sum_{i \in \til{N}_0}y x_i (1 - y x_{i_0}) \prod_{j \in \til{N}_0 \setminus \{i\}}(1 - y x_j).                                      \label{eqn:base_probability}
\end{align}
and similarly, $\mb{P}[ \til{L}_0 \le 1 \mid Y_1 =y]$ is equal to
\begin{equation} \label{eqn:exchange_probability}
\prod_{j \in \til{N}_0}(1- y x_{j})  + \sum_{i \in \til{N}_0}y x_i (1 - y \til{x}_{i_0})\prod_{j \in \til{N}_0 \setminus \{i\}}(1 - y x_j).
\end{equation}
Thus, applying \eqref{eqn:base_probability} and \eqref{eqn:exchange_probability}, we can rewrite the integrand of \eqref{eqn:exchange_integral} as
\begin{align} 
&&e^{-y x(N_1)}\left( (1 - e^{-y \til{x}_{i_1}}) \prod_{j \in \til{N}_0}(1- y x_{j}) + \sum_{i \in \til{N}_0} (1 - y x_{i_0} - e^{-y \til{x}_{i_1}}(1 - y \til{x}_{i_0})) y x_i \prod_{j \in \til{N}_0 \setminus \{i\}}(1 - y x_j)\right). \notag \\
&&= e^{-y x(N_1)}\left( (1 - e^{-y \til{x}_{i_1}}) \prod_{j \in \til{N}_0}(1- y x_{j}) - \sum_{i \in \til{N}_0} (e^{-y \til{x}_{i_1}}(1 - y \til{x}_{i_0}) - (1 - y x_{i_0})) y x_i \prod_{j \in \til{N}_0 \setminus \{i\}}(1 - y x_j)\right) \label{eqn:first_difference}.
\end{align}
Now, since $\til{x}_{i_1} \ge 0$, and $\til{x}_{i_1} + \til{x}_{i_0} = x_{i_0}$,
we know that for all $y \in [0,1]$,
\begin{equation} \label{eqn:leading_terms_sign}
\text{$1 - e^{-y \til{x}_{i_1}} \ge 0$ and $e^{-y \til{x}_{i_1}}(1 - y \til{x}_{i_0}) - (1 - y x_{i_0}) \ge (1 - y \til{x}_{i_1})(1 - y \til{x}_{i_0}) - (1 - y x_{i_0}) \ge 0$,}
\end{equation}
where we have applied the elementary inequality $1 - z \le e^{-z}$.
% where the second inequality follows from \ref{prop:exchange_ineq_1}. of \Cref{prop:exchange_ineq}.
Thus, we focus on lower and upper bounding 
$\prod_{j \in \til{N}_0}(1- y x_{j})$ and $\sum_{i \in \til{N}_0} y x_i \prod_{j \in \til{N}_0 \setminus \{i\}}(1 - y x_j)$, respectively. 
Starting with the first term, observe that
for any fixed $y \in [0,1]$ and $k := |\til{N}_0|$,
the function $\psi: [0,1]^{k} \rightarrow [0,1]$, $\psi(z_1, \ldots ,z_k) := \prod_{j =1}^k(1- y z_{j})$
is \textit{Schur-concave} (see \cite{Peari1992}). This implies
that its minimum over all $(z_1, \ldots ,z_k)  \in [0,1]^k$ with $\sum_{i=1}^k z_i = x(\til{N}_0)$
is attained when $z_i =1$ for $i=1, \ldots, \lfloor x(\til{N}_0)  \rfloor$,
and $z_{ \lfloor x(\til{N}_0)  \rfloor + 1} = x(\til{N}_0) - \lfloor x(\til{N}_0)  \rfloor$. We thus get that
\begin{equation} \label{eqn:lower_bound_product}
    \prod_{j \in \til{N}_0}(1- y x_{j}) \ge (1 - y)^{\lfloor x(\til{N}_0) \rfloor} (1 -y (x(\til{N}_0) - \lfloor x(\til{N}_0) \rfloor))
\end{equation}
On the other hand, using the elementary inequality $1 - z \le e^{-z}$, and that $x_i \le 1$
for all $i \in \til{N}_0$,
\begin{align} \label{eqn:upper_bound_sum_product}
    \sum_{i \in \til{N}_0} y x_i \prod_{j \in \til{N}_0 \setminus \{i\}}(1 - y x_j) \le \sum_{i \in \til{N}_0} y x_i e^{-y(x(\til{N}_0) - x_i)} \le  y x(\til{N}_0) e^{-y(x(\til{N}_0) - 1)}.
\end{align}
Thus, applying \eqref{eqn:leading_terms_sign}, \eqref{eqn:lower_bound_product} and \eqref{eqn:upper_bound_sum_product} to \eqref{eqn:first_difference}, we can lower bound
the integrand of \eqref{eqn:exchange_integral} by
\begin{align} 
    &e^{-y x(N_1)}(1 - e^{-y \til{x}_{i_1}}) (1 - y)^{\lfloor x(\til{N}_0) \rfloor} (1 -y (x(\til{N}_0) -\lfloor x(\til{N}_0) \rfloor)) \notag \\
    &- (e^{-y \til{x}_{i_1}}(1 - y \til{x}_{i_0}) - (1 - y x_{i_0})) y x(\til{N}_0) e^{-y(x(\til{N}_0) + x(N_1) - 1)}. \label{eqn:final_simplification}
\end{align}
Observe next that $x(\til{N}_0) =  x(N_0) - x_{i_0} = 2 - x(N_1) - x_1 - x_{i_0}$,
since $x(N_0) = 2 - x(N_1) - x_1$ by \eqref{eqn:feasibility_constraints_simplified}.
We simplify \eqref{eqn:final_simplification}, depending on whether $x_{i_0} > 1 - x(N_1) - x_1$, or $x_{i_0} \le 1- x(N_1) - x_1$. 

If $x_{i_0} > 1 - x(N_1) - x_1$, then $x(\til{N}_0) =  2 - x(N_1) - x_1 - x_{i_0} \le 1$,
so $\lfloor x(\til{N}_0) \rfloor =0$. Moreover, $\til{x}_{i_1} = 1 - x(N_1) - x_1$, and $\til{x}_{i_0} = x_{i_0} - (1 - x(N_1) - x_1)$. Thus, \eqref{eqn:final_simplification} simplifies to
\begin{align} \label{eqn:final_simplification_case_1}
    & (e^{-y x(N_1)} - e^{-y(1- x_1)})(1 - y(2 - x(N_1) - x_1 - x_{i_0})) \notag \\
    &- (e^{-y (1-x(N_1) - x_1)}(1 - y (x_{i_0} - (1- x(N_1) - x_1)) - (1 - y x_{i_0})) y(2 - x(N_1) - x_1 - x_{i_0}) e^{-y(1 - x_1 - x_{i_0})}.
\end{align}
% \begin{align} \label{eqn:final_simplification_case_1}
%     & e^{-y x(N_1)}\left( (1 - e^{-y(1- x(N_1) - x_1)})(1 - y(2 - x(N_1) - x_1 - x_{i_0}) \right) \notag \\
%     &- e^{-y x(N_1)}\left( e^{-y (1 - x(N_1) - x_1)}( 1 - y(x_{i_0} -(1 -  x(N_1) - x_1))) \right)
% \end{align}
% $g(\til{x}_{i_0}, \til{x}_{i_0},y):= e^{-y (1-x(N_1) - x_1)}(1 - y (x_{i_0} - (1- x(N_1) - x_1)) - (1 - y x_{i_0})$
% =
Otherwise, if $x_{i_0} \le 1- x(N_1) - x_1$, then $x(\til{N}_0) \ge 1$, so $\lfloor x(\til{N}_0) \rfloor =1$. More, $\til{x}_{i_0} = 0$ and $\til{x}_{i_1} = x_{i_0}$. Thus, \eqref{eqn:final_simplification} simplifies to
\begin{align} 
    &e^{-y x(N_1)}(1 - e^{-y x_{i_0}}) (1 - y) (1 -y (1 - x(N_1) - x_1 - x_{i_0})) \notag \\
    &- (e^{-y x_{i_0}} - (1 - y x_{i_0})) y (2 - x(N_1) - x_1 - x_{i_0}) e^{-y (1 - x_1 - x_{i_0})}. \label{eqn:final_simplification_case_2}
\end{align}
If we integrate \eqref{eqn:final_simplification_case_1} (respectively, \eqref{eqn:final_simplification_case_2})
over $y \in [0,1]$, then we are left with a function of non-negative variables $x(N_1), x_{i_0}$ and $x_1$, where $x(N_1) + x_1 \le 1$
and $x_{i_0} > 1 - x(N_1) - x_1$ (respectively, $x_{i_0}\le 1 - x(N_1) - x_1$). We numerically verify that restricted to the appropriate domain of values, either function is non-negative, and so \eqref{eqn:exchange_integral} holds.
The lemma is thus proven.
\end{proof}

% \subsection{Lemma $7$ from \cite{BavejaBCNSX18}} \label{sec:lemma_baveja}

\subsection{Details of Patience $3 \le \ell < 120$} \label{sec:patience_mid}
For $\ell \ge 3$, let us define $y_c := (\ell -1)/x(N_0)$ for convenience,
and recall that $x(N_0) = \ell - x_1 - x(N_1)$.
Then, for $y \in [0, y_c)$, we can apply \eqref{eqn:poisson_regime}
to ensure that
\begin{equation} \label{eqn:poisson_regime_restated}
    \mb{P}[L_0 \le \ell -1 \mid Y_1 =y] \ge  \sum_{k=0}^{\ell -1} \frac{ e^{-y (\ell - x_1 - x(N_1))} y^k (\ell - x(N_1) - x_1)^{k}}{k!}
\end{equation}
Thus, after simplification,
\eqref{eqn:splitting_argument} is lower bounded by
\begin{equation} \label{eqn:availability_lower_bound}
  b(x_1) \int_{0}^{y_c} \sum_{k=0}^{\ell -1} \frac{ e^{-y (\ell - x_1)} y^k (\ell - x(N_1) - x_1)^{k}}{k!} dy
\end{equation}
If $\Gamma(s,z):= \int_{z}^{\infty} \zeta^{s-1} e^{-\zeta} d\zeta$ for $s > 0$ and $z \in \mb{R}$ is the \textit{upper incomplete gamma function} (here $\Gamma(s,0)$ is the usual \textit{gamma function}), then the integral
in \eqref{eqn:availability_lower_bound} has the closed-form expression:

\begin{equation*}
F_{\ell}(x_1, x(N_1)):=\frac{\Gamma(\ell) - \Gamma(\ell, \ell -1) e^{\frac{-(l - 1) x(N_1)}{\ell - x(N_1) - x_1}}-\left(\frac{\ell - x_1 - x(N_1)}{\ell - x(N_1)} \right)^{\ell} \left(\Gamma(\ell) -
    \Gamma(\ell, \ell -1 + \frac{( \ell -1) x(N_1)}{\ell - x(N_1) - x_1}) \right)}{x(N_1) \Gamma(\ell)}.
\end{equation*}
If $\ell =3$, then this simplifies to
$$F_{3}(x_1, x(N_1)) =\frac{2 - 10 e^{\frac{-2 (3 -x(N_1))}{3 - x(N_1) - x_1}}-\left(\frac{3 - x_1 - x(N_1)}{3 - x(N_1)} \right)^{3} \left(2 -
    \Gamma(3, 2 + \frac{2 x(N_1)}{3 - x(N_1) - x_1}) \right)}{2x(N_1) },
$$
which we numerically verified is minimized at $x(N_1) = 1 - x_1$ for any $0 \le x_1 \le 1$.
Otherwise, if $4 \le \ell < 120$, $D_2 F_{\ell}(x_1, x(N_1))$ is non-positive for all $0 \le x_1 \le 1$
and $0 \le x(N_1) \le 1$. Thus, we conclude
that no matter the choice of $x_1$, $F_{\ell}$ is minimized when $x(N_1) = 1- x_1$ as desired.
\begin{proposition} \label{prop:worst_case_patience_mid}
For each $3 \le \ell < 120$, $x_1, x(N_1) \in [0,1]$ with $0 \le x_1 + x(N_1) \le 1$,
it holds that $F_{\ell}(x_1, x(N_1)) \ge F_{\ell}(x_1, 1 - x_1)$.    
\end{proposition}

\subsection{Details of Patience $\ell \ge 120$} \label{sec:patience_big}

For this regime of $\ell$, we make use of Bennett's inequality:

\begin{theorem}[Bennett's Inequality -- \cite{Bennett}] \label{thm:bennett's}
Suppose that $Z_1, \ldots , Z_k$ are random variables with finite variance,
and for which $|Z_i - \mb{E}[Z_i]| \le c$ for all $i \in [k]$.
If $Z = \sum_{i=1}^k Z_i$, $\mu = \sum_{i=1}^k \mb{E}[Z_i]$ and $\sigma^2 = \sum_{i=1}^{k} \mb{E}[(Z_i - \mb{E}[Z_i])^2]$,
then for each $t \ge 0$,
\begin{equation}
    \mb{P}[Z \ge \mu + t] \le \exp\left( - \frac{\sigma^2}{c^2} h\left(\frac{b t}{\sigma^2}\right)\right) 
\end{equation}
where $h(s) = (1 +s) \log(1 + s) - s$ for $s \ge 0$. In particular, if $c=1$ and $\sigma^2 \le \mu$,
then 
\begin{equation} \label{eqn:bennett_simplified}
    \mb{P}[Z \ge \mu + t] \le \exp\left( - (\mu +t) \log\left( \frac{\mu+t}{t}\right) +t\right).
\end{equation}
\end{theorem}

Recall that conditional on $Y_1 = y$ for $y \in [0,1]$, $L_{0}$ is a sum of $|N_{0}|$ independent Bernoulli random variables with means $(y x_i)_{i \in N_0}$
 
We can therefore bound its (conditional) variance by $\sum_{i \in N_0} x_i y x_i \le y x(N_0) = \mb{E}[L_0 \mid Y_1 = y]$. Thus, we can apply Bennett's inequality (stated as \eqref{eqn:bennett_simplified} in \Cref{thm:bennett's}) with
parameters $c =1$, $\mu = y x(N_0)$  and $t = \ell - y x(N_1)$, to conclude that,
\begin{align}
    \mb{P}[L_0 < \ell \mid Y_1 =y] &\ge 1- e^{-(\mu +t) \log\left( \frac{\mu +t}{\mu}\right) +t}  = 1- e^{-\ell \log\left(\frac{\ell}{y x(N_0)}\right) - yx(N_0) + \ell}. \label{eqn:bennett_inequality_applied}
\end{align}

Recall that $x(N_0) \le \ell$ by \eqref{eqn:feasibility_constraints_simplified} due to \Cref{lem:p_values}. For any fixed $y \in [0,1]$ and $\ell$,
the right-hand side of \eqref{eqn:bennett_inequality_applied} is minimized when $x(N_0)  =\ell$,
so we may conclude that
\begin{equation} \label{eqn:bennett_inequality_simplified}
\mb{P}[L_0 < \ell \mid Y_1 =y] \ge 1- e^{-\ell(y + \log\left(\frac{1}{y}\right) - 1)}
\end{equation}
Finally, for any fixed $y \in [0,1]$, the right-hand side of \eqref{eqn:bennett_inequality_simplified} is a decreasing function of $\ell$. Thus, since $\ell \ge 120$, 
$$
\mb{P}[L_0 < \ell \mid Y_1 =y] \ge 1- e^{-120(y + \log\left(\frac{1}{y}\right) - 1)},
$$
as claimed in \eqref{eqn:final_opt_lower_bound_large}.

\section{Approximately Solving \ref{LP:social_welfare}} \label{sec:solve_LP}
Given $\eps > 0$, we must efficiently compute a solution to \ref{LP:social_welfare} with value at least $(1- \eps) \LPOPT(G)$. In order to do so, we first take the dual of \ref{LP:social_welfare}.

\begin{align*}\label{LP:sw_dual}
	\tag{D-LP-c}
	&\text{minimize} 
	& \sum_{u \in U} \alpha_{u} + \sum_{u \in U} \ell_u \gamma_u + \sum_{v \in V} \beta_{v}  \\
	&\text{s.t.} 
	& \beta_{v} \ge \srew_{\bm{r},\bm{q}}(\bm{e}, \bm{\ac}) 
	- \sum_{i=1}^{|\bm{e}|} ( q_{e_i}(\ac_i) \alpha_{u_i} + \gamma_{u_i}) 
	\cdot \prod_{j < i} (1 - q_{e_j}(\ac_j))  && \forall v \in V, (\bm{e}, \bm{\ac}) \in \scr{C}_v:  \\
	&&&& \begin{aligned}
		% &\bm{e} = (e_1, \ldots , e_k), &
		% & \bm{a} = (a_1, \ldots , a_k), \\
		& \text{$e_i$} = (u_i,v) & & \forall i \in [|\bm{e}|]
	\end{aligned} \\
	&& \alpha_{u}, \gamma_u \ge 0 && \forall u \in U\\
	&& \beta_{v} \ge 0 && \forall v \in V
\end{align*}

Let us suppose that we are presented $((\alpha_u)_{u \in U}, (\gamma_u)_{u \in U}, (\beta_v)_{v \in V})$, an arbitrary (potentially infeasible) dual solution to \ref{LP:sw_dual} .
Observe that for $(\bm{e}, \bm{\ac}) \in \scr{C}_v$ with $e_{i}=(u_{i},v)$ for $i=1, \ldots ,|\bm{e}|$,
\begin{align}
\srew_{\bm{r},\bm{q}}(\bm{e}, \bm{\ac}) - \sum_{i=1}^{|\bm{e}|} (q_{e_i}(\ac_i) \cdot \alpha_{u_i} + \gamma_{u_i}) \prod_{j < i} (1 - q_{e_j}(\ac_j)) \notag \\
=\sum_{i=1}^{|\bm{e}|} (\erew_{e_i}(\ac_i) - \alpha_{u_i} - \frac{\gamma_{u_i}}{q_{e_i}(\ac_i)}) q_{e_i}(\ac_i) \prod_{j < i} (1 - q_{e_j}(\ac_j)) \label{eqn:objective_dual_rewrite}.
\end{align}
This motivates defining new edge rewards,
where for each $(u,v) \in E$ and $\ac \in \scr{A}$,
\begin{align}
\hat{r}_{u,v}(\ac) := \max\{ r_{u,v}(\ac) - \alpha_{u} - \gamma_{u}/q_{u,v}(\ac), 0\} \label{eqn:new_reward_function_v}
\end{align}
For each $v \in V$, consider the \textit{star graph} input
to the action-reward problem on $\partial(v)$ with patience $\ell_v$, action space $\scr{A}$, edge rewards $\hat{\bm{r}}^{(v)} = (\hat{r}_{u,v}(a))_{u \in U, a \in \scr{A}}$
and edge probabilities $\bm{q}^{(v)} = (q_{u,v}(a))_{u \in U, a \in \scr{A}}$. On such a restricted
type of input, the goal of this problem is to solve the following maximization problem: 
\begin{align} 
	&\text{maximize} \quad \srew_{\hat{\bm{r}}^{(v)}, \bm{q}^{(v)}}(\bm{e}, \bm{\ac})\label{eqn:demand_oracle}\\
	&\text{subject to} \quad  (\bm{e}, \bm{\ac}) \in \scr{C}_v  \nonumber
\end{align}
Notice that if we could efficiently solve \eqref{eqn:demand_oracle} for each $v \in V$,
then due to \eqref{eqn:objective_dual_rewrite}, this would immediately yield a \textit{separation oracle} for \ref{LP:sw_dual}.

It is well known that one can use ellipsoid method with a separation oracle for \ref{LP:sw_dual}
to then solve \ref{LP:social_welfare} efficiently (see Chapter $4$ of \cite{jens2012}).
Unfortunately, unless the action space satisfies $|\scr{A}| =1$ (see \cite{purohit19a,brubach2025star,borodin2025online}), we cannot expect to solve \eqref{eqn:demand_oracle} \textit{exactly}, due to the NP-hardness
result of \citet{agrawal2020}. 

Fortunately, there exists a generalization of this approach for covering and packing LP's provided in Theorem 5.1 of \citet{Jansen03}. In particular, to get a multiplicative approximation of \ref{LP:social_welfare}, it suffices to design
a \textit{(strong) approximate separation oracle} for \ref{LP:sw_dual}.
% This is implied by Theorem 5.1 of \cite{Jansen03}, and we include the details in \Cref{app:solve_LP}.
We argue in \Cref{app:LP_round} that if we can compute a $(1- \eps)$-approximation of \eqref{eqn:demand_oracle} for each $v \in V$, then we can design an approximate separation oracle,
and thus get a $(1- \eps)$-approximate solution to \ref{LP:social_welfare}.

\begin{proposition}[Proof in \Cref{pf:prop:reduction_to_dual_approximation}] \label{prop:reduction_to_dual_approximation}
Suppose that for every $v \in V$, we can compute a $(1-\eps)$-approximate solution
to the maximization problem in \eqref{eqn:demand_oracle} in time $O( f(1/\eps) \cdot \poly(|U|, |\scr{A}|))$,
where $f$ is an arbitrary function of $\eps$. Then, a solution to \ref{LP:social_welfare} with value at least $(1-\eps) \LPOPT(G)$ can be computed in time $O( f(1/\eps) \cdot \poly(|G|, |\scr{A}|))$.
\end{proposition}

In \Cref{sec:santa_claus_approximation}, we approximately solve the star graph problem for an arbitrary input in the proposed runtime of \Cref{prop:reduction_to_dual_approximation}.  This suffices to get the required approximate solution to \ref{LP:social_welfare} in the runtime claimed in \Cref{thm:main_theorem}.

\subsection{Approximately Solving the Star Graph Problem} \label{sec:santa_claus_approximation}

In this section, we provide an approximation scheme
for the star graph problem. We do this by providing
a reduction to the Santa Claus problem (see \Cref{def:santa_claus}).
The techniques in this section are mostly adapted from Section 3.5 of \citet{segev2025efficient},
however we include all of the proofs for completeness.

We first restate the definition of the star graph problem in a slightly more convenient
notation and terminology than that which was used previously. An instance of this problem consists of: 

\begin{itemize}
    \item A star graph with $n$ edges, where each edge $e \in [n]$ has a \textit{reward} $r_e(a)$ for each action $a \in \mathcal{A}$ independently with probability $q_e(a)$ and otherwise $0$.

    \item A \textit{budget/patience} constraint $\ell \leq n$ limiting the number of edges that can be \textit{queried}.
   
\end{itemize}

We define a \textit{policy} to
be a pair $(\bm{e},\bm{a})$ where:

\begin{itemize}
\item $\bm{e} = (e_1, \ldots , e_{\ell})$ specifies the edges queried, and the order this is done.
\item $\bm{a} =(a_1, \ldots , a_{\ell})$ specifies the actions set for each edge of $\bm{e}$.
\end{itemize}

The \textit{expected reward} of policy $(\bm{e},\bm{a})$ is then 

\begin{align}
\label{reward}
\text{val}_{\bm{r,q}}(\bm{e}, \bm{a}) = \sum_{i=1}^{\ell} r_{{e_i}}(a_{e_i}) \cdot q_{{e_i}}(a_{e_i}) \cdot \prod_{j=1}^{i-1} \left(1 - q_{{e_j}}(a_{e_j})\right).
\end{align}
The objective is to find a policy $(\bm{e}^*, \bm{a}^*)$ that maximizes \eqref{reward}. We denote $\OPT$ 
as the value of an optimal policy, and approximately solve this problem:

\begin{theorem}\label{thm:eptas-threshold-policy}
For the star graph problem, there exists an efficient polynomial-time approximation scheme (EPTAS): For any $\varepsilon > 0$, there exists an algorithm that runs in time $O\left(\left(\frac{1}{\varepsilon^2}\right)^{O(1/\varepsilon)} \cdot \textit{poly}(n,|\mathcal{A}|)\right)$ and outputs a policy $(\bm{e},\bm{a})$ with $\text{val}_{\bm{r,q}}(\bm{e}, \bm{a}) \ge (1 - 7\varepsilon) \cdot \textit{OPT}$.
\end{theorem}

\paragraph{Outline.}
We prove \Cref{thm:eptas-threshold-policy} via three main steps:
\begin{enumerate}
\item \textit{Parameter Estimation and Bucketing:} Given an estimate $\mathcal{E}$ of the optimal value, partition the optimal policy into a constant number of buckets based on reward contribution patterns.
  
\item \textit{Santa Claus Reduction:} Define load contributions that capture how much each edge/action pair can contribute to each bucket's value requirement. Formulate an integer program that assigns edge/action pairs to buckets while satisfying capacity and load constraints.
\item \textit{Policy Reconstruction:} Convert the bucket assignment back to a policy and optimize actions via dynamic programming.
\end{enumerate}

\subsubsection{Step 1: Parameter Estimation and Bucketing}

We first assume that $1/\epsilon $ is an integer, and that we have an estimate $\mathcal{E}$ for which $(1-\epsilon) \cdot \text{OPT} \leq \mathcal{E} \leq \text{OPT}$. While the first assumption is clearly without loss, in order to see the second, first write \ref{LP:social_welfare} with $|U| =1$. In our current notation, we get the following LP:
\begin{align}\label{LP:bansal_single}
	\tag{LP-C-s}
	&\text{maximize} &  \sum_{e \in [n], a \in \scr{A}} r_{e}(a) q_{e}(a) z_{e}(a)\\
	&\text{subject to} & \sum_{e \in [n], a \in \scr{A}} q_{e}(a) z_{e}(a)  \leq 1 \\
        && \sum_{e \in [n], a \in \scr{A}} z_{e}(a) \leq \ell \\
        && \sum_{ a \in \scr{A}} z_{e}(a) \le 1 && \forall e \in [n] \\
	&&z_{e}(a) \ge 0 && \forall e \in [n], a \in \scr{A}
\end{align}
Suppose that $\MLPOPT$ denotes the optimal value of \ref{LP:bansal_single}. Clearly, $\OPT \le \MLPOPT$. On the other hand, after solving for an optimal solution $(z_{e}(a))_{e \in [n], a \in \scr{A}}$ to \ref{LP:bansal_single}, \Cref{thm:reduction_to_trs}
implies that we can use an $\alpha$-selectable P-RCRS (see \Cref{def:patience_rcrs})
to get policy with expected reward at least $\alpha \sum_{e \in [n], a \in \scr{A}} r_{e}(a) q_{e}(a) z_{e}(a)$. Thus, since \Cref{thm:trs} guarantees
the existence of a $\frac{1}{27} (19 - \frac{67}{\mathrm{e}^3})$-selectable P-RCRS where $\frac{1}{27} (19 - \frac{67}{\mathrm{e}^3}) > 1/2$, we know that 
\begin{equation} \label{eqn:estimate_trs}
\MLPOPT/2 \le \OPT \le \MLPOPT.
\end{equation}
We can therefore try all powers of $(1 + \epsilon)$ in the interval $[\MLPOPT/2, \MLPOPT]$ and at least one of them must be an $(1+\epsilon)$-approximation of $\OPT$. This justifies our second
assumption involving $\mathcal{E}$. 

Now, for any subset of ordered edges $\bm{e} =(e_1, \ldots , e_{\ell})$, let $R(\bm{e})$ denote the expected reward obtained by an algorithm that uses the \textit{optimal} strategy on $\bm{e}$. This quantity can be computed in polynomial time via dynamic programming by considering the optimal actions for each position. Specifically, for the optimal policy $(\bm{e}^*, \bm{a}^*)$, we can recursively define the \emph{future value function} as follows. Let $R^*_i$ be the maximum expected value achievable from positions $i, i+1, \ldots, \ell$ when following the optimal policy:
\begin{align}
R^*_{\ell+1} &= 0\\
R^*_i &= \max_{a \in \scr{A}} \left\{r_{e^*_i}(a) \cdot q_{e^*_i}(a) + R^*_{i+1} \cdot (1 - q_{e^*_i}(a))\right\}
\end{align}
where the maximum is taken over all possible actions $a$ for position $i$. The optimal expected value is $\text{OPT} = R^*_1$.

Intuitively, $R^*_i$ represents the expected value we can achieve starting from position $i$ in the optimal ordering, assuming we use optimal actions for all remaining positions. The sequence $R^*_1 \geq R^*_2 \geq \ldots \geq R^*_{\ell} \geq R^*_{\ell+1} = 0$ is non-increasing, as each position can only decrease the future value by potentially accepting a suboptimal offer.

Let position $i \in [\ell]$ be a \emph{jump} if $R^*_i - R^*_{i+1} \geq \epsilon \cdot R_1^*$.
Since $R^*_1 = \text{OPT} \geq \mathcal{E}$ and each jump contributes at least $\epsilon \cdot \mathcal{E}$ to the total decrease, we can have at most $\mathcal{E}/(\epsilon \cdot \mathcal{E}) = 1/\epsilon$ jumps.

\paragraph{Bucket Construction.}

Let $\tilde{t}_1 < \tilde{t}_2 < \ldots < \tilde{t}_K$ be the positions of reward jumps, where $K \leq 1/\epsilon$. We partition the positions $[\ell]$ into $M = 2K + 1$ buckets:

\begin{itemize}
\item Stable Bucket $B_1$: Positions $\{\tilde{t}_K + 1, \tilde{t}_K + 2, \ldots, \ell\}$ (after the last jump)
\item Jump Bucket $B_2$: Position $\{\tilde{t}_K\}$ (the last jump position)  
\item Stable Bucket $B_3$: Positions $\{\tilde{t}_{K-1} + 1, \tilde{t}_{K-1} + 2, \ldots, \tilde{t}_K - 1\}$
\item Jump Bucket $B_4$: Position $\{\tilde{t}_{K-1}\}$
\item And so on...
\end{itemize}

For each bucket $B_i$ we define $\text{BaseVal}(B_i)$ as the minimum future value $R^*_j$ that corresponds to any position appearing immediately after this bucket in the optimal policy $(\bm{e}^*, \bm{a}^*)$. Formally, 
$\text{BaseVal}(B_i) = R^*_{\max\{j : j \text{ is a position in bucket } B_i\} + 1}$
where we set $R^*_{\ell+1} = 0$ by convention. This means that $\text{BaseVal}(B_1) = 0$, $\text{BaseVal}(B_2) = R^*_{\tilde{t}_K + 1}$, $\text{BaseVal}(B_3) = R^*_{\tilde{t}_K}$, and so forth.

Similarly, we define $\text{DeltaVal}(B_i)$ as the total value contribution of bucket $B_i$, which equals the difference between the future value before and after processing this bucket:
$\text{DeltaVal}(B_i) = \text{BaseVal}(B_{i+1}) - \text{BaseVal}(B_i)$.

\paragraph{Parameter Guessing.}

Since we do not know the optimal policy, we guess approximations to these parameters:

For each bucket $B_i$, we guess:
\begin{itemize}
\item $\text{BaseGuess}(B_i) \in \{0, \epsilon^2 \mathcal{E}, 2\epsilon^2 \mathcal{E}, \ldots, \mathcal{E}\}$
\item $\text{DeltaGuess}(B_i) \in \{0, \epsilon^2 \mathcal{E}, 2\epsilon^2 \mathcal{E}, \ldots, \mathcal{E}\}$
\end{itemize}

such that:
\begin{align*}
\text{BaseVal}(B_i) - \epsilon^2 \mathcal{E} \leq \text{BaseGuess}(B_i) &\leq \text{BaseVal}(B_i) \\
\text{DeltaVal}(B_i) - \epsilon^2 \mathcal{E} \leq \text{DeltaGuess}(B_i) &\leq \text{DeltaVal}(B_i) 
\end{align*}

The total number of parameter combinations to try is $(1/\epsilon^2)^{O(1/\epsilon)}$.

\subsubsection{Step 2: Santa Claus Reduction}
At this step, we reduce our problem to an allocation problem, where we allocate the variables to buckets such that each bucket has a load of at least DeltaGuess, with a small error. To solve this, we use another problem called the Santa Claus problem, which was studied in \citet{segev2025efficient}.  Here is the definition:
\begin{definition} \label{def:santa_claus}
An instance of the \textit{Santa Claus} problem consists of:
\begin{itemize}
\item A set of $m$ unrelated machines, where each machine $i \in [m]$ has:
    \begin{itemize}
    \item An upper bound $\ell_i \in \mathbb{N}$ on the number of jobs it can process
    \item A lower bound $L_i \in \mathbb{R}_{+}$ on its required load
    \end{itemize}
\item A collection of $n$ jobs, where assigning job $j \in [n]$ to machine $i \in [m]$ incurs a scalar load $c_{ij} \in \mathbb{R}_{+}$ 

\item A global budget $\ell \in \mathbb{N}$ on the total number of assigned jobs

\end{itemize}
A feasible assignment is a function $\phi: [n] \to [m] \cup \{\emptyset\}$ such that:
\begin{enumerate}
\item $|\{j \in [n] : \phi(j) = i\}| \leq \ell_i$ for all $i \in [m]$ (capacity constraints)
\item $\sum_{j: \phi(j) = i} c_{ij} \geq L_i$ for all $i \in [m]$ (load constraints)
\item $\sum_{j \in [n]} \mathbf{1}_{\phi(j) \neq \emptyset} \leq \ell$ (budget constraint)
\end{enumerate}
\end{definition}
Assume without loss of generality that $L_i = 1$ (by scaling) and $c_{ij} \in [0,1]$ for all $i,j$.
\begin{align}
\text{(SC-IP)} \quad &\\
x_{ij} &\in \{0,1\} && \forall i \in [m], j \in [n] \tag{SC-1}\\
\sum_{i \in [m]} x_{ij} &\leq 1 && \forall j \in [n] \tag{SC-2}\\
\sum_{j \in [n]} x_{ij} &\leq {\ell}_i && \forall i \in [m] \tag{SC-3}\\
\sum_{j \in [n]} c_{ij} x_{ij} &\geq 1 && \forall i \in [m] \tag{SC-4}\\
\sum_{i \in [m]} \sum_{j \in [n]} x_{ij} &\leq \ell && \tag{SC-5} 
\end{align}
\begin{remark}
The EPTAS of ~\citet{segev2025efficient} was designed for the more general \emph{Multi-Dimensional Santa Claus} problem, where loads and lower bounds are $D$-dimensional vectors. In our application, we only require the single-dimensional case ($D = 1$), for which the same techniques apply and simplify considerably.
It also does not include the global budget constraint (SC-5). However, their techniques generalize to account for this, allowing us to ensure
that this constraint can be enforced with probability $1$. This is because \cite{segev2025efficient} 
introduce a ``strengthened'' LP which partitions the jobs $j$ based on their load $(c_{i,j})_{i \in [m]}$ -- these are called the load classes. They then guess the number of jobs from any load class assigned to any machine. By restricting to guesses in which these load class guesses satisfy the global budget constraint, their algorithm outputs a solution
for which (SC-5) holds.
\end{remark}

\begin{theorem}[Adapted from Theorem 2.1 of~\cite{segev2025efficient}]\label{thm:santa_claus}
For any $\epsilon > 0$, when (SC-IP) is feasible, there exists an algorithm running in time $O(2^{m \log(1/\epsilon)} \cdot \text{poly}(n))$ that computes $\{0,1\}$-valued random variables $(X_{i,j})_{i \in [m], j \in [n]}$ such that with probability at least $1/2$: 
% \andisheh{repeating for $1/\epsilon$ times?}
\begin{enumerate}
\item Constraints (SC-1), (SC-2), (SC-3), and (SC-5) are satisfied exactly.
\item Constraint (SC-4) is $\epsilon$-violated: $\sum_{j \in [n]} c_{ij}X_{ij} \geq 1 - \epsilon$ for all $i \in [m]$.
\end{enumerate}

\end{theorem}

To reduce to the Santa Claus problem with Global Budget, we need:
\begin{itemize}
    \item $n$ \textit{jobs} that correspond to the edges. We also have $2K+1$ \textit{machines} corresponding to the buckets. 
    \item For variable $j \in [n]$ assigned to bucket $B_i \in [2K+1]$ with baseline future value $\text{BaseGuess}(B_i)$, we set the \textit{assignment load} to be 
    $$c_{ij} = \max_{a } \left\{[(r_j(a) - \text{BaseGuess}(B_i)) \cdot q_j(a)]^+\right\}$$

    \item Jump machines can be assigned at most one job, and in total we can assign $\ell$ jobs to all the machines.
    \item Each machine $i \in [2K+1]$ has a lower bound of $\text{DeltaGuess}(B_i)$ on its total load.
\end{itemize}
We formulate the variable-to-bucket assignment as:
\begin{align}
\text{(IP-TH)} \quad &\\
&\xi_{ij} \in \{0,1\} &&\forall i \in [M], j \in [n] \label{eq:binary}\\
&\sum_{i=1}^M \xi_{ij} \leq 1 &&\forall j \in [n] \label{eq:assignment}\\
&\sum_{j=1}^n \xi_{ij} \leq 1 &&\forall i \in [M] : B_i \text{ is a jump bucket} \label{eq:capacity}\\
&\begin{multlined}[t]
\sum_{j=1}^n  \max_{a} \left\{[(r_j(a) - \text{BaseGuess}(B_i)) \cdot q_j(a)]^+\right\}\xi_{ij} \\
\geq \text{DeltaGuess}(B_i)
\end{multlined} &&\forall i \in [M] \label{eq:load}\\
&\sum_{i=1}^M \sum_{j=1}^n \xi_{ij} \leq \ell &&\label{eq:budget}
\end{align}
In order to use \Cref{thm:santa_claus}, we need to prove that the IP above admits a feasible solution.
\begin{lemma}[Proof in \Cref{pf:lem:feasible_IP_solution}]\label{lem:feasible_IP_solution}
    IP-TH admits a feasible solution. 
\end{lemma}

\begin{corollary}[of \Cref{thm:santa_claus}]
\label{corr:eptas}
% \CMc{Is this $\bm{x}$ or $\xi$?}
    There is an EPTAS for computing a cardinality-feasible variable-to-machine assignment $\bm{\xi} \in \{0,1\}^{(2K+1)\times n}$ such that every machine $i \in [2K+1]$ receives a total load of $\sum_{j \in [n]}c_{ij}\xi_{ij} \geq (1 - \epsilon) \text{DeltaGuess}(B_i)$.
\end{corollary}
\subsubsection{Step 3: Policy Reconstruction}
Given the assignment vector $\xi$ from Corollary \ref{corr:eptas}, we construct a policy $(\bm{e}, \bm{a})$ as follows:
We define the ordering by processing buckets in reverse order. Specifically:
 The last edges to be queried are those assigned to bucket $B_1$ (corresponding to the stable bucket after the last jump). Just before them, we query those assigned to bucket $B_2$ (the last jump bucket). Continuing this pattern through all buckets $B_3, B_4, \ldots, B_{2K+1}$.  Within each bucket, the internal ordering of edges is arbitrary.

Let $e$ be the resulting ordering of the at most $\ell$ selected edges.
 We compute the optimal actions and expected future values for our ordering $e$ using dynamic programming. Let $\tilde{R}_i$ denote the maximum expected value achievable from positions $i, i+1, \ldots, \ell$ when using optimal actions. We compute:
% \CMc{Expectations here.}
\begin{align*}
\tilde{R}_{\ell+1} &= 0\\
\tilde{R}_i &= \max_{a \in \scr{A}} \left\{r_{e_i}(a) \cdot q_{e_i}(a) + \tilde{R}_{i+1} \cdot (1 - q_{e_i}(a))\right\}
\end{align*}
for $i = \ell, \ell-1, \ldots, 1$. The optimal action for position $i$ is:
$${a}_i = \arg\max_{a \in \scr{A}} \left\{r_{e_i}(a) \cdot q_{e_i}(a) + \tilde{R}_{i+1} \cdot (1 - q_{e_i}(a))\right\}$$

The expected value of our reconstructed policy is $\text{val}_r(\bm{e}, \bm{a}) = \tilde{R}_1$.
 We say our policy is \textit{behind schedule} at position $i$ if $\tilde{R}_{i+1} < \text{BaseGuess}(B^{-1}_{e_i})$, where $B^{-1}_{e_i}$ denotes the bucket to which edge $e_i$ is assigned. Otherwise, the policy is \textit{ahead of schedule}.

Let $i_{\min}$ be the first position where we are ahead of schedule (this is well-defined since we are always ahead of schedule at the last position, as $\tilde{R}_{\ell+1} = 0 = \text{BaseGuess}(B_1)$). By the recursive definition, we can decompose the total expected value as:
$$\text{val}_{r,q}(e, a) = \tilde{R}_1 = \tilde{R}_{i_{\min}} + \sum_{i=1}^{i_{\min}-1} (\tilde{R}_i - \tilde{R}_{i+1})$$
\begin{lemma}[Proof in \Cref{pf:lem:ahead-schedule}]\label{lem:ahead-schedule}
The following inequality holds:
$$\tilde{R}_{i_{\min}} \geq \text{BaseVal}(B^{-1}_{e_{i_{\min})}}) + (1-\epsilon) \cdot \text{DeltaGuess}(B^{-1}_{e_{i_{\min}}}) - 2\epsilon \cdot \text{OPT}.$$
\end{lemma}

\begin{lemma}[Proof in \Cref{pf:lem:behind-schedule}] \label{lem:behind-schedule}
The following inequality holds:
$$\sum_{i=1}^{i_{\min}-1} (\tilde{R}_i - \tilde{R}_{i+1}) \geq (1-\epsilon) \cdot \sum_{j > B^{-1}_{e_{i_{\min}}}} (\text{BaseVal}(B_{j+1}) - \text{BaseVal}(B_j)) - 3\epsilon \cdot \text{OPT}.$$
\end{lemma}
We can now prove that the reconstructed policy $(\bm{e}, \bm{a})$ achieves an expected reward of at least $(1 - 7\epsilon) \cdot \text{OPT}$.
\begin{proof}[Proof of \Cref{thm:eptas-threshold-policy}]
Combining Lemmas~\ref{lem:ahead-schedule} and~\ref{lem:behind-schedule},
\begin{align*}
\text{val}_{\bm{r,q}}(\bm{e}, \bm{a}) &\geq \text{BaseVal}(B^{-1}_{e_{i_{\min}}}) + (1-\epsilon) \cdot \text{DeltaGuess}(B^{-1}_{e_{i_{\min}}}) - 2\epsilon \cdot \text{OPT}\\
&\quad + (1-\epsilon) \cdot \sum_{j > B^{-1}_{e_{i_{\min}}}} (\text{BaseVal}(B_{j+1}) - \text{BaseVal}(B_j)) - 3\epsilon \cdot \text{OPT}\\
&\geq (1 - 7\epsilon) \cdot \text{OPT},
\end{align*}
where the final inequality holds for $\eps$ sufficiently small.
\end{proof}

\section{Putting it all together: Proving \Cref{thm:main_theorem}} \label{sec:final_proof}
Suppose we are given a graph $G=(U,V,E)$
with actions $\scr{A}$, edge rewards $\bm{r} =(r_{e}(\ac))_{e \in E, \ac \in \scr{A}}$ and edge probabilities $\bm{q} =(q_{e}(\ac))_{e \in E, \ac \in \scr{A}}$, and recall that $\beta := 1-1/e$ if $\ell_u \in \{1, \infty\}$ for all $u \in U$,
and $\beta := {27} (19 - \frac{67}{e^3})$ otherwise.
For each $\eps > 0$, we design a policy whose expected reward is at least $(1- \eps) \cdot \beta \cdot \OPT(G)$,
where $\OPT(G)$ is the expected reward of the optimal policy.

For any $v \in V$, we can apply \Cref{thm:eptas-threshold-policy},
to the maximization problem \eqref{eqn:demand_oracle}, and solve it up to a multiplicative factor of $(1 - \eps)$. 
Thus, by \Cref{prop:reduction_to_dual_approximation}, we can compute a solution to \ref{LP:social_welfare},
say $(z_{v}(\bm{e},\bm{\ac}))_{v \in V, (\bm{e},\bm{\ac}) \in \scr{C}_v}$, for
which
\begin{equation} \label{eqn:approx_compute_restate}
     \sum_{v \in V} \sum_{(\bm{e},\bm{\ac}) \in \scr{C}_v } \srew(\bm{e}, \bm{\ac})  z_{v}(\bm{e}, \bm{\ac}) \ge (1 - \eps) \LPOPT(G),
\end{equation}
where $\LPOPT(G)$ is the optimal value of \ref{LP:social_welfare}.  Now, by \Cref{thm:trs},
there is a $\beta$-selectable P-RCRS $\psi_u$ for each $u \in U$. By passing $(z_{v}(\bm{e},\bm{\ac}))_{v \in V, (\bm{e},\bm{\ac}) \in \scr{C}_v}$ and these P-RCRS's $(\psi_u)_{u \in U}$ to \Cref{alg:social_welfare}, \Cref{thm:reduction_to_trs} implies
that we get a policy whose expected reward is at least
\begin{equation} \label{eqn:prcrs_restate}
    \beta \cdot \sum_{v \in V} \sum_{(\bm{e},\bm{\ac}) \in \scr{C}_v } \srew(\bm{e}, \bm{\ac})  z_{v}(\bm{e}, \bm{\ac}) \ge \beta (1 - \eps) \LPOPT(G) \ge (1 - \eps) \OPT(G), 
\end{equation}
where first inequality applies \eqref{eqn:approx_compute_restate}, and the second inequality applies \Cref{thm:LP_relaxation}.
Thus, \eqref{eqn:prcrs_restate} implies an approximation ratio of $(1- \eps) \cdot \beta$ as claimed in \Cref{thm:main_theorem}. Finally,
the runtime follows due to the runtime of \Cref{thm:eptas-threshold-policy}, \Cref{prop:reduction_to_dual_approximation}, and the fact that the remaining steps (i.e., the P-RCRS's of \Cref{thm:trs} and \Cref{alg:social_welfare}) are polynomial time, and in fact do not depend on $\eps$.

\section{Details of Reductions} 
\label{sec:reductions}
In this section, we show how both the sequential pricing problem and the free-order prophet matching problem can be reduced to our \emph{action-reward query-commit matching problem}.

\paragraph{Sequential pricing.}
In the sequential pricing problem introduced by \citet{DBLP:journals/corr/abs-2205-08667}, the input is a bipartite graph, say $G=(U,V,E)$, in which  $U$ are the \emph{jobs}, each of which has a known value $b_u \ge 0$ to the platform for being completed, and $V$ are the \textit{workers}. The platform may query an edge $e=(u,v)$ by offering a worker $v \in V$ a payment $\tau \ge 0$ to perform job $u$. The offer is then accepted independently with a known probability $p_{u,v}(\tau)$. As in our setting, each vertex is associated with a known patience constraint limiting the number of incident payment offers, and each $e=(u,v) \in E$ may be queried via at most one payment. The platform’s objective is to maximize either its expected \emph{revenue} or its expected \emph{social welfare}.

This problem naturally reduces to our \emph{action-reward query-commit matching problem} by appropriately defining the action space. Specifically, for each edge $e=(u,v)$, we introduce an action $a \in \mathcal{A}$ for every payment $\tau \ge 0$. Querying edge $e$ using the action corresponding to $\tau$ succeeds with probability $q_e(a)=p_{u,v}(\tau)$. The reward associated with an action depends on the optimization objective. For the revenue objective, the reward of the action  $a$ corresponding to payment $\tau$ is the platform’s revenue $r_e(a)= b_u - \tau$, where $b_u$ denotes the value of job $u$. For the social welfare objective, the reward is given by
$r_e(a)=b_u - \mathbb{E}\!\left[C_{u,v} \mid C_{u,v} \le \tau\right]$,
where $C_{u,v}$ is a random variable representing the cost incurred by worker $v$ for performing job $u$, drawn from a known distribution.

\paragraph{Free-order prophet matching.}
In the \emph{free-order prophet matching problem}, each edge $e$ of a graph $G$ has a weight $W_e$ drawn independently from a known distribution $\mathcal{D}_e$. An online algorithm processes the edges sequentially in an order of its choosing; upon processing an edge $e$, it observes the realization of $W_e$ and makes an irrevocable decision on whether to add $e$ to the current matching. The objective is to maximize the total weight of the resulting matching.

It is known that this problem admits an optimal threshold policy: before processing any edge $e$, the policy selects a threshold $\tau_e$ and includes $e$ in the matching if and only if $W_e \ge \tau_e$. Moreover, it suffices to consider thresholds $\tau_e$ drawn from the support of $\mathcal{D}_e$. To reduce this problem to our action-reward query-commit matching framework, we introduce, for each edge $e$, an action $a \in \mathcal{A}$ corresponding to each feasible threshold $\tau_e$. Querying edge $e$ using the action associated with threshold $\tau_e$ succeeds with probability
\[
q_e(a) = \mb{P}[W_e \ge \tau_e],
\]
and yields a reward
\[
r_e(a) = \mathbb{E}[W_e \mid W_e \ge \tau_e].
\]
We can also capture the revenue objective for this problem in a similar way.

This concludes our reduction of both the sequential pricing problem and the free-order prophet matching problem to the action-reward query-commit matching problem.

\bibliographystyle{amsalpha}
\bibliography{ref}

@inproceedings{DBLP:conf/soda/Derakhshan023,
  author       = {Mahsa Derakhshan and
                  Alireza Farhadi},
  editor       = {Nikhil Bansal and
                  Viswanath Nagarajan},
  title        = {Beating {(1} - 1/e)-Approximation for Weighted Stochastic Matching},
  booktitle    = {Proceedings of the 2023 {ACM-SIAM} Symposium on Discrete Algorithms,
                  {SODA} 2023, Florence, Italy, January 22-25, 2023},
  pages        = {1931--1961},
  publisher    = {{SIAM}},
  year         = {2023},
  url          = {https://doi.org/10.1137/1.9781611977554.ch74},
  doi          = {10.1137/1.9781611977554.CH74},
  bibsource    = {dblp computer science bibliography, https://dblp.org}
}

@inproceedings{huang2025edgeweightedmatchingdark,
      title={Edge-weighted Matching in the Dark}, 
      author={Zhiyi Huang and Enze Sun and Xiaowei Wu and Jiahao Zhao},
booktitle    = { IEEE 66th Annual Symposium on Foundations of Computer Science (FOCS 2025)},
      year={2025},
}

@article{DBLP:journals/corr/abs-2205-08667,
  author       = {Tristan Pollner and
                  Mohammad Roghani and
                  Amin Saberi and
                  David Wajc},
  title        = {Improved Online Contention Resolution for Matchings and Applications
                  to the Gig Economy},
  journal      = {Math. Oper. Res.},
  volume       = {49},
  number       = {3},
  pages        = {1582--1606},
  year         = {2024},
  url          = {https://doi.org/10.1287/moor.2023.1388},
  doi          = {10.1287/MOOR.2023.1388},
  bibsource    = {dblp computer science bibliography, https://dblp.org}
}

@inproceedings{macruryinduction2023,
  author       = {Calum MacRury and
                  Will Ma},
  editor       = {Bojan Mohar and
                  Igor Shinkar and
                  Ryan O'Donnell},
  title        = {Random-Order Contention Resolution via Continuous Induction: Tightness
                  for Bipartite Matching under Vertex Arrivals},
  booktitle    = {Proceedings of the 56th Annual {ACM} Symposium on Theory of Computing,
                  {STOC} 2024, Vancouver, BC, Canada, June 24-28, 2024},
  pages        = {1629--1640},
  publisher    = {{ACM}},
  year         = {2024},
  url          = {https://doi.org/10.1145/3618260.3649788},
  doi          = {10.1145/3618260.3649788},
  bibsource    = {dblp computer science bibliography, https://dblp.org}
}

@InProceedings{fu_2021,
  author =	{Fu, Hu and Tang, Zhihao Gavin and Wu, Hongxun and Wu, Jinzhao and Zhang, Qianfan},
  title =	{{Random Order Vertex Arrival Contention Resolution Schemes for Matching, with Applications}},
  booktitle =	{48th International Colloquium on Automata, Languages, and Programming (ICALP 2021)},
  pages =	{68:1--68:20},
  series =	{Leibniz International Proceedings in Informatics (LIPIcs)},
  ISBN =	{978-3-95977-195-5},
  ISSN =	{1868-8969},
  year =	{2021},
  volume =	{198},
  editor =	{Bansal, Nikhil and Merelli, Emanuela and Worrell, James},
  publisher =	{Schloss Dagstuhl -- Leibniz-Zentrum f{\"u}r Informatik},
  address =	{Dagstuhl, Germany},
  URL =		{https://drops.dagstuhl.de/opus/volltexte/2021/14137},
  URN =		{urn:nbn:de:0030-drops-141376},
  doi =		{10.4230/LIPIcs.ICALP.2021.68}
}

@InProceedings{purohit19a,
  title = 	 {Hiring Under Uncertainty},
  author =       {Purohit, Manish and Gollapudi, Sreenivas and Raghavan, Manish},
  booktitle = 	 {Proceedings of the 36th International Conference on Machine Learning},
  pages = 	 {5181--5189},
  year = 	 {2019},
  editor = 	 {Chaudhuri, Kamalika and Salakhutdinov, Ruslan},
  volume = 	 {97},
  series = 	 {Proceedings of Machine Learning Research},
  month = 	 {09--15 Jun},
  publisher =    {PMLR},
  url = 	 {https://proceedings.mlr.press/v97/purohit19a.html}
}

@article{brubach2025star,
author = {Brubach, Brian and Grammel, Nathaniel and Ma, Will and Srinivasan, Aravind},
title = {Online Matching Frameworks Under Stochastic Rewards, Product Ranking, and Unknown Patience},
year = {2025},
issue_date = {March-April 2025},
publisher = {INFORMS},
address = {Linthicum, MD, USA},
volume = {73},
number = {2},
issn = {0030-364X},
url = {https://doi.org/10.1287/opre.2021.0371},
doi = {10.1287/opre.2021.0371},
journal = {Oper. Res.},
month = mar,
pages = {995–1010},
numpages = {16}
}

@article{brubach2024offline,
  author    = {Brian Brubach and
               Nathaniel Grammel and
               Will Ma and
               Aravind Srinivasan},
  title     = {Improved Guarantees for Offline Stochastic Matching via new Ordered
               Contention Resolution Schemes},
  journal = {Mathematics of Operations Research 0(0)},
  year      = {2024},
  url       = {https://doi.org/10.1287/moor.2022.0256},
}

@inproceedings{macrury_contention_2023,
  author    = {Calum MacRury and
               Will Ma and
               Nathaniel Grammel},
  editor    = {Nikhil Bansal and
               Viswanath Nagarajan},
  title     = {On (Random-order) Online Contention Resolution Schemes for the Matching
               Polytope of (Bipartite) Graphs},
  booktitle = {Proceedings of the 2023 {ACM-SIAM} Symposium on Discrete Algorithms,
               {SODA} 2023, Florence, Italy, January 22-25, 2023},
  pages     = {1995--2014},
  publisher = {{SIAM}},
  year      = {2023},
  url       = {https://doi.org/10.1137/1.9781611977554.ch76},
  doi       = {10.1137/1.9781611977554.ch76},
  bibsource = {dblp computer science bibliography, https://dblp.org}
}

@inproceedings{karp1981maximum,
  title={Maximum matching in sparse random graphs},
  author={Karp, Richard M and Sipser, Michael},
  booktitle={22nd Annual Symposium on Foundations of Computer Science (sfcs 1981)},
  pages={364--375},
  year={1981},
  organization={IEEE}
}

@inbook{chen_2024,
author = {Ziyun Chen and Zhiyi Huang and Dongchen Li and Zhihao Gavin Tang},
title = {Prophet Secretary and Matching: the Significance of the Largest Item},
booktitle = {Proceedings of the 2025 Annual ACM-SIAM Symposium on Discrete Algorithms (SODA)},
chapter = {},
year = {2025},
pages = {1371-1401},
doi = {10.1137/1.9781611978322.42},
URL = {https://epubs.siam.org/doi/abs/10.1137/1.9781611978322.42},
eprint = {https://epubs.siam.org/doi/pdf/10.1137/1.9781611978322.42}
}

@book{jens2012,
author = {Korte, Bernhard and Vygen, Jens},
title = {Combinatorial Optimization: Theory and Algorithms},
year = {2012},
isbn = {3642244874},
publisher = {Springer Publishing Company, Incorporated},
edition = {5th}
}

@misc{segev2025efficient,
      title={Efficient Approximation Schemes for Stochastic Probing and Selection-Stopping Problems}, 
      author={Danny Segev and Sahil Singla},
      year={2025},
      eprint={2007.13121},
      archivePrefix={arXiv},
      primaryClass={cs.DS},
      url={https://arxiv.org/abs/2007.13121}, 
}

@inproceedings{segev2021,
author = {Segev, Danny and Singla, Sahil},
title = {Efficient Approximation Schemes for Stochastic Probing and Prophet Problems},
year = {2021},
isbn = {9781450385541},
publisher = {Association for Computing Machinery},
address = {New York, NY, USA},
url = {https://doi.org/10.1145/3465456.3467614},
doi = {10.1145/3465456.3467614},
booktitle = {Proceedings of the 22nd ACM Conference on Economics and Computation},
pages = {793–794},
numpages = {2},
location = {Budapest, Hungary},
series = {EC '21}
}

@article{Feige2006,
author = {Feige, Uriel},
title = {On Sums of Independent Random Variables with Unbounded Variance and Estimating the Average Degree in a Graph},
journal = {SIAM Journal on Computing},
volume = {35},
number = {4},
pages = {964-984},
year = {2006},
doi = {10.1137/S0097539704447304},

URL = { 
    
        https://doi.org/10.1137/S0097539704447304
    
    

},
eprint = { 
    
        https://doi.org/10.1137/S0097539704447304
    
    

}
}

@misc{paulin2017conjecturessamuelsfeige,
      title={On some conjectures of Samuels and Feige}, 
      author={Roland Paulin},
      year={2017},
      eprint={1703.05152},
      archivePrefix={arXiv},
      primaryClass={math.PR},
      url={https://arxiv.org/abs/1703.05152}, 
}

@article{Bennett,
author = {George Bennett},
title = {Probability Inequalities for the Sum of Independent Random Variables},
journal = {Journal of the American Statistical Association},
volume = {57},
number = {297},
pages = {33--45},
year = {1962},
publisher = {ASA Website},
doi = {10.1080/01621459.1962.10482149},


URL = { 
    
    
        https://www.tandfonline.com/doi/abs/10.1080/01621459.1962.10482149
    

},
eprint = { 
    
    
        https://www.tandfonline.com/doi/pdf/10.1080/01621459.1962.10482149
    

}

}

@article{Samuels1966,
 ISSN = {00034851},
 URL = {http://www.jstor.org/stable/2238704},
 author = {S. M. Samuels},
 journal = {The Annals of Mathematical Statistics},
 number = {1},
 pages = {248--259},
 publisher = {Institute of Mathematical Statistics},
 title = {On a Chebyshev-Type Inequality for Sums of Independent Random Variables},
 urldate = {2025-06-11},
 volume = {37},
 year = {1966}
}

@article{Samuels1969,
  title={The Markov inequality for sums of independent random variables},
  author={Samuels, S. M.},
  journal={The Annals of Mathematical Statistics},
  volume={40},
  number={6},
  pages={1980--1984},
  year={1969},
  publisher={JSTOR}
}

@InProceedings{Lee2018,
  author =	{Euiwoong Lee and Sahil Singla},
  title =	{{Optimal Online Contention Resolution Schemes via Ex-Ante Prophet Inequalities}},
  booktitle =	{26th Annual European Symposium on Algorithms (ESA 2018)},
  pages =	{57:1--57:14},
  series =	{Leibniz International Proceedings in Informatics (LIPIcs)},
  ISBN =	{978-3-95977-081-1},
  ISSN =	{1868-8969},
  year =	{2018},
  volume =	{112},
  editor =	{Yossi Azar and Hannah Bast and Grzegorz Herman},
  publisher =	{Schloss Dagstuhl--Leibniz-Zentrum fuer Informatik},
  address =	{Dagstuhl, Germany},
  URL =		{http://drops.dagstuhl.de/opus/volltexte/2018/9520},
  URN =		{urn:nbn:de:0030-drops-95208},
  doi =		{10.4230/LIPIcs.ESA.2018.57}
}

@article{borodin2025online,
  title={Online bipartite matching in the probe-commit model},
  author={Borodin, Allan and MacRury, Calum},
  journal={Mathematical Programming},
  pages={1--54},
  year={2025},
  publisher={Springer}
}

@inproceedings{agrawal2020,
author = {Agrawal, Shipra and Sethuraman, Jay and Zhang, Xingyu},
title = {On Optimal Ordering in the Optimal Stopping Problem},
year = {2020},
isbn = {9781450379755},
publisher = {Association for Computing Machinery},
address = {New York, NY, USA},
url = {https://doi.org/10.1145/3391403.3399484},
doi = {10.1145/3391403.3399484},
booktitle = {Proceedings of the 21st ACM Conference on Economics and Computation},
pages = {187–188},
numpages = {2},
location = {Virtual Event, Hungary},
series = {EC '20}
}

@inproceedings{Gamlath2019,
author = {Gamlath, Buddhima and Kale, Sagar and Svensson, Ola},
title = {Beating Greedy for Stochastic Bipartite Matching},
year = {2019},
publisher = {Society for Industrial and Applied Mathematics},
address = {USA},
booktitle = {Proceedings of the Thirtieth Annual ACM-SIAM Symposium on Discrete Algorithms},
pages = {2841–2854},
numpages = {14},
location = {San Diego, California},
series = {SODA ’19}
}

@InProceedings{costello2012matching,
author="Costello, Kevin P.
and Tetali, Prasad
and Tripathi, Pushkar",
editor="Czumaj, Artur
and Mehlhorn, Kurt
and Pitts, Andrew
and Wattenhofer, Roger",
title="Stochastic Matching with Commitment",
booktitle="Automata, Languages, and Programming",
year="2012",
publisher="Springer Berlin Heidelberg",
address="Berlin, Heidelberg",
pages="822--833",
isbn="978-3-642-31594-7"
}

@article{GandhiGKSP06,
author = {Gandhi, Rajiv and Khuller, Samir and Parthasarathy, Srinivasan and Srinivasan, Aravind},
title = {Dependent Rounding and Its Applications to Approximation Algorithms},
year = {2006},
issue_date = {May 2006},
publisher = {Association for Computing Machinery},
address = {New York, NY, USA},
volume = {53},
number = {3},
issn = {0004-5411},
url = {https://doi.org/10.1145/1147954.1147956},
doi = {10.1145/1147954.1147956},
journal = {J. ACM},
month = may,
pages = {324–360},
numpages = {37}
}

@Article{BavejaBCNSX18,
author="Baveja, Alok
and Chavan, Amit
and Nikiforov, Andrei
and Srinivasan, Aravind
and Xu, Pan",
title="Improved Bounds in Stochastic Matching and Optimization",
journal="Algorithmica",
year="2018",
month="Nov",
day="01",
volume="80",
number="11",
pages="3225--3252",
issn="1432-0541",
doi="10.1007/s00453-017-0383-4",
url="https://doi.org/10.1007/s00453-017-0383-4"
}

@book{Peari1992,
  title={Convex Functions, Partial Orderings, and Statistical Applications},
  author={Peajcariaac, J.E. and Tong, Y.L.},
  isbn={9780080925226},
  lccn={91034153},
  series={Mathematics in Science and Engineering},
  url={https://books.google.com/books?id=rCAOFpic7AkC},
  year={1992},
  publisher={Elsevier Science}
}

@article{BrubachSSX20,
  author    = {Brian Brubach and
               Karthik Abinav Sankararaman and
               Aravind Srinivasan and
               Pan Xu},
  title     = {Attenuate Locally, Win Globally: Attenuation-Based Frameworks for
               Online Stochastic Matching with Timeouts},
  journal   = {Algorithmica},
  volume    = {82},
  number    = {1},
  pages     = {64--87},
  year      = {2020},
  url       = {https://doi.org/10.1007/s00453-019-00603-7},
  doi       = {10.1007/s00453-019-00603-7},
  bibsource = {dblp computer science bibliography, https://dblp.org}
}

@article{Jansen03,
  author       = {Klaus Jansen},
  title        = {Approximate strong separation with application in fractional graph
                  coloring and preemptive scheduling},
  journal      = {Theor. Comput. Sci.},
  volume       = {302},
  number       = {1-3},
  pages        = {239--256},
  year         = {2003},
  url          = {https://doi.org/10.1016/S0304-3975(02)00829-0},
  doi          = {10.1016/S0304-3975(02)00829-0},
  bibsource    = {dblp computer science bibliography, https://dblp.org}
}

@article{BansalGLMNR12,
  author    = {Nikhil Bansal and
               Anupam Gupta and
               Jian Li and
               Juli{\'{a}}n Mestre and
               Viswanath Nagarajan and
               Atri Rudra},
  title     = {When {LP} Is the Cure for Your Matching Woes: Improved Bounds for
               Stochastic Matchings},
  journal   = {Algorithmica},
  volume    = {63},
  number    = {4},
  pages     = {733--762},
  year      = {2012},
  url       = {https://doi.org/10.1007/s00453-011-9511-8},
  doi       = {10.1007/s00453-011-9511-8},
  bibsource = {dblp computer science bibliography, https://dblp.org}
}

@article{Adamczyk11,
  author    = {Marek Adamczyk},
  title     = {Improved analysis of the greedy algorithm for stochastic matching},
  journal   = {Inf. Process. Lett.},
  volume    = {111},
  number    = {15},
  pages     = {731--737},
  year      = {2011},
}

@inproceedings{Chen,
 author = {Chen, Ning and Immorlica, Nicole and Karlin, Anna R. and Mahdian, Mohammad and Rudra, Atri},
 title = {Approximating Matches Made in Heaven},
 booktitle = {Proceedings of the 36th International Colloquium on Automata, Languages and Programming: Part I},
 series = {ICALP '09},
 year = {2009},
 isbn = {978-3-642-02926-4},
 location = {Rhodes, Greece},
 pages = {266--278},
 numpages = {13},
}

@inproceedings{Adamczyk15,
  author    = {Marek Adamczyk and
               Fabrizio Grandoni and
               Joydeep Mukherjee},
  editor    = {Nikhil Bansal and
               Irene Finocchi},
  title     = {Improved Approximation Algorithms for Stochastic Matching},
  booktitle = {Algorithms - {ESA} 2015 - 23rd Annual European Symposium, Patras,
               Greece, September 14-16, 2015, Proceedings},
  series    = {Lecture Notes in Computer Science},
  volume    = {9294},
  pages     = {1--12},
  publisher = {Springer},
  year      = {2015},
}
\appendix

\section{Proof of \Cref{thm:LP_relaxation} from \Cref{sec:config_LP}} \label{pf:thm:LP_relaxation}

    We introduce a relaxed action-reward problem where: (1) vertices in $U$ must be matched at most once in expectation, rather than at most once with probability 1, and (2) each vertex $u \in U$ has an expected patience constraint of at most $\ell_u$ queries, rather than a hard limit of $\ell_u$ queries.

Let $\overline{\text{OPT}}(G)$ be the optimal value of this relaxed problem. Since this is a relaxation of the original constraints, we have $\text{OPT}(G) \leq \overline{\text{OPT}}(G)$.
Therefore, to complete the proof, we will show that $\overline{\text{OPT}}(G) \leq \text{LP}(G)$.

Let $\textsc{OptAlg}$ be an optimal policy for the relaxed problem. We construct a \textit{vertex-iterative} policy $\textsc{Alg}$ for the relaxed problem that achieves the same expected reward as $\textsc{OptAlg}$. In other words, $\textsc{Alg}$ will process vertices in $V$ one by one in an arbitrary order, dealing completely with all edges incident to a vertex before moving to the next.

For each vertex $v \in V$, $\textsc{Alg}$ runs a simulation of $\textsc{OptAlg}$ where for any edge $f \notin \partial(v)$ and action $a \in \mathcal{A}$, it draws a simulated Bernoulli bit $\til{Q}_f(a) \sim \text{Ber}(q_f(a))$ and gives it to $\textsc{OptAlg}$ as the real query outcome.
Based on this simulation, $\textsc{Alg}$ determines the adaptive ordering of edges and specific actions in $\partial(v)$ that $\textsc{OptAlg}$ would use, without actually querying edges incident to other vertices. It then follows this induced adaptive policy when processing $v$. After processing $v$, $\textsc{Alg}$ discards all simulated information and moves to the next vertex. 

Let $p_v(\bm{e},\bm{a})$ be the probability that $\textsc{OptAlg}$ uses the sequence-action pair $(\bm{e},\bm{a}) \in \scr{C}_v$ when processing vertex $v$. Since $\textsc{Alg}$ faithfully simulates the behavior of $\textsc{OptAlg}$ with respect to the marginal distributions of the outcomes, it induces the same distribution over sequence-action pairs for each vertex $v \in V$. Therefore, the distributions over matched and queried edges are identical between $\textsc{Alg}$ and $\textsc{OptAlg}$.

We now verify that $\textsc{Alg}$ satisfies the relaxed constraints, and is an optimal policy
for the relaxed problem:

\begin{itemize}
    \item \textit{Matching constraint:} For each $u \in U$, the probability that $u$ is matched is the sum of the probabilities that each incident edge $(u,v)$ is selected. Since $\textsc{Alg}$ mirrors $\textsc{OptAlg}$, this expected total is at most 1.
    
    \item \textit{Query budget constraint:} For each vertex $u \in U$, the expected number of queries to edges in $\partial(u)$ equals that of $\textsc{OptAlg}$ and is at most $\ell_u$ by assumption.

    \item \textit{Objective value:} For any edge $e = (u,v)$ and action $a$, the probability it is queried with action $a$ depends only on the sequence-action pair chosen for vertex $v$. Since these distributions are identical between the algorithms, the expected reward of $\textsc{Alg}$ equals that of $\textsc{OptAlg}$.
    By definition, $\textsc{OptAlg}$ is an optimal policy for the relaxed problem, so its expected objective value is $\overline{\text{OPT}}(G)$. Therefore, the expected objective value of $\textsc{Alg}$ is also $\overline{\text{OPT}}(G)$.
\end{itemize}

Finally, we create a solution to \ref{LP:social_welfare}. Let $z_v(\bm{e},\bm{a}) := p_v(\bm{e},\bm{a})$ for all $v \in V$ and $(\bm{e},\bm{a}) \in \scr{C}_v$. We verify that this solution satisfies all LP constraints:

\begin{itemize}
\item Constraint \eqref{eqn:sw_distribution} requires $\sum_{(\bm{e},\bm{a}) \in \scr{C}_v} z_v(\bm{e},\bm{a}) \leq 1$ for each $v \in V$. Since $\sum_{(\bm{e},\bm{a}) \in \scr{C}_v} z_v(\bm{e},\bm{a}) = \sum_{(\bm{e},\bm{a}) \in \scr{C}_v} p_v(\bm{e},\bm{a}) \le 1$ (as $\textsc{Alg}$ must choose some policy for each vertex, or query none of its edges), this constraint is satisfied.

\item Constraint (\ref{eqn:sw_matching}) requires $\sum_{v \in V} \sum_{(\bm{e},\bm{a}) \in \scr{C}_v: e_i=(u,v)} q_{u,v}(a_i) \cdot \prod_{j<i} (1-q_{e_j}(a_j)) \cdot z_v(\bm{e},\bm{a}) \leq 1$ for all $u \in U$. The left side equals the expected
number of edges assigned to $u$ by $\textsc{Alg}$, which is at most $1$ in the relaxed problem.

\item Constraint (\ref{eqn:sw_offline_patience}) requires $\sum_{v \in V} \sum_{(\bm{e},\bm{a}) \in \scr{C}_v: e_i=(u,v)} \prod_{j<i} (1-q_{e_j}(a_j)) \cdot z_v(\bm{e},\bm{a}) \leq \ell_u$ for all $u \in U$. The left side equals the expected number of queries to edges incident to $u$ by $\textsc{Alg}$, which is at most $\ell_u$ by the relaxed patience constraint.
\end{itemize}
Given that our solution satisfies all constraints of \ref{LP:social_welfare} and its objective value equals $\overline{\text{OPT}}(G)$, we have $\overline{\text{OPT}}(G) \leq \text{LP}(G)$. Combining with our earlier inequality, we conclude that $\text{OPT}(G) \leq \text{LP}(G)$.
\qed

\section{Details of \Cref{sec:LP_round}} \label{app:LP_round}

\subsection{Proof of \Cref{lem:edge_variable}} \label{pf:lem:edge_variable}

First observe that for any edge $e \in E$,
$\sum_{a \in \scr{A}} Z_{e}(a) \le 1$, as \Cref{alg:social_welfare_relaxed} queries $e$ via
at most one action. Thus, since $\mb{E}[Z_{e}(a)] = \til{z}_e(a)$ by definition (see \eqref{eqn:edge_variable}),
$$
    \sum_{a \in \scr{A}} \til{z}_{e}(a) = \sum_{a \in \scr{A}} \mb{E}[Z_{e}(a)]  \le 1,
$$
which verifies the first property.

Now, if we fix $u \in U$, then we write the expected number of $v \in V$
with $\sum_{a \in \scr{A}} Z_{u,v}(a) =1$ as:
$$
\sum_{v \in V} \sum_{\ac \in \scr{A}} \til{z}_{u,v}(\ac) = \sum_{v \in V} \sum_{a \in \scr{A}} \mb{P}[Z_{u,v}(a) =1] = \sum_{v \in V} \sum_{\substack{i \in [\ell_v],  (\bm{e}, \bm{\ac}) \in \scr{C}_v: \\ e_i = (u,v)}} 
	 \prod_{j < i} (1 - q_{e_j}(\ac_j)) \cdot z_v( \bm{e}, \bm{\ac}).
$$
More, we can relate the expected number of edges $(u,v) \in \partial(u)$ which are included in $\scr{N}$:
$$
\sum_{v \in V} \sum_{\ac \in \scr{A}} q_{u,v}(\ac) \til{z}_{u,v}(\ac) = \sum_{v \in V} \mb{P}[(u,v) \in \scr{N}] = \sum_{v \in V} \sum_{\substack{i \in [\ell_v],  (\bm{e}, \bm{\ac}) \in \scr{C}_v: \\ e_i = (u,v)}} 
	q_{u,v}(\ac_i) \cdot \prod_{j < i} (1 - q_{e_j}(\ac_j)) \cdot z_v( \bm{e}, \bm{\ac}).
$$
By applying constraints \eqref{eqn:sw_offline_patience} and \eqref{eqn:sw_matching} of \ref{LP:social_welfare}, the remaining properties hold, and so the lemma follows.
\qed

\subsection{Proof of \Cref{lem:fixed_query_vertex}} \label{pf:lem:fixed_query_vertex}
As mentioned following \Cref{alg:social_welfare}, for any $u \in U$, $\psi_u$ is passed the input $$(\scr{A}, \partial(u), (q_{u,v}(a))_{v \in V, a \in \scr{A}}, (\til{z}_{u,v}(a))_{v\in V, a \in \scr{A}}),$$ which satisfies the required inequalities in \Cref{def:patience_rcrs}, due to \Cref{lem:edge_variable}.
We now verify that random variables used in line \ref{eqn:coupling}. are drawn via $(q_{u,v}(a))_{v \in V, a \in \scr{A}}$ and $(\til{z}_{u,v}(a))_{v\in V, a \in \scr{A}}$ in the required way as specified in \Cref{def:patience_rcrs}.

Clearly, the random variables $(Q_{u,v}(a))_{a \in \scr{A}, v \in V}$ are independent and have the correct marginals. On the other hand, notice that since
each $v \in V$ draws $(\bm{e}, \bm{\ac}) \in \scr{C}_v$ independently, $((\til{Z}_{u,v}(a))_{a  \in \scr{A}})_{v \in V}$ are independent across $v \in V$. Moreover, $\sum_{a \in \scr{A}} \til{Z}_{u,v}(a) \le 1$
for each $v \in V$, as $(u,v)$ is queried by at most once action.
We next argue that $(\til{Z}_{e}(a))_{e \in E, a \in \scr{A}}$ have the correct marginals,
i.e., for each $e \in E$ and $a \in \scr{A}$, $\mb{P}[\til{Z}_e(a) = 1] = \til{z}_e(a)$,
which will ensure \Cref{alg:social_welfare} is well-defined, and complete the proof.

Let us fix first fix $v \in V$, and define $\scr{P}_{v} \in \scr{C}_{v}$ to be the random policy drawn
by $v$ in \Cref{alg:social_welfare}. More, for each $e\in \partial(v)$ and $a \in \scr{A}$, set $F_{e}(a) :=1$ if and only if $e$ is queried via $a$ by \Cref{alg:social_welfare}.
Define the vector $\bm{F}^{(v)} := (F_{e}(a))_{a \in \scr{A}, e \in \partial(v)}$ for convenience.

Now, consider a fixed $u^* \in U$ and $a^* \in \scr{A}$, as well as $(\bm{e}, \bm{a}) \in \scr{C}_{v}$ (a possible instantiation of $\scr{P}_v$).
We refer to $(\bm{e}, \bm{a})$ as \textit{permissible} (for $(u^*,v)$ and $a^*$),
provided  $e_i = (u^*,v)$ and $a_i = a^*$ for some  $1 \le i \le |\bm{e}|$.
First observe that if $(\bm{e}, \bm{a})$ is \textit{not} permissible, then
\begin{equation}
        \mb{P}[\til{Z}_{u^*,v}(a^*) = 1 \mid \scr{P}_v = (\bm{e}, \bm{a})] = 0
\end{equation}
Thus, shall compute $\mb{P}[\til{Z}_{u^*,v}(a^*) = 1 \mid \scr{P}_v = (\bm{e}, \bm{a})]$ 
for each $(\bm{e}, \bm{a})$ which is permissible. Let us assume that $e_i = (u^*,v)$ and 
$a_i = a^*$ for some $1 \le i \le |\bm{e}|$. 
Finally, we consider an arbitrary vector $\bm{f}^{(v)} = (f_{e}(a))_{a \in \scr{A}, e \in \partial(v)}$, where $f_{e}(a) \in \{0,1\}$. 

Suppose now that we condition on $\scr{P}_v = (\bm{e}, \bm{a})$, and $\bm{F}^{(v)} = \bm{f}^{(v)}$. Observe then 
that $(u^*,v)$ is queried via $a^*$ if and only if $Q_{e_j}(a_j) = 0$ for each $j < i$ with $f_{e_j}(a_j) =1$ (i.e., all ``real'' queries fail) and $\til{Q}_{e_{j'}}(a_{j'}) = 0$ for all $j' < i$ with $f_{e_{j'}}(a_{j'}) =0$ (i.e., all ``fake'' queries fail).
Thus,  we can write $\mb{E}[\til{Z}_{u^*,v}(a^*) \mid \scr{P}_v = (\bm{e}, \bm{a}), \bm{F}^{(v)} = \bm{f}^{(v)}]$ as:
\begin{align}
    &\mb{E}\left [ \prod_{j < i: f_{e_j}(a_j) =1} (1-Q_{e_j}(a_j)) \prod_{j' < i: f_{e_{j'}}(a_{j'}) =0} (1-\til{Q}_{e_{j'}}(a_{j'}))  \mid  \scr{P}_v = (\bm{e}, \bm{a}), \bm{F}^{(v)} = \bm{f}^{(v)} \right]. \label{eqn:big_conditional_prob}          
\end{align}
% since conditional on  $\scr{P}_v = (\bm{e}, \bm{a})$ and  $\bm{F}^{(v)} = \bm{f}^{(v)}$, 
On the other hand, we can further simplify \eqref{eqn:big_conditional_prob} as:
\begin{align*}
    & \prod_{\substack{j < i: \\ f_{e_j}(a_j) =1}} \mb{E}[1-Q_{e_{j'}}(a_j) \mid \scr{P}_v = (\bm{e}, \bm{a}), \bm{F}^{(v)} = \bm{f}^{(v)}] \prod_{\substack{j' < i: \\ f_{e_{j'}}(a_{j'}) =0}} \mb{E}[1-\til{Q}_{e_{j'}}(a_{j'}) \mid \scr{P}_v = (\bm{e}, \bm{a}), \bm{F}^{(v)} = \bm{f}^{(v)}] \\
    &= \prod_{j < i}(1 - q_{e_j}(a_j)).
\end{align*}
Here we've used that \Cref{alg:social_welfare} decides whether or not to query any $e \in \partial(v)$ via $a$, prior to revealing
$Q_{e}(a)$ or $\til{Q}_e(a)$. The last equality uses this fact once more,
as well as that $Q_e(a)$ and $\til{Q}_e(a)$ are distributed identically. After averaging over the instantiations of $\bm{F}^{(v)}$, we get that
$$
\mb{P}[\til{Z}_{u^*,v}(a^*) = 1 \mid \scr{P}_v = (\bm{e}, \bm{a})] = \prod_{j < i}(1 - q_{e_j}(a_j)).
$$
Finally, after averaging over all permissible $(\bm{e}, \bm{a})$ for $(u^*,v)$ and $a^*$,
we get that
$$
\mb{P}[\til{Z}_{u^*,v}(a^*) = 1] =  \sum_{\substack{i \in [\ell_v], (\bm{e},\bm{\ac}) \in \scr{C}_v: \\ e_i = (u^*,v), \ac_i =\ac^*}} \prod_{j < i} (1 - q_{e_j}(\ac_j))\cdot z_{v}( \bm{e},\bm{\ac}) =\til{z}_{u^*,v}(a^*), 
$$
where the final equality follows by \eqref{eqn:edge_variable}.
\qed

\subsection{Proof of \Cref{lem:trs_to_prcrs}} \label{pf:lem:trs_to_prcrs}
Suppose that $(\scr{A}, n,\ell, (p_i(a))_{a \in \scr{A}, i \in [n]}, (x_i(a))_{a \in \scr{A}, i \in [n]})$ is a P-RCRS input for an arbitrary action set $\scr{A}$. We first construct a P-RCRS input 
$(n, \ell, (p_i)_{i=1}^n, (x_i)_{i=1}^n)$ on the trivial action set, where for each $i \in [n]$
\begin{equation}
    \text{$x_i := \sum_{a \in \scr{A}} x_i(a)$ and $p_i := \frac{\sum_{a \in \scr{A}} p_{i}(a) x_{i}(a)}{\sum_{a \in \scr{A}} x_i(a)}$.}
\end{equation}
First observe that $(n, \ell, (x_i)_{i=1}^n, (p_i)_{i=1}^n)$ is indeed a valid P-RCRS input,
as
\begin{equation}
    \text{$\sum_{i=1}^n x_i = \sum_{i=1}^n \sum_{a \in \scr{A}} x_i(a) \le \ell$, $\sum_{i=1}^n p_i x_i = \sum_{i=1}^n \sum_{a \in \scr{A}} p_{i}(a) x_{i}(a) \le 1$, and $x_i = \sum_{a \in \scr{A}} x_i(a) \le 1$},  
\end{equation}
where the inequalities follow since $(\scr{A}, n,\ell, (p_i(a))_{a \in \scr{A}, i \in [n]}, (x_i(a))_{a \in \scr{A}, i \in [n]})$ is a P-RCRS input. 

Now suppose
that $\psi$ is a P-RCRS which is $\alpha$-selectable on $(n, \ell, (p_i)_{i=1}^n, (x_i)_{i=1}^n)$. We shall
use it to design a P-RCRS $\psi'$ for $(\scr{A}, n,\ell, (p_i(a))_{a \in \scr{A}, i \in [n]}, (x_i(a))_{a \in \scr{A}, i \in [n]})$. Let us first assume
that $(P_i(a))_{a \in \scr{A}, i \in [n]}$ and $(X_i(a))_{a \in \scr{A}, i \in [n]}$
are the states and suggestions for the latter input, as draw in \Cref{def:patience_rcrs}.

We now define the random variables $(P_i,X_i)_{i=1}^n$ with the following properties:
\begin{enumerate}
    \item $(P_i,X_i)_{i \in [n]}$ are independent, \label{eqn:independent_coupling}
    \item  $X_i = \sum_{a \in \scr{A}} X_{i}(a)$ and $P_i X_i = \sum_{a \in \scr{A}} P_{i}(a) X_{i}(a)$ for all $i \in [n]$. \label{eqn:coupling_of_variables}
\end{enumerate}
We shall use these as the states and suggestions for the execution of $\psi$ on $(n, \ell, (x_i)_{i=1}^n, (p_i)_{i=1}^n)$.
To see that these properties are attainable, draw $\til{P}_i \sim \Ber(p_i)$ independently for each $i \in [n]$. Then, we define $P_i$ in the following way:
\begin{align*}
P_i= \begin{cases}
1 & \text{if $\sum_{a \in \scr{A}} P_{i}(a) X_{i}(a) =1$ and $\sum_{a \in \scr{A}} X_{i}(a) =1$} \\
0 & \text{if $\sum_{a \in \scr{A}} P_{i}(a) X_{i}(a) = 0$ and $\sum_{a \in \scr{A}} X_{i}(a) =1$}, \\
\til{P}_i & \text{otherwise.}
\end{cases}
\end{align*}
It is easy to see that this construction satisfies the above properties. More, 
if $\pi$ is the arrival order of $[n]$, then $(P_i,X_i)_{i \in [n]}$
can be constructed in an online fashion with respect to $\pi$ as a P-RCRS executes. Specifically,
when $i \in [n]$ arrives and $(X_i(a))_{a \in \scr{A}}$ are revealed, we set $X_i = \sum_{a \in \scr{A}} X_{i}(a)$.
More, if $\sum_{a \in \scr{A}} X_{i}(a) = 0$, then we can clearly drawn $\til{P}_i$ and set $P_i = \til{P}_i$.
On the other hand, if $\sum_{a \in \scr{A}} X_{i}(a) = 1$
and we query $i$ via some action $a$, then we can set $P_i \in \{0,1\}$, depending on if $Q_{i}(a) = 0$
or $Q_{i}(a) =1$. 
If we do \textit{not} query $i$ via any action, then since we will never output
$(i,a)$ for any action $a \in \scr{A}$, we can still reveal $(Q_{i}(a))_{a \in \scr{A}}$ and assign $P_i$ in this way.

Our P-RCRS $\psi'$ for $(\scr{A}, n,\ell, (p_i(a))_{a \in \scr{A}, i \in [n]}, (x_i(a))_{a \in \scr{A}, i \in [n]})$ executes $\psi$ on $(n, \ell, \bm{x}, \bm{p})$ in the same order $\pi$, using the coupled random variables $(P_i,X_i)_{i=1}^n$ 
to make its decisions. Specifically, when processing $i \in [n]$, it learns
$(X_i(a))_{a \in \scr{A}}$. If $X_i(a) =1$ for some $a \in \scr{A}$, then $X_i = 1$,
and it asks $\psi$ whether to make a query. If `yes', then it queries $i$ via $a$,
learns $P_i(a)$, and outputs $(i,a)$ if $P_i(a) =1$. If `no', then it still updates
$P_i$ as previously described, and moves to the next element with respect to $\pi$.
If $\sum_{a \in \scr{A}} X_i(a) = 0$, then $X_i = 0$, and it sets $P_i = \til{P}_i$,
and moves to the next element.

Now, $\psi$ is $\alpha$-selectable by assumption, so for each $i \in [n]$,
$$
\Pr[\text{$i$ is queried by $\psi$} \mid X_i =1] \ge \alpha
$$
On the other hand, the decisions of $\psi'$ are completely determined
by $(P_i,X_i)_{i=1}^n$ and do not actually depend on the particular instantiations of
$(P_i(a))_{a \in \scr{A}, i \in [n]}$ and $(X_i(a))_{a \in \scr{A}, i \in [n]}$. Thus, for each $\ac  \in \scr{A}$,
$$
\Pr[\text{$i$ is queried via $a$ by $\psi'$} \mid X_i(a) = 1] = \Pr[\text{$i$ is queried by $\psi$} \mid  X_i =1].
$$
By combining the above equations, we have show that $\psi'$ is $\alpha$-selectable on
the required input. The lemma is thus proven.
\qed

\subsection{Proof of \Cref{lem:p_values}} \label{pf:lem:p_values}
We provide a sketch of the argument indicating
how the main details proceed as in Lemma $2$ of \cite{brubach2024offline}.
We focus on outlining why we may assume without loss that $p_1 =1$ and $p_j \in \{0,1\}$ for all $j \in [n] \setminus \{1\}$ when lower bounding \eqref{eqn:original_output_probability}. We omit the explanation for why we may assume $\sum_{i \in [n]} x_i = \ell$, as this can be attained
at the very end, where we add additional elements $j$ with $p_j= 0$ and $x_j =1$, and apply a simple coupling argument.

Begin with the arbitrary input $\scr{I}$, and an arbitrary element of $i \in [n] \setminus \{1\}$ with $0 < p_i < 1$, which we denote by $i=n$ for simplicity. We construct a new input $\til{\scr{I}} = (n+1, \ell, (\til{x}_i)_{i=1}^{n+1}, (\til{p}_i)_{i=1}^{n+1})$ with $\til{p}_1 = \til{p}_n =1$,
and an additional element $n+1$ with $\til{p}_{n+1} =0$. Specifically, 
\begin{enumerate}
        \item $(\til{p}_1, \til{x}_1) = (1 , p_1 x_1)$.
    \item  $(\til{p}_i,\til{x}_i) = (p_i, x_i)$ for $i =2, \ldots ,n-1$.
    \item $(\til{p}_n, \til{x}_n) = (1, p_n x_n)$ and $(\til{p}_{n+1}, \til{x}_{n+1}) = (0, (1- p_n) x_n)$  
\end{enumerate}
Observe first that since $\scr{I}$ is a feasible input, $\sum_{i=1}^{n+1} \til{p}_i \til{x}_i = p_1 x_1 + p_n x_n + \sum_{i=2}^{n-1} p_i x_i \le 1$, and $\sum_{i=1}^{n+1} \til{x}_i = p_1 x_1 + \sum_{i=2}^{n-1} x_i + p_n x_n + (1- p_n) x_n \le \sum_{i=1}^n x_i \le \ell$,
so $\scr{\til{I}}$ is a feasible input.
Let us thus define the random variables $(\til{B}_i)_{i \in [n+1]}$, $(\til{X}_i)_{i \in [n+1]}$,
and $(\til{Y}_{i})_{i \in [n+1]}$ for $\scr{\til{I}}$ in analogy to $(B_i)_{i \in [n]}$, $(X_i)_{i \in [n]}$,
and $(Y_i)_{i \in [n]}$ for $\scr{I}$ (here $B_i \sim \Ber(b(p_i x_i))$ is the attenuation
bit for element $i$ used by \Cref{alg:P-RCRS}). Our goal is to show that for all $y \in [0,1]$,
\begin{equation} \label{eqn:p_values_main_goal}
    \mb{P}[\text{$1$ of $\scr{I}$ is queried} \mid Y_1 = y, X_1 =1] \ge 
    \mb{P}[\text{$1$ of $\scr{\til{I}}$ is queried} \mid \til{Y}_1 = y, \til{X}_1 =1]
\end{equation}
If we can establish \eqref{eqn:p_values_main_goal}, then we can iterate the argument a finite number of times until we're left with an input $\scr{\til{I}}$ with $\til{p}_1=1$,
and where each other element $i$ has $\til{p}_i \in \{0,1\}$. Thus, for the remainder of the proof,
we focus on explaining how \eqref{eqn:p_values_main_goal} relates to Lemma 2 of \cite{brubach2024offline}.

Denote $\safe_{1}(\scr{I})$ as the event that
$1$ is safe during the execution on input $\scr{I}$,
and recall that $\mb{P}[\text{$1$ of $\scr{I}$ is queried} \mid Y_1 = y, X_1 =1]   = b(p_1 x_1) \mb{P}[\safe_{1}(\scr{I}) \mid Y_1 =y]$. Now, $p_1 x_1 = \til{x}_1$, and so $b(p_1 x_1) = b(\til{x}_1)$. Thus, in order to establish \eqref{eqn:p_values_main_goal}, 
it suffices to prove that 
\begin{equation*} 
    \mb{P}[ \safe_{1}(\scr{I}) \mid Y_1 = y, X_1 =1] \ge \mb{P}[\safe_{1}(\til{\scr{I}}) \mid \til{Y}_1 = y, \til{X}_1 =1]. 
\end{equation*}
This is exactly the statement proven in Lemma $2$ of \cite{brubach2024offline},
and so we defer the remaining details of the argument.\qed

\subsection{Proof of Lemma \ref{lem:splitting_argument}} \label{pf:lem:splitting_argument}
Given the input $\scr{I} = (n,\ell, (x_i)_{i=}^n, (p_i)_{i=1}^n)$, we first construct a new input 
$\scr{\til{I}} =(n +1, \ell, (\til{x}_i)_{i=1}^{n+1}, (\til{p}_i)_{i=1}^{n+1})$, with an additional
element $n +1$, and which modifies $\scr{I}$ in the following way:
\begin{enumerate}
    \item $\til{p}_{n} = 1$ and $\til{x}_{n} := x_n/2$,
    \item $\til{p}_{n+1} = 1$ and $\til{x}_{n+1} = x_n/2$,
    \item $(\til{x}_i, \til{p}_i) = (x_i,p_i)$ for all $i \in [n -1]$.
\end{enumerate}
Clearly, $\scr{\til{I}}$ is a feasible input,
which satisfies the conditions of \Cref{lem:p_values}. I.e., $\til{p}_j \in \{0,1\}$ for all $j \in [n+1]$, $\til{p_1} = 1$ and $\sum_{j \in [n+1]} \til{x}_j = \ell$. For each $i \in [n+1]$,
we draw $\til{X}_i \sim \Ber(\til{x}_i)$ independently. 
Our goal is to argue that
\begin{equation} \label{eqn:splitting_goal}
     \mb{P}[\text{$1$ from $\scr{I}$ is queried} \mid X_1 =1] - \mb{P}[\text{$1$ from $\scr{\til{I}}$ is queried} \mid \til{X}_1 =1] \ge 0,
\end{equation}
as this will complete the proof of \Cref{lem:splitting_argument}.
Now, recall that \eqref{eqn:simplified_output_probability} ensures
\begin{equation} \label{eqn:simplified_output_probability_restated} 
    \mb{P}[\text{$1$ from $\scr{I}$ is queried} \mid X_1 =1] = b(x_1) \int_{0}^{1}\mb{P}[L_0 \le \ell -1 \mid Y_1 =y] \prod_{j \in N_1} (1 - y b(x_j) x_j) dy.
\end{equation}
Thus, if we define $\til{N}_0$, $\til{N}_1$, $\til{L}_0$, and $\til{Y}_1$ 
in the analogous way as done for $N_0, N_1$, $L_0$, and $Y_1$
of $\scr{I}$, then observe that $\til{x}_1 = x_1$, $\til{N_0} = N_0$, $\til{N}_1 = N_1 \cup \{n+1\}$. More, we may assume that $\til{X}_1 = X_1$ and $\til{Y}_1 = Y_1$,
and it is clear that $\mb{P}[\til{L}_{0} \le \ell -1 \mid Y_1 = y] = \mb{P}[L_{0} \le \ell -1 \mid \til{Y}_1 = y]$.
Thus, using the same argument for deriving \eqref{eqn:simplified_output_probability},
we can write
$\mb{P}[\text{$1$ from $\scr{\til{I}}$ is queried} \mid \til{X}_1 =1]$ as
\begin{align}
&b(\til{x}_1) \int_{0}^{1}\mb{P}[\til{L}_0 \le \ell -1 \mid \til{Y}_1 =y] \prod_{j \in \til{N}_1} (1 - y b(\til{x}_j) \til{x}_j)  dy. \notag \\
&= b(x_1) \int_{0}^{1} (1 - y x_n b(x_n/2)/2)^2 \prod_{j \in N_1 \setminus \{n\}} (1 - y b(x_j) x_j) \mb{P}[ L_0 \le \ell -1 \mid Y_1 =y] dy, \label{eqn:splitting_simplification}
\end{align}
where the second line follows from the aforementioned relations between $\scr{I}$
and $\scr{\til{I}}$. Now, applying \eqref{eqn:simplified_output_probability_restated} and
\eqref{eqn:splitting_simplification}, we can rewrite the left-hand side of \eqref{eqn:splitting_goal} as
\begin{equation}
    b(x_1) \int_{0}^{1} \left( (1 - y x_n b(x_n)) - (1 - y x_n b(x_n/2)/2)^2\right)                \prod_{j \in N_1 \setminus \{n\}} (1 - y b(x_j) x_j) \mb{P}[ L_0 \le \ell -1 \mid Y_1 =y] dy
\end{equation}
The final property of \Cref{property:first_order} ensures that there exists $y_{x_n} \in [0,1]$
such that the function $y \rightarrow (1 - y x_n b(x_n)) - (1 - y x_n b(x_n/2)/2)^2$
is non-negative for $y \in [0, y_{x_n}]$, and non-positive for $y \in (y_{x_n},1]$.
Thus, since both 
$y \rightarrow \prod_{j \in N_1 \setminus \{n\}} (1 - y b(x_j) x_j)$ 
and $y \rightarrow \mb{P}[ L_0 \le \ell -1 \mid Y_1 =y]$ are non-negative and non-increasing
functions of $y$, we may conclude that
\begin{align*}
& \int_{0}^{1} \left( (1 - y x_n b(x_n)) - (1 - y x_n b(x_n/2)/2)^2\right)                \mb{P}[ L_0 \le \ell -1 \mid Y_1 =y] \prod_{j \in N_1 \setminus \{n\}} (1 - y b(x_j) x_j) dy \\
& \ge \mb{P}[ L_0 \le \ell -1 \mid Y_1 =y_{x_n}] \prod_{j \in N_1 \setminus \{n\}} (1 - y_{x_n} b(x_j) x_j) \int_{0}^{1} \left( (1 - y x_n b(x_n)) - (1 - y x_n b(x_n/2)/2)^2\right) dy.
\end{align*}
This lower bound is non-negative, as$\int_{0}^{1} \left( (1 - y x_n b(x_n)) - (1 - y x_n b(x_n/2)/2)^2\right) dy \ge 0$, by the final property of \Cref{property:first_order}. Thus, since $b(x_1) \ge 0$, \eqref{eqn:splitting_goal} holds, and so the proof is complete.
\qed

\subsection{Proof of \Cref{property:first_order}} \label{pf:property_first_order}
The first property follows immediately from the definition of $b$.
In order to see the the final two property, consider the function
$$f(y):= (1 - y z b(z)) - (1 - y z b(z/2)/2)^2 = - y^2 z^2 b(z/2)^2/4 + y (z b(z/2) - z b(z)).$$
Now,
\begin{equation} \label{eqn:integral_check}
\int_{0}^1 f(y) dy =   (z b(z/2) - z b(z))/2 - z^2 b(z/2)^2/12,
\end{equation}
and for our choice of $b$, we verify that \eqref{eqn:integral_check} is non-negative
for all $z \in [0,1]$.
To see the final property, clearly $f$ is a quadratic function in $y$, and it is easy to see that its roots are $0$ and $4 (b(z/2) - b(z))/z b(z/2)^2$.
Now, since $b$ is decreasing, $b(z/2) - b(z) \ge 0$, and so the latter root
is non-negative. Finally, the leading term is $- z^2 b(z/2)^2/4 $, which is negative.
Thus, $f(y) \ge 0$ for $y \in [0, 4 (b(z/2) - b(z))/z b(z/2)^2]$,
and $f(y) \le 0$ for $y \in [0, 4 (b(z/2) - b(z))/z b(z/2)^2]$.
\qed

\section{Details of \Cref{sec:solve_LP}}

\subsection{Proof of \Cref{prop:reduction_to_dual_approximation}} \label{pf:prop:reduction_to_dual_approximation}
Given $\eps > 0$, consider the \textit{relaxed dual}:

\begin{align*}\label{LP:sw_dual_relaxed}
	\tag{r-D-LP-c}
	&\text{minimize} 
	& \sum_{u \in U} \alpha_{u} + \sum_{u \in U} \ell_u \gamma_u + \sum_{v \in V} \beta_{v}  \\
	&\text{s.t.} 
	& \beta_{v} \ge (1-\eps) \srew_{\bm{r},\bm{q}}(\bm{e}, \bm{\ac}) 
	- \sum_{i=1}^{|\bm{e}|} ( q_{e_i}(\ac_i) \alpha_{u_i} + \gamma_{u_i}) 
	\cdot \prod_{j < i} (1 - q_{e_j}(\ac_j))  && \forall v \in V, (\bm{e}, \bm{\ac}) \in \scr{C}_v:  \\
	&&&& \begin{aligned}
		% &\bm{e} = (e_1, \ldots , e_k), &
		% & \bm{a} = (a_1, \ldots , a_k), \\
		& \text{$e_i$} = (u_i,v) & & \forall i \in [|\bm{e}|]
	\end{aligned} \\
	&& \alpha_{u}, \gamma_u \ge 0 && \forall u \in U\\
	&& \beta_{v} \ge 0 && \forall v \in V
\end{align*}

(Here the only difference from \ref{LP:sw_dual} is that in the constraint
corresponding to $v \in V$ and $(\bm{e}, \bm{\ac})$,
the term $\srew_{\bm{r}, \bm{q}}(\bm{e}, \bm{\ac})$ is multiplied by $(1-\eps)$).
Now, a \textit{strong approximate separation oracle} is presented an arbitrary selection of dual variables
$((\alpha_u)_{u \in U}, (\gamma_u)_{u \in U}, (\beta_v)_{v \in V})$, and does one
of the following:
\begin{enumerate}
    \item Asserts that $((\alpha_u)_{u \in U}, (\gamma_u)_{u \in U}, (\beta_v)_{v \in V})$ is a feasible solution to \ref{LP:sw_dual_relaxed}. \label{item:approx_feasible}
    \item Outputs a constraint of \ref{LP:sw_dual} violated by $((\alpha_u)_{u \in U}, (\gamma_u)_{u \in U}, (\beta_v)_{v \in V})$. \label{item:violated}
\end{enumerate}
Observe that if $\eps =0$, then this is the usual definition of a strong separation oracle defined in Chapter 4 of \cite{jens2012}. It is shown in Theorem 5.1 of \cite{Jansen03} that if one can provide
a strong approximate separation oracle, then the ellipsoid algorithm can be used to efficiently
compute a solution to \ref{LP:social_welfare} with value at least $(1- \eps) \LPOPT(G)$.
More, its runtime will be its usual runtime, multiplied by the runtime
of the strong approximate separation oracle. 

For the remainder of the section, 
we explain how to establish
\ref{item:approx_feasible}. or \ref{item:violated}. for $((\alpha_u)_{u \in U}, (\gamma_u)_{u \in U}, (\beta_v)_{v \in V})$ under the assumptions of \Cref{prop:reduction_to_dual_approximation}.

Observe first that we can efficiently check whether $((\alpha_u)_{u \in U}, (\gamma_u)_{u \in U}, (\beta_v)_{v \in V})$ are non-negative (and indicate if this is false as required by \ref{item:violated}.). Thus, we assume that $((\alpha_u)_{u \in U}, (\gamma_u)_{u \in U}, (\beta_v)_{v \in V})$ 
are non-negative in what follows. Now, recall the edge rewards $\hat{\bm{r}} = (r_{u,v})_{u \in U, v \in V}$ from \eqref{eqn:new_reward_function_v},
whose definition ensures that for all $v \in V$ and $(\bm{e}, \bm{\ac}) \in \scr{C}_v$ where $e_{i}=(u_{i},v)$ for $i=1, \ldots ,|\bm{e}|$,
\begin{equation} \label{eqn:new_reward_restatement}
   \srew_{\hat{\bm{r}}, \bm{q}}(\bm{e}, \bm{\ac}) = \srew_{\bm{r}, \bm{q}}(\bm{e}, \bm{\ac}) - \sum_{i=1}^{|\bm{e}|} ( q_{e_i}(\ac_i) \alpha_{u_i} + \gamma_{u_i}) \cdot \prod_{j < i} (1 - q_{e_j}(\ac_j)).
\end{equation}
As in the assumption of \Cref{prop:reduction_to_dual_approximation}, suppose that for each $v \in V$, we can efficiently compute $(\bm{e}^{(v)}, \bm{\ac}^{(v)}) \in \scr{C}_v$
such that
\begin{equation} \label{eqn:approximation_oracle_appendix}
    \srew_{\bm{\hat{r}}, \bm{q}}(\bm{e}^{(v)}, \bm{\ac}^{(v)}) \ge (1- \eps) \max_{(\bm{e}, \bm{\ac}) \in \scr{C}_v} \srew_{\bm{\hat{r}}, \bm{q}}(\bm{e}, \bm{\ac}). 
\end{equation}
There are two cases to consider. First, suppose that $\srew_{\bm{\hat{r}}, \bm{q}}(\bm{e}^{(v)}, \bm{\ac}^{(v)}) > \beta_v$
for some $v \in V$. In this case, by \eqref{eqn:new_reward_restatement}, the constraint of \ref{LP:sw_dual}
corresponding to $v$ and $(\bm{e}^{(v)}, \bm{\ac}^{(v)}) \in \scr{C}_v$ is violated. Thus,
our strong approximate separation oracle may output this constraint and satisfy \ref{item:violated}.
On the other hand, suppose that $\srew_{\bm{\hat{r}}, \bm{q}}(\bm{e}^{(v)}, \bm{\ac}^{(v)}) \le \beta_v$
for all $v \in V$. Then, by \eqref{eqn:approximation_oracle_appendix}, we know that
for all $v \in V$ and $(\bm{e}, \bm{\ac}) \in \scr{C}_v$, where $e_{i}=(u_{i},v)$ for $i=1, \ldots ,|\bm{e}|$:
\begin{align}
    \beta_v &\ge (1- \eps) \srew_{\bm{\hat{r}}, \bm{q}}(\bm{e}, \bm{\ac}) \notag \\
            & = (1 - \eps)(\srew_{\bm{r}, \bm{q}}(\bm{e}, \bm{\ac}) - \sum_{i=1}^{|\bm{e}|} ( q_{e_i}(\ac_i) \alpha_{u_i} + \gamma_{u_i}) \cdot \prod_{j < i} (1 - q_{e_j}(\ac_j))) \notag \\
            &\ge (1 - \eps) \srew_{\bm{r}, \bm{q}}(\bm{e}, \bm{\ac}) -  \sum_{i=1}^{|\bm{e}|} ( q_{e_i}(\ac_i) \alpha_{u_i} + \gamma_{u_i}) \cdot \prod_{j < i} (1 - q_{e_j}(\ac_j))). \label{eqn:approximate_feasibility_app}
\end{align}
Here the first equality follows by \eqref{eqn:new_reward_restatement}, and the final inequality follows
since $((\alpha_u)_{u \in U}, (\gamma_u)_{u \in U}, (\beta_v)_{v \in V})$ are non-negative.
As such, by \eqref{eqn:approximate_feasibility_app}, in this case our approximate separation oracle may assert
that $((\alpha_u)_{u \in U}, (\gamma_u)_{u \in U}, (\beta_v)_{v \in V})$ is a feasible solution to \ref{LP:sw_dual_relaxed},
and thus satisfy \ref{item:approx_feasible}.
\qed

\subsection{Proof of \Cref{lem:feasible_IP_solution}} \label{pf:lem:feasible_IP_solution}
    To establish the feasibility of IP-TH, we construct an assignment that satisfies all the constraints. For each bucket $B_j$ that is a jump bucket, there is at most one edge $e^*_t$ assigned to it. Set $\xi^*_{j{e^*_t}} = 1$ and $\xi^*_{ji} = 0$ for all other $i$. Then for each stable bucket $B_j$, if it consists of edges $e^*_t,..., {e^*_{t+h}}$, we set $\xi^*_{je^*_t},...,\xi^*_{je^*_{t+h}} = 1$ and $\xi^*_{ji} = 0$ for all other $i$. These are all binary (\textit{constraint \eqref{eq:binary}}), each job is assigned to at most one machine (\textit{constraint \eqref{eq:assignment}}, each jump machine gets at most $1$ job assigned to it (\textit{constraint \eqref{eq:capacity}}), and at most $\ell$ jobs are assigned to machines (\textit{constraint{\eqref{eq:budget}}}). For \textit{constraint \eqref{eq:load}} we consider 2 cases:

    \textit{Case 1}: Jump Buckets: for a jump bucket $B_j$, let the edge assigned to it to be $e^*_p$. Let $a_{opt}$ be the optimal action for $e^*_p$ that maximizes the expected future value. By definition, we have: 
    $$ R^*_{e^*_p} - R^*_{e^*_{p+1}} = r_{e^*_p}(a_{opt}) \cdot q_{e^*_p}(a_{opt}) - R^*_{e^*_{p+1}} \cdot q_{e^*_p}(a_{opt}) $$
    $$ = (r_{e^*_p}(a_{opt}) - R^*_{e^*_{p+1}}) \cdot q_{e^*_p}(a_{opt}) $$
    $$ =(r_{e^*_p}(a_{opt}) - \text{BaseVal}(B_j)) \cdot q_{e^*_p}(a_{opt})  = \max_{a \in \scr{A}} \left\{(r_{e^*_p}(a) - \text{BaseVal}(B_j)) \cdot q_{e^*_p}(a)]\right\}$$
    Then we have:
    $$\sum_{j=1}^n  \max_{a\in \scr{A}} \left\{(r_{e^*_p}(a) - \text{BaseGuess}(B_i)) \cdot q_j(a)]\right\}\xi_{ij}  $$
    $$\ge \max_{a\in \scr{A}} \left\{(r_{e^*_p}(a) - \text{BaseGuess}(B_j))  \cdot q_{e_p}(a)]\right\} \geq \max_{a\in \scr{A}} \left\{(r_{e^*_p}(a) - \text{BaseVal}(B_j)) \cdot q_{e^*_p}(a)]\right\} $$
    
    $$ = R^*_{e^*_p} - R^*_{e^*_{p+1}}  \geq \text{DeltaVal}(B_j) \geq \text{DeltaGuess}(B_j)  $$
    \textit{Case 2: Stable Buckets}: Let $B_i$ be a stable bucket containing consecutive positions $\{p, p+1, \ldots, p+y\}$ with corresponding buyers $\{e^*_p, \ldots, {e^*_{p+y}}\}$. For each position $z \in \{p, \ldots, p+y\}$, denote the optimal action by $a^*_z$.

The aggregate load contribution is:
\begin{align*}
&\sum_{z=p}^{p+y} \max_{a\in \scr{A}} \{(r_{e^*_p}(a) - \text{BaseGuess}(B_i)) \cdot q_{e^*_z}(a)\} \\
&\geq \sum_{z=p}^{p+y} (r_{e^*_p}({a}^*_z) - \text{BaseGuess}(B_i)) \cdot q_{e^*_z}({a}^*_z)
\end{align*}
For each position $z$ in the stable bucket, the optimal action $a^*_z$ satisfies:
$$R^*_z - R^*_{z+1} = (r_{e^*_p}({a}^*_z) - R^*_{z+1}) \cdot q_{e^*_z}({a}^*_z)$$
Since all positions $q \in \{p, \ldots, p+y\}$ belong to the same stable bucket $B_i$, we have $R^*_{q+1} \geq \text{BaseVal}(B_i)$ for all such positions. 
Therefore, for each position $z$ in the stable bucket:
\begin{align*}
R^*_z - R^*_{z+1} &= (r_{e^*_p}({a}^*_z) - R^*_{z+1}) \cdot q_{e^*_z}({a}^*_z) \\
&\leq (r_{e^*_p}({a}^*_z) - \text{BaseVal}(B_i)) \cdot q_{e^*_z}({a}^*_z) \\
&\leq (r_{e^*_p}({a}^*_z) - \text{BaseGuess}(B_i)) \cdot q_{e^*_z}({a}^*_z)
\end{align*}
Summing over all positions in the stable bucket:

\begin{align*}
&\sum_{z=p}^{p+y} (r_{e^*_p}({a}^*_z) - \text{BaseGuess}(B_i)) \cdot q_{e^*_z}({a}^*_z) \\
&\geq \sum_{z=p}^{p+y} (R^*_z - R^*_{z+1}) \\
&= R^*_p - R^*_{p+y+1} \\
&= \text{DeltaVal}(B_i) \\
&\geq \text{DeltaGuess}(B_i)
\end{align*}
This completes the proof that \textit{constraint \eqref{eq:load}} is satisfied for all stable buckets, and thus IP-TH admits a feasible solution.
\qed

\subsection{Proof of \Cref{lem:ahead-schedule}} \label{pf:lem:ahead-schedule}
By the recursive definition:
$$\tilde{R}_{i_{\min}} = \max_{a\in \scr{A}} \left\{ r_{e_{i_{\min}}}(a) \cdot q_{e_{i_{\min}}}(a) + \tilde{R}_{i_{\min}+1} \cdot (1 - q_{e_{i_{\min}}}(a)) \right\}$$
Since we are ahead of schedule at position $i_{\min}$, we have $\tilde{R}_{i_{\min}+1} \geq \text{BaseGuess}(B^{-1}_{e_{i_{\min}}})$. Therefore:
\begin{align*}
\tilde{R}_{i_{\min}} &\geq \max_{a\in \scr{A}} \left\{ r_{e_{i_{\min}}}(a) \cdot q_{e_{i_{\min}}}(a) + \text{BaseGuess}(B^{-1}_{e_{i_{\min}}}) \cdot (1 - q_{e_{i_{\min}}}(a)) \right\}\\
&= \text{BaseGuess}(B^{-1}_{e_{i_{\min}}}) + \max_{a\in \scr{A}} \left\{ (r_{e_{i_{\min}}}(a) - \text{BaseGuess}(B^{-1}_{e_{i_{\min}}})) \cdot q_{e_{i_{\min}}}(a) \right\}
\end{align*}
The term $\max_{a\in \scr{A}} \left\{ (r_{e_{i_{\min}}}(a) - \text{BaseGuess}(B^{-1}_{e_{i_{\min}}})) \cdot q_{e_{i_{\min}}}(a) \right\}$ is exactly the load contribution $c_{i_{\min}, e_{i_{\min}}}$ of buyer $e_{i_{\min}}$ to bucket $B^{-1}_{e_{i_{\min}}}$ in our Santa Claus formulation.

When $B^{-1}_{e_{i_{\min}}}$ is a jump bucket, our capacity constraints ensure that edge $e_{i_{\min}}$ is the only one assigned, so by Corollary \ref{corr:eptas}:
$$c_{i_{\min}, e_{i_{\min}}} \geq (1-\epsilon) \cdot \text{DeltaGuess}(B^{-1}_{e_{i_{\min}}})$$
When $B^{-1}_{e_{i_{\min}}}$ is a stable bucket, we have $\text{DeltaGuess}(B^{-1}_{e_{i_{\min}}}) \leq \epsilon \cdot \text{OPT}$ since this bucket does not contain any reward jump of size $\geq \epsilon \cdot \text{OPT}$. In either case:
\begin{align*}
\tilde{R}_{i_{\min}} &\geq \text{BaseGuess}(B^{-1}_{e_{i_{\min}}}) + (1-\epsilon) \cdot \text{DeltaGuess}(B^{-1}_{e_{i_{\min}}}) - \epsilon \cdot \text{OPT}\\
&\geq \text{BaseVal}(B^{-1}_{e_{i_{\min}}}) + (1-\epsilon) \cdot \text{DeltaGuess}(B^{-1}_{e_{i_{\min}}}) - 2\epsilon \cdot \text{OPT}
\end{align*}
\qed

\subsection{Proof of \Cref{lem:behind-schedule}} \label{pf:lem:behind-schedule}
For each $i < i_{\min}$, by the recursive definition:
$$\tilde{R}_i - \tilde{R}_{i+1} = \max_{a\in \scr{A}} \left\{ (r_{e_{i}}(a) - \tilde{R}_{i+1}) \cdot q_{e_{i}}(a) \right\}$$

Since we are behind schedule at position $i < i_{\min}$, we have $\tilde{R}_{i+1} < \text{BaseGuess}(B^{-1}_{e_{i}})$, so:
$$\tilde{R}_i - \tilde{R}_{i+1} \geq \max_{a\in \scr{A}} \left\{ (r_{e_{i}}(a) - \text{BaseGuess}(B^{-1}_{e_{i}})) \cdot q_{e_{i}}(a) \right\}$$

Summing over all edges assigned to bucket $j$:
\begin{align*}
\sum_{i: B^{-1}_{e_{i}} = j} (\tilde{R}_i - \tilde{R}_{i+1}) &\geq \sum_{i: B^{-1}_{e_{i}} = j} \max_{a\in \scr{A}} \left\{ (r_{e_{i}}(a) - \text{BaseGuess}(B_j)) \cdot q_{e_{i}}(a) \right\}\\
&= \sum_{y: \xi_{jy} = 1} c_{jy}\\
&\geq (1-\epsilon) \cdot \text{DeltaGuess}(B_j)
\end{align*}
Thus:
\begin{align*}
\sum_{i < i_{\min}} (\tilde{R}_i - \tilde{R}_{i+1})
& \geq \sum_{j > B^{-1}_{e_{i_{\min}}}}  \sum_{i: B^{-1}_{e_{i}} = j} (\tilde{R}_i - \tilde{R}_{i+1})\\
&\geq (1-\epsilon) \cdot \sum_{j > B^{-1}_{e_{i_{\min}}}} \text{DeltaGuess}(B_j)\\
&\geq (1-\epsilon) \cdot \sum_{j > B^{-1}_{e_{i_{\min}}}} (\text{BaseVal}(B_{j+1}) - \text{BaseVal}(B_j)) - \epsilon^2 \cdot \mathcal{E} \cdot (2K+1)\\
&\geq (1-\epsilon) \cdot \sum_{j > B^{-1}_{e_{i_{\min}}}} (\text{BaseVal}(B_{j+1}) - \text{BaseVal}(B_j)) - 3\epsilon \cdot \text{OPT},
\end{align*}
where we've used that $K \leq 1/\epsilon$ and $\mathcal{E} \leq \text{OPT}$.

\end{document}